\documentclass[twocolumn, pra,superscriptaddress, amsmath,amssymb, aps,prr]{revtex4-2}

\usepackage{amsmath,amssymb,amsfonts,relsize,mathtools,times}%

\usepackage{amsthm}%
\usepackage{mathrsfs}%
\usepackage[title]{appendix}%
\usepackage{xcolor}%
\usepackage{textcomp}%
\usepackage{subfigure}

\usepackage{booktabs}%
\usepackage[linesnumbered,ruled,vlined]{algorithm2e}
\usepackage{algpseudocode}%
\usepackage{listings}%
\usepackage{bm}
\usepackage{bbold}
\usepackage{color}
\usepackage[english]{babel}
\usepackage{lipsum}
\usepackage[normalem]{ulem}
\usepackage{grffile}
\usepackage{braket}
\usepackage{appendix}
\usepackage[colorlinks, linkcolor=blue, anchorcolor=blue, citecolor=blue, urlcolor=blue]{hyperref}

\newcommand{\bfZ}{\mathbb{Z}}

\newcommand{\I}{\mathrm{i}}

\newcommand{\poly}{\mathrm{poly}}

\newcommand{\dyadic}[1]{\mathbf{#1}}

\newcommand{\sket}[1]{\left.\left|{#1}\right>\!\right>}
\newcommand{\sbra}[1]{\left<\!\left<{#1}\right|\right.}

\newcommand{\eqmoment}{\mathbb{E}_{\ket{\phi^{{\rm req}}_\dyadic{A}}\sim \mathcal{D}\left(\mathcal{S}^{{\rm req}}_n\right)}}
\newcommand{\eqmomentop}{\mathbb{E}_{\ket{\phi^{{\rm req}}_\dyadic{A}}\sim \mathcal{D}\left(\mathcal{S}^{{\rm req}}_n\right)}\ket{\phi^{{\rm req}}_\dyadic{A}}\bra{\phi^{{\rm req}}_\dyadic{A}}}
\newcommand{\momentop}{\mathbb{E}_{\ket{\phi^{{\rm eq}}_\dyadic{A}}\sim\mathcal{D}\left(\mathcal{S}^{{\rm eq}}_n\right)}\ket{\phi^{{\rm eq}}_\dyadic{A}}\bra{\phi^{{\rm eq}}_\dyadic{A}}}
\newcommand{\momentopp}{\mathbb{E}_{\ket{\phi^{{\rm (r)eq}}_\dyadic{A}}\sim\mathcal{D}\left(\mathcal{S}^{{\rm (r)eq}}_n\right)}\ket{\phi^{{\rm (r)eq}}_\dyadic{A}}\bra{\phi^{{\rm (r)eq}}_\dyadic{A}}}
\newcommand{\moment}{\mathbb{E}_{\ket{\phi^{{\rm eq}}_\dyadic{A}}\sim \mathcal{D}\left(\mathcal{S}^{{\rm eq}}_n\right)}}

\newcommand{\tr}[1]{\mathrm{tr}\!\left(#1\right)}
\newcommand{\rank}{{\rm rank}}

\newtheorem{lemma}{Lemma}
\newtheorem{theorem}{Theorem}
\newtheorem{proposition}[theorem]{Proposition}%

\begin{document}

\preprint{APS/123-QED}

\title{Resource-efficient shadow tomography using equatorial stabilizer measurements}


\author{Guedong Park}
	\affiliation{IRC NextQuantum, Department of Physics and Astronomy, Seoul National University, Seoul, 08826, Republic of Korea}
	\author{Yong Siah Teo}
        \email{ys\_teo@snu.ac.kr}
	\affiliation{IRC NextQuantum, Department of Physics and Astronomy, Seoul National University, Seoul, 08826, Republic of Korea}
	
	
	\author{Hyunseok Jeong}
        \email{h.jeong37@gmail.com}
	\affiliation{IRC NextQuantum, Department of Physics and Astronomy, Seoul National University, Seoul, 08826, Republic of Korea}
	


\begin{abstract}

We propose a resource-efficient shadow-tomography scheme using equatorial-stabilizer measurements generated from subsets of Clifford unitaries. For $n$-qubit systems, equatorial-stabilizer-based shadow-tomography schemes can estimate $M$ observables (up to an additive error $\varepsilon$) using $\mathcal{O}(\log(M),\mathrm{poly}(n),1/\varepsilon^2)$ sampling copies for a large class of observables, including those with traceless parts possessing polynomially-bounded Frobenius norms. For arbitrary quantum-state observables with a constant Frobenius norm, sampling complexity becomes $n$-independent. Our scheme only requires an $n$-depth controlled-$Z$~(CZ) circuit [$\mathcal{O}(n^2)$ CZ~gates] and Pauli measurements per sampling copy. Alternatively, our scheme is realizable with $2n$-depth circuits comprising $n^2$ nearest-neighboring CNOT gates, exhibiting a smaller maximal gate count relative to previously-known randomized-Clifford-based proposals. We numerically confirm our theoretically-derived shadow-tomographic sampling complexities with random pure states and multiqubit graph states. Finally, we demonstrate that equatorial-stabilizer-based shadow~tomography is more noise-tolerant than randomized-Clifford-based schemes in terms of fidelity estimation for Greenberger--Horne--Zeilinger (GHZ) state and W~state.

\end{abstract}

\maketitle

\section{Introduction}
	Quantum algorithms, which are algorithms based on quantum-mechanical principles, have been shown to outperform classical algorithms for many computational tasks \cite{shor1999, deutsch1992,gavinsky2011}. For these purposes, certifying the quality of the prepared (noisy) quantum state is an important prerequisite. Quantum tomography \cite{adriano2002,lnp:2004uq,donnel2016,nielsen2021} is a general method to reconstruct a typically unknown noisy $n$-qubit quantum state $\rho$ (or a channel). This, however, requires exponentially many copies of~$\rho$ in~$n$ for an accurate reconstruction.
	
	On the other hand, if we are to estimate certain classes of physical properties of unknown states, \emph{shadow tomography} \cite{aronson2018} enables us to do so with a polynomially-large number of samples. The well-known randomized Pauli tomography, or ``six-Paulis'' tomography~\cite{huang2020,acharya2021}, is useful for estimating the expectation values of Hamiltonians with low-weighted Pauli observables, learning quantum channels~\cite{huang2023} and measuring non-stabilizerness (or magic)~\cite{leone2022,oliviero2022}. This procedure uses a simple single-layer circuit structure, and is therefore ineffective in estimating highly-entangled properties of large quantum systems. The second protocol, randomized-Clifford shadow tomography~\cite{huang2020}, is better for estimating the fidelity between the unknown input and entangled target states for large qubit number~$n$. This advantage originates from the fact that Clifford operations are unitary $3$-designs \cite{Webb:2016Clifford,Zhu:2017multiqubit,kueng2015,Pucha2017}. However, for large~$n$, this protocol utilizes at most $\frac{7n^2}{2}$ neighboring CNOT~gates \cite{maslov2023,bravyi2021}. Therefore, shadow tomography with arbitrary Clifford gates could lead to noisy measurement results in real experiments as the large number of noisy gates accumulates physical noise. 
	
	To cope with such problems, recent works showed that target observables, on average over suitable ensembles of input states, possess good shadow norms even with neighboring gates of shallow depths~\cite{bertoni2024,akhtar2023,hu2023}.  However, it is uncertain whether arbitrary state inputs reside in the average case, and such shallow shadow~\cite{bertoni2024} preserves the shadow norm scaling and thus the efficiency of input-state copy number. It is known that generalized measurements~\cite{acharya2021,grier2022,nguyen2022} can reduce the gate count or sampling complexity. However, noise-tolerant implementation of such non-Clifford measurements is believed to be challenging~\cite{bravyi2005,wang2023,bravyi2012}. There exist alternatives with reduced Clifford-gate resources \cite{flammia2011,silva2011}. Nevertheless, applicable input states and target observables that lead to $\poly(n)$-shadow norms are restrictive. Hence, a complete analysis of an optimal tomography scheme for broad classes of observables remains an open problem.
	
	In this work, we introduce resource-efficient schemes for shadow tomography, catered to observable-expectation-value estimations, that offer a lower gate count and circuit depth while still being applicable to a large observable class. The key point is that it is not necessary to utilize the entire set of stabilizer states \cite{aaronson2004,bravyi2016i} comprising all Clifford-rotated bases. Instead, our proposed schemes employ a smaller stabilizer subset, namely the \emph{equatorial stabilizer positive operator-valued measure}~(ESPOVM) \cite{bravyi2019} together with the computational basis. Additionally, it is sufficient to perform shadow tomography with an even smaller subset, namely the \emph{real ESPOVM} or~RESPOVM, when the target observable is real with respect to the computational basis. Our scheme only involves a uniform sampling of commuting CZ gates as the only two-qubit gate resources. 
	
	Next, we theoretically prove shadow-tomographic capabilities for (R)ESPOVM schemes by deriving upper bounds of sampling-copy complexity for the observable-expectation-value estimation, which is directly related to \emph{shadow norm}~\cite{huang2020,bertoni2024}. Specifically, we note that a sampling complexity of $\mathcal{O}(\log(M),\mathrm{poly}(n),1/\varepsilon^2)$ suffices to simultaneously estimate (up to an additive error~$\varepsilon$) $M$ complex observables belonging to a broad class, which \emph{sufficiently} include observables possessing polynomially-bounded Frobenius norms for their traceless parts. For arbitrary \emph{quantum-state} observables (\emph{fidelity estimation}), the sampling complexity is reduced to simply $\mathcal{O}(\log(M),1/\varepsilon^2)$, which is constant in $n$ and on par with randomized-Clifford tomography~\cite{huang2020}. These results are promising when compared with a recent results~\cite{zhang2024,wang2024} which utilizes smaller subsets of CZ-circuit-based measurements consisting of mutually-unbiased bases~(MUB)~\cite{Durt:2010mutually}, for estimation of observable expectation values, but has an exponentially large bound of sampling copies to estimate arbitrary state observables. 
	
    Furthermore, we shall show that our scheme can be implemented by a more simplified circuit structure, and attain improved noise tolerance. These results are derived from the fact that it involves a uniform sampling of commuting CZ gates as the only two-qubit gate resources. This enables our scheme to require only $n$-depth of CZ gates or $2n$-depth, hence at most $n^2<\frac{7n^2}{2}$ nearest neighboring~(NN) CNOT~gates~\cite{maslov2018}. Such a scaling factor optimization renders our scheme to become more robust to gate noise. Ref.~\cite{bertoni2024} analyzed the trade-off among the sampling, computational complexities, and required circuit-depth scaling, from constant- to linear-depth. On the other hand, our work aims to reduce the required depth while preserving the efficiency of sampling and computational complexity for arbitrary observables. Recent work~\cite{schuster2025} showed that small-biased shadow tomography is possible even with log-scaled depth. However, we observed that our circuit depth is still much lower in the intermediate qubit number regime, where many practical cases of fidelity estimation reside. 
    
    In addition, we show that the estimator of (R)ESPOVM shadow tomography itself can preserve its output value from the various measurement bit-flip errors. This observation suggests another crucial structural consideration for Clifford circuits in terms of noise tolerance, aside from reducing the gate complexity. If neighboring gates have negligible noise, the state-injection-based implementation of (R)ESPOVM measurements to tailor the arbitrary noise into the dephasing noise reinforces such maintenance to the noise in long-ranged CZ gates. We provide its numerical evidence in the fidelity estimation results of Greenberger--Horne--Zeilinger (GHZ) state~\cite{greenberger2007} and W states~\cite{dur2000}.

\section{Preliminaries}
	
	\subsection{Clifford group and equatorial stabilizer POVM}\label{sec:Clifford_group_and_equatorial_stabilizer}
	
	Before presenting our main results, we define and explain several concepts and terminologies that shall be used throughout the article. We consider a source that allows one to prepare multiple copies of an $n$-qubit quantum state~$\rho$. The corresponding observable $O$ shall be a Hermitian operator and is therefore symmetric whenever its matrix elements in the computational basis are all real, at which we call $O$ is \emph{real}. We define the \emph{Clifford group} $\mathrm{Cl}_n$ as the set of operators \cite{gottesman1998,aaronson2004},
	\begin{equation}
		\mathrm{Cl}_n=\left\{U\,|\,\forall E\in \mathcal{P}_n,\,UEU^{\dagger}\in \mathcal{P}_n\right\},
	\end{equation}
	where $\mathcal{P}_{n}$ is the \emph{Pauli group} generated by tensor products of the standard Pauli operators $X$, $Y$, $Z$ and the identity~$I$. 
	Next, we define the \emph{stabilizer(-states)} set $\mathcal{S}_n$ as the following set of quantum states:
	\begin{equation}
		\mathcal{S}_n= \left\{\ket{\psi}|\ket{\psi}=U\ket{0^{\otimes n}}\;{\rm for\; some}\; U\in \mathrm{Cl}_n\right\},
	\end{equation}
	where the element $\ket{\psi}\in\mathcal{S}_n$ is a \emph{(pure) stabilizer state}.
	
	We are now ready to look into subclasses of the stabilizer set~$\mathcal{S}_n$. Specifically, the complex \emph{(pure) equatorial stabilizer (-states) set} \cite{bravyi2019} $\mathcal{S}^{{\rm eq}}_n$ and its real counterpart $\mathcal{S}^{{\rm req}}_n$ are respectively given by
\begin{widetext}
	\begin{align}
		&\mathcal{S}^{{\rm eq}}_n=\left\{\ket{\phi^{{\rm eq}}_\dyadic{A}}\equiv\frac{1}{\sqrt{2^n}}\sum_{\bm{x}\in\mathbb{Z}^n_2} \I^{\bm{x}^\top\dyadic{A}\bm{x}}\ket{\bm{x}}\bigg|\,\dyadic{A}=(a_{ij})_{i,j\in [n]}\;{\rm where}\;a_{ij}=a_{ji},\,a_{ij}\in\begin{cases}
			\bfZ_4 & \text{if }i=j\\
			\bfZ_2 & \text{if }i>j
		\end{cases}\right\},\label{main:eq1}\\
		&\mathcal{S}^{{\rm req}}_n=\left\{\ket{\phi^{{\rm req}}_\dyadic{A}}\equiv\frac{1}{\sqrt{2^n}}\sum_{\bm{x}\in\mathbb{Z}^n_2} (-1)^{\bm{x}^\top\dyadic{A}\bm{x}}\ket{\bm{x}}\bigg|\,\dyadic{A}=(a_{ij})_{i,j\in [n]}\;{\rm where}\;a_{ij}\in\begin{cases}
			\bfZ_2 & \text{if }i\le j\\
			\{0\} & \text{if }i>j
		\end{cases}\right\}.\label{main:eq2}
	\end{align}
\end{widetext}
	 From Eqs.~\eqref{main:eq1} and~\eqref{main:eq2}, we know that the number of complex equatorial stabilizer states is $\left|\mathcal{S}^{{\rm eq}}_n\right|=2^{\frac{n^2+3n}{2}}$, and that of real equatorial stabilizer states is $\left|\mathcal{S}^{{\rm req}}_n\right|=2^{\frac{n^2+n}{2}}$. As the names suggest, these are subsets of $\mathcal{S}_n$. The $\mathcal{S}^{{\rm eq}}_n$ was originally used for a faster classical simulation \cite{bravyi2019,bravyi2016i}. More explicitly, the original norm estimation technique~\cite{bravyi2016i}, which is a key step of classical simulation, involves the uniform sampling of stabilizer states and taking the inner product with the target magic state. However, if we restrict the sampling set to the subset $\mathcal{S}^{{\rm eq}}_n\subsetneq \mathcal{S}_n$, we may boost the speeds of sampling and inner product calculation~\cite{bravyi2019}.
	
	In this work, we shall turn these two stabilizer subclasses into POVMs for shadow tomography, and consequently reveal the key advantages that are related to the exploitation of their smaller set sizes relative to $\mathcal{S}_n$. That $\mathcal{S}^{{\rm eq}}_n$ and $\mathcal{S}^{{\rm req}}_n$ may be readily transformed into POVMs is straightforward~\cite{supple}, as the corresponding sets of positive operators $\left\{\frac{2^n}{\left|\mathcal{S}^{\rm eq}_n\right|}\ket{\phi^{{\rm eq}}_\dyadic{A}}\bra{\phi^{{\rm eq}}_\dyadic{A}}\right\}$ and $\left\{\frac{2^n}{\left|\mathcal{S}^{\rm req}_n\right|}\ket{\phi^{{\rm req}}_\dyadic{A}}\bra{\phi^{{\rm req}}_\dyadic{A}}\right\}$ house elements that sum to the identity (see Sec.~\ref{sec:technical} for the proof). We coin them \emph{(real) equatorial stabilizer POVMs}, or (R)ESPOVMs for short.

	\subsection{Random Clifford tomography and classical shadows}\label{sec:random_clifford_tomography}
	
	One of the typical methods for estimating the expectation value of a given observable $O$ is randomized-Clifford shadow tomography~\cite{huang2020}. 
	We briefly explain how the random Clifford operation may be used to estimate $\braket{O}=\tr{\rho\,O}$ (see Ref.~\cite{huang2020} for details). For this purpose, the algorithm using $N$ sampling-copy number is stated below: 
	
	\begin{enumerate}
		\item Take $\rho$ as an input and uniformly choose a Clifford unitary $U$ from $\mathrm{Cl}_n$.
		\item Take $\rho\leftarrow U\rho\,U^{\dagger}$.
		\item Measure in the $Z$~basis and obtain the outcome $\bm{p}\in\bfZ^n_2$.	
		\item Compute 
        \begin{align}
            \widehat{O}&=\tr{O\mathcal{M}^{-1}(U^{\dag}\ket{\bm{p}}\bra{\bm{p}}U)}\nonumber\\&=(2^n+1)\bra{\bm{p}}UOU^{\dag}\ket{\bm{p}}-\tr{O}.
        \end{align}
        Repeat the aforementioned steps $N$ times and take the average of $N$~$\widehat{O}$'s to obtain the estimator $\widehat{\left<O\right>}$ for $\left<O\right>$.
	\end{enumerate}
	
	Here, $\mathcal{M}^{-1}(O)\equiv (2^n+1)O-\tr{O}I$ for an arbitrary operator $O$ is known as the \emph{classical shadow}~\cite{huang2020}. Reference~\cite{huang2020} showed that such randomized-Clifford tomography achieves $\varepsilon$-accurate estimation of $\braket{O}$ with the sampling-complexity upper bound $\mathcal{O}(3\tr{O^2}/\varepsilon^2)$. Hence, this scheme is efficient when $O$ is a target pure state (\emph{fidelity estimation}) of which $\tr{O^2}$ is a constant, albeit with the assumption that every $\widehat{O}$ is efficiently computable.
	Reference~\cite{bravyi2021} showed that we can efficiently sample a Clifford unitary $U$ from $\mathrm{Cl}_n$ in $\mathcal{O}(n^2)$-time. In addition, the sampled form contains the \textsc{hf}$'$\textsc{---sw---h---hf} sections (where small uppercase letters here refer to gate layers as opposed to the regular uppercase ones that denote gate labels, and \textsc{hf} and \textsc{hf}$'$ are \emph{Hadamard-free sections}) where \textsc{hf}-section is a circuit containing layers of \textsc{cnot$^\downarrow$}---\textsc{cz}---\textsc{s}---\textsc{pauli}. Here, \textsc{h}, \textsc{sw}, \textsc{cnot$^\downarrow$}, \textsc{cz} and \textsc{s} respectively refer to \emph{pure layers} of Hadamard~gates, SWAP~gates, CNOT~gates with control qubits higher than the target ones when visualized in a standard circuit diagram, CZ~gates and $S$~gates. Naturally, \textsc{sw} contains many CNOT gates. 

    Direct fidelity-estimation scheme~\cite{flammia2011} is another conventional way for estimating~fidelities. While this method invokes a lower gate count and circuit depth than our schemes, the estimable region of observables is more restricted than both the randomized-Clifford and our (R)ESPOVM schemes. As such, only the randomized-Clifford tomography schemes shall be of interest to be compared with our work.

    \begin{algorithm}[t]
    	\caption{(R)ESPOVM shadow tomography algorithm (with the MOM estimation technique)}
    	\label{main:algo1}
    	\KwIn{$n$-qubit unknown quantum state $\rho$, total input-state copies $N = N'K \in \mathbb{N}$ (with even $N'$ and $K \in \mathbb{N}$), observable $O$, and the (R)ESPOVM scheme}
    	\KwOut{Estimate $m$}
    	
    	\For{$z \in [K]$}{
    		$m_z \gets 0$; \quad $m' \gets 0$\;
    		
    		\For{$k \in [N'/2]$}{
    			Prepare $\rho$ as the input state\;
    			
    			Randomly choose $\ket{\phi^{\mathrm{(r)eq}}_{\mathbf{A}}} \in \mathcal{S}^{\mathrm{(r)eq}}_n$, where $\mathbf{A} = (a_{ij})$\;
    			
    			\For{$i < j \in [n]$}{
    				\If{$a_{ij} = 1$}{
    					Apply $\mathrm{CZ}_{ij}$: $\rho \gets \mathrm{CZ}_{ij} \rho \mathrm{CZ}_{ij}$\;
    				}
    			}
    			
    			\For{$i \in [n]$}{
    				\If{ESPOVM is used}{
    					\uIf{$a_{ii} \equiv 1 \pmod{4}$}{Apply $S^\dagger$ on qubit $i$\;}
    					\uElseIf{$a_{ii} \equiv 2 \pmod{4}$}{Apply $Z$ on qubit $i$\;}
    					\uElseIf{$a_{ii} \equiv 3 \pmod{4}$}{Apply $S$ on qubit $i$\;}
    				}
    				\ElseIf{$a_{ii} \equiv 1 \pmod{2}$ (RESPOVM)}{
    					Apply $Z$ on qubit $i$\;
    				}
    				Apply $H$ on qubit $i$\;
    			}
    			
    			Measure to obtain bit-string $\bm{p} \in \mathbb{Z}^n_2$\;
    			
    			\For{$i \in [n]$\label{algo1:start_for}}{
    				\uIf{ESPOVM is used}{
    					$a_{ii} \gets a_{ii} + 2p_i \pmod{4}$\;
    				}
    				\ElseIf{RESPOVM is used}{
    					$a_{ii} \gets a_{ii} + p_i \pmod{2}$\;
    				}
    			}\label{algo1:end_for}
    			
    			Prepare $\rho$ again and measure in Pauli-$Z$ basis to obtain $\bm{p'} \in \mathbb{Z}^n_2$\;
    			
    			\uIf{ESPOVM is used}{
    				\[
    				\begin{aligned}
    					m' \gets m' + {} & 2^n \bra{\phi^{\mathrm{eq}}_{\mathbf{A}}} O \ket{\phi^{\mathrm{eq}}_{\mathbf{A}}} \\
    					& + \bra{\bm{p'}} O \ket{\bm{p'}} - \mathrm{tr}(O)
    				\end{aligned}
    				\]
    			}
    			\ElseIf{RESPOVM is used}{
    				\[
    				\begin{aligned}
    					m' \gets m' + {} & 2^{n-1} \bra{\phi^{\mathrm{req}}_{\mathbf{A}}} O \ket{\phi^{\mathrm{req}}_{\mathbf{A}}} \\
    					& + \bra{\bm{p'}} O \ket{\bm{p'}} - \tfrac{1}{2} \mathrm{tr}(O)
    				\end{aligned}
    				\]
    			}
    		}
    		
    		$m_z \gets m' / (N'/2)$\;
    	}
    	
    	$m \gets \mathrm{median}(m_1, m_2, \ldots, m_K)$\;
    \end{algorithm}
	
	\section{Equatorial-stabilizer shadow tomography}
	\begin{figure*}[t]
		\centering
		\includegraphics[width=\textwidth]{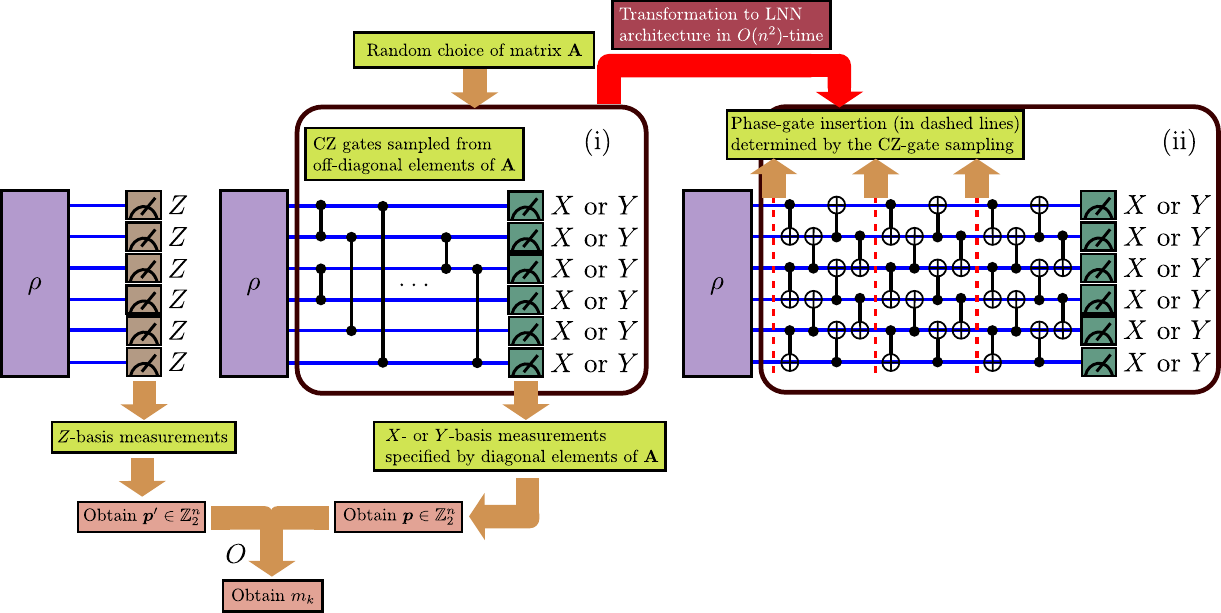}
		
		\caption{Schematic for 6-qubit (R)ESPOVM shadow tomography to estimate~$\left<O\right>=\tr{\rho\,O}$ of an unknown state~$\rho$ and target observable~$O$. Sampled CZ circuits can be transformed to a linear nearest-neighboring (LNN) architecture \cite{maslov2018}~[circuit~(ii)]. Also in~(ii), the inserted phase gates are single-qubit gates and the measurement is done in the same manner as in~(i). Any one of the two architectures~(i) and~(ii) can be executed to obtain the measurement bit-string column~$\bm{p}\in\bfZ^n_2$. Together with the bit-string column~$\bm{p}'$ acquired from the computational-basis measurement, the observable expectation value can be reconstructed after repeating $N$ rounds of $(\bm{p},\bm{p}')$ collection.}
		\label{fig:fig1}
	\end{figure*}
	
	We now introduce the theoretical framework of the shadow tomography scheme that uses (real)~equatorial stabilizer POVMs. We remember that we want to estimate the target value is $\tr{\rho\,O}$. We can show that it can be rewritten as, 
    \begin{align}
        {\rm tr}(\rho\, O)&=\sum_{\ket{\phi^{{\rm eq}}_\dyadic{A}}\in \mathcal{S}^{{\rm eq}}_n}\frac{2^{2n}}{\left|\mathcal{S}^{{\rm eq}}_n\right|}\bra{\phi^{{\rm eq}}_{\dyadic{A}}}\rho\ket{\phi^{{\rm eq}}_{\dyadic{A}}}\bra{\phi^{{\rm eq}}_\dyadic{A}}O\ket{\phi^{{\rm eq}}_{\dyadic{A}}}\nonumber\\&+\sum_{\bm{x}\in\bfZ^n_2}\bra{\bm{x}}\rho\ket{\bm{x}}\bra{\bm{x}}O\ket{\bm{x}}-{\rm tr}(O)\,.\label{eq:equatorial_tomography}
    \end{align}
    If $O$ is a real matrix with respect to the computational basis,
    \begin{align}
        {\rm tr}(\rho\, O)&=\sum_{\ket{\phi^{{\rm req}}_\dyadic{A}}\in \mathcal{S}^{{\rm req}}_{n}}\frac{2^{2n-1}}{\left|\mathcal{S}^{{\rm req}}_n\right|}\bra{\phi^{{\rm req}}_{\dyadic{A}}}\rho\ket{\phi^{{\rm req}}_{\dyadic{A}}}\bra{\phi^{{\rm req}}_{\dyadic{A}}}O\ket{\phi^{{\rm req}}_{\dyadic{A}}}\nonumber\\&+\sum_{\bm{x}\in\bfZ^n_2}\bra{\bm{x}}\rho\ket{\bm{x}}\bra{\bm{x}}O\ket{\bm{x}}-\frac{1}{2}{\rm tr}(O).\label{eq:real_equatorial_tomography}
    \end{align}
    A detailed proof can be seen in Sec.~\ref{sec:technical}. In Sec.~\ref{sec:Clifford_group_and_equatorial_stabilizer}, we explained that $\left\{\frac{2^n}{\left|\mathcal{S}^{\rm eq}_n\right|}\ket{\phi^{{\rm eq}}_\dyadic{A}}\bra{\phi^{{\rm eq}}_\dyadic{A}}\right\}$ and $\left\{\frac{2^n}{\left|\mathcal{S}^{\rm req}_n\right|}\ket{\phi^{{\rm req}}_\dyadic{A}}\bra{\phi^{{\rm req}}_\dyadic{A}}\right\}$ are POVM. As a result, we obtain the following equality, 
    \begin{align}
        \sum_{\dyadic{A}}\frac{2^n}{\left|\mathcal{S}^{\rm eq}_n\right|}\braket{\phi^{{\rm eq}}_\dyadic{A}|\rho|\phi^{{\rm eq}}_\dyadic{A}}= \sum_{\dyadic{A}}\frac{2^n}{\left|\mathcal{S}^{\rm req}_n\right|}\braket{\phi^{{\rm req}}_\dyadic{A}|\rho|\phi^{{\rm req}}_\dyadic{A}}=1.
    \end{align}
    Therefore, we can sample the $\dyadic{A}$, which represents a specific (real) equatorial stabilizer state, from the probability distributions $\left\{\frac{2^n}{\left|\mathcal{S}^{\rm eq}_n\right|}\braket{\phi^{{\rm eq}}_\dyadic{A}|\rho|\phi^{{\rm eq}}_\dyadic{A}}\right\}$ and $\left\{\frac{2^n}{\left|\mathcal{S}^{\rm req}_n\right|}\braket{\phi^{{\rm req}}_\dyadic{A}|\rho|\phi^{{\rm r
    eq}}_\dyadic{A}}\right\}$. 

    Now we present the estimation scheme for the ESPOVM. The RESPOVM follows a similar manner and will be shown in Sec.~\ref{sec:technical}. We prepare two copies of $\rho$ [$N$ (even) copies in total]. The first copy is measured to obtain the outcome $\dyadic{A}$ by the distribution  $\left\{\frac{2^n}{\left|\mathcal{S}^{\rm eq}_n\right|}\braket{\phi^{{\rm eq}}_\dyadic{A}|\rho|\phi^{{\rm eq}}_\dyadic{A}}\right\}$, and the second copy is measured by the computational basis to obtain the outcome $\bm{p'}\in \bfZ^n_2$. Next, we take the estimator $\widehat{O}$ as, 
    \begin{align}
        \widehat{O}=2^n\braket{\phi^{{\rm eq}}_\dyadic{A}|O|\phi^{{\rm eq}}_\dyadic{A}}+\braket{\bm{\bm{p'}}|O|\bm{p'}}-\tr{O}.
    \end{align}
    We repeat such measurements many times to have many copies of $\widehat{O}$'s and take the average to obtain the desired estimation $\widehat{\braket{O}}$. To further reduce the sampling copies, we can apply the median of means (MOM) estimation technique~\cite{lerasle2019lecture,Blair:1985problem,Jerrum:1986random} by simply repeating the scheme and take the median of the collection of estimation values. The complete systematic procedure for (R)ESPOVM shadow tomography is shown in Algo.~\ref{main:algo1} whose schematic illustration is given in Fig.~\ref{fig:fig1}. The theoretical validation of such algorithms, including the one of RESPOVM which is derived by a similar logic, will be shown in Sec.~\ref{sec:technical}.
    
    Let us supply the details for the main steps.
    From the perspective of classical shadows~\cite{huang2020,aronson2018}, we note that our (R)ESPOVM tomography operates with the following "one-to-two-copy" measurement channel $\mathcal{M}$
    \begin{align}\label{eq:measurement_channel}
        &\mathcal{M}(\rho)\nonumber\\&=\left(\sum_{\dyadic{A}}\frac{2^n}{\mathcal{S}_n^{\rm eq}}\braket{\phi^{\rm eq}_{\dyadic
        A}|\rho|\phi^{\rm eq}_{\dyadic{A}}}\ket{\phi^{\rm eq}_{\dyadic
        A}}\bra{\phi^{\rm eq}_{\dyadic
        A}},\sum_{\bm{x}}\braket{\bm{x}|\rho|\bm{x}}\ket{\bm{x}}\bra{\bm{x}}\right).
    \end{align}
    The measurement channel for RESPOVM similarly follows. Hence, the scheme outputs the following respective classical shadows (inverted channel map $\mathcal{M}^{-1}$~\cite{huang2020}): 
		\begin{align}\label{eq:classical_shadow}
			&\mathcal{M}^{-1}(O_1,O_2)\nonumber\\&=
			\begin{cases}
				2^n O_1+O_2-I&(\small{{\rm ESPOVM+Comp.basis}})\,,\\
				2^{n-1}O_1+O_2-\frac{I}{2}&(\small{\rm RESPOVM+Comp. basis}).
			\end{cases}
		\end{align}
        Such a "one-to-two-copy" measurement channel again implies that we need two copies of $\rho$ for the single measurement step. We can note that indeed, $\mathcal{M}^{-1}\circ\mathcal{M}(O)=O$ (see Sec.~\ref{sec:technical} for details) for arbitrary (real part for RESPOVM) Hermitian operator $O$.

    The reason why we need the computational basis measurement is that the ESPOVM is \emph{not} informationally complete~(IC), whereas the union of the ESPOVM and computational basis is IC \cite{supple,scott2006}. In other words, along with the computational basis measurements, we are capable of reconstructing an arbitrary $\rho$ and, thus, general observable expectation values. 
    So, two copies of $\rho$ are required for a single measurement trial, which would give us a bit-string column~$\bm{p}$ from the ESPOVM and another bit-string column~$\bm{p}'$ from the computational basis. However, in Sec.~\ref{subsec:STprop}, we shall see that only a single copy of the input state is necessary for each trial if~$O$ is a cluster state. 

    From Algo.~\ref{main:algo1}, we see that ESPOVM measurement requires $S,\,H $ gates before the computational basis measurements, but RESPOVM does not need $S$ gates. Equivalently, the ESPOVM needs $X,Y$-basis measurements while RESPOVM only needs $X$-basis measurement. While both structures need similar two-copy resources for the single copy measurement, we shall see in the next section that in many cases, the RESPOVM requires much fewer sampling-copy numbers than ESPOVM to achieve the same estimation accuracy.

		We summarize the requirements for implementing Algo.~\ref{main:algo1} for (R)ESPOVM shadow tomography in
		\begin{theorem}\label{main:thm0}
			{\bf [(R)ESPOVM implementation]} On an $n$-qubit state~$\rho$, the implementation of $n$-qubit ESPOVM shadow tomography requires CZ-gate circuits of depths of at most~$n$ by Vizing's theorem~\cite{maslov2022}, and $n$~single-qubit Pauli~(X,Y,Z) measurements. For RESPOVM shadow tomography, the $Y$-basis measurement is not needed. 
		\end{theorem}\vspace{2ex}
		
		\section{Sampling-copy complexity}
		\label{subsec:STprop}
		
		Shadow tomography aims to estimate state properties with a sampling complexity that grows at most polynomially quickly with the qubit number $n$. To be more precise, let us specify that a single property estimator is accurate up to a maximal additive error of $\varepsilon>0$. Then, the number of sampling copies $N$ should scale as at most $\mathcal{O}(\mathrm{poly}(n)/\varepsilon^2)$. We shall prove that our (R)ESPOVM algorithms also possess upper bounds of similar scaling behaviors with randomized-Clifford tomography. 

        Before doing so, we first explain the relation between required sampling copies and the variance of the estimator. We denote $\widehat{O}^{{\rm eq}}_{\dyadic{A}}\equiv2^{n}\bra{\phi^{{\rm eq}}_{\dyadic{A}}}O\ket{\phi^{{\rm eq}}_{\dyadic{A}}}$,  $\widehat{O}^{{\rm req}}_{\dyadic{A}}\equiv2^{n-1}\bra{\phi^{{\rm req}}_{\dyadic{A}}}O\ket{\phi^{{\rm req}}_{\dyadic{A}}}$, and $\widehat{O}^{{\rm bin}}_{\bm{p'}}\equiv\bra{\bm{p'}}O\ket{\bm{p'}}$ for one observable $O$ that are measured and obtained in each step of the shadow-tomography procedure, where the matrix~$\dyadic{A}$ and bit string~$\bm{p'}$ are sampled in accordance with Algo.~\ref{main:algo1}. The first two random numbers $\widehat{O}^{{\rm eq}}_\dyadic{A}$ and $\widehat{O}^{{\rm req}}_\dyadic{A}$ refer to the respective outcome estimators for ESPOVM and RESPOVM, and $\widehat{O}^{{\rm bin}}_{\bm{p'}}$ is the outcome of a $Z$-basis measurement. Upon defining $\mathbb{E}(\widehat{O}_{\dyadic{A},\bm{p'}})$ as the analytic average value of $\widehat{O}_{\dyadic{A},\bm{p'}}$, which is defined as 
		\begin{align}
			\widehat{O}_{\dyadic{A},\bm{p'}}\equiv\begin{cases}\widehat{O}^{{\rm eq}}_{\dyadic{A}}+\widehat{O}^{{\rm bin}}_{\bm{p'}}-\tr{O}&({\rm ESPOVM})\\
				\widehat{O}^{{\rm req}}_{\dyadic{A}}+\widehat{O}^{{\rm bin}}_{\bm{p'}}-\frac{1}{2}\tr{O}&({\rm RESPOVM}),
			\end{cases}
		\end{align} the estimation variance ${\rm Var}(\widehat{O}_{\dyadic{A},\bm{p'}})\equiv \mathbb{E}(\widehat{O}_{\dyadic{A},\bm{p'}}^2)-(\mathbb{E}(\widehat{O}_{\dyadic{A},\bm{p'}}))^2$ of our two shadow-tomography schemes is shown as
		\begin{align}\label{main:eq_vareq}
			{\rm Var}^{{\rm (r)eq}}(\widehat{O}_{\dyadic{A},\bm{p'}})&={\rm Var}(\widehat{O}^{{\rm (r)eq}}_{\dyadic{A}})+{\rm Var}(\widehat{O}^{{\rm bin}}_{\bm{p'}})\nonumber\\&={\rm Var}(\widehat{O_0}^{{\rm (r)eq}}_{\dyadic{A}})+{\rm Var}(\widehat{O_0}^{{\rm bin}}_{\bm{p'}}),
		\end{align}
         where $O_0\equiv O-\tr{O}\frac{I}{2^n}$ (traceless part),
		where the first equality is because $(\dyadic{A},\bm{p'})$ sampling is independent of each other.
		The variance terms on the right side, after maximizing over the input state~$\rho$, are equivalent to the respective squared shadow norms~\cite{huang2020} of the (R)ESPOVM and $Z$-basis measurement.
        To estimate $\braket{O}$, we take the sample mean $\widehat{\braket{O}}$ over the $N$ numbers of sampled values $\widehat{O}_{\dyadic{A},\bm{p'}}$. 
        
	The sample variance of $\widehat{\left<O\right>}$ for estimating a single property $\left<O\right>$ is given by
        
		\begin{equation}
			\mathrm{Var}\!\left(\widehat{\left<O\right>}\right)=\begin{cases}
				\dfrac{{\rm Var}^{{\rm eq}}(\widehat{O}_{\dyadic{A},\bm{p'}})}{N/2} & (\text{ESPOVM, comp. basis})\,,\\[2ex]
				\dfrac{{\rm Var}^{{\rm req}}(\widehat{O}_{\dyadic{A},\bm{p'}})}{N/2} & (\text{RESPOVM, comp. basis})\,.
			\end{cases}
		\end{equation}
		Note that since $\widehat{\left<O\right>}$ is unbiased, $\mathrm{Var}\!\left(\widehat{\left<O\right>}\right)$ is identical to the mean-squared error of $\widehat{\left<O\right>}$ that measures the estimation accuracy of $\widehat{\left<O\right>}$ with respect to $\left<O\right>$:
		\begin{equation}
			\mathrm{Var}\!\left(\widehat{\left<O\right>}\right)=\mathbb{E}\!\left(\left(\widehat{\left<O\right>}-\left<O\right>\right)^2\right)\,.
			\label{eq:MSE_defn}
		\end{equation}
        Therefore, Chevyshev inequality implies that ${\rm Var}^{{\rm (r)eq}}(\widehat{O}_{\dyadic{A},\bm{p'}})$ determines the required number of sampling copies $\mathcal{O}\left(\frac{{\rm Var}^{{\rm (r)eq}}(\widehat{O}_{\dyadic{A},\bm{p'}})}{\epsilon^2\delta}\right)$ to achieve the desired accuracy within the failure probability $\delta$. It is known that with the aid of MOM technique~\cite{Jerrum:1986random,huang2020}, we can reduce such complexity to $\mathcal{O}\left(\frac{{\rm Var}^{{\rm (r)eq}}(\widehat{O}_{\dyadic{A},\bm{p'}})}{\epsilon^2}\log(\delta^{-1})\right)$. 

        The dominant term determining the scaling of ${\rm Var}^{{\rm (r)eq}}(\widehat{O}_{\dyadic{A},\bm{p'}})$ is ${\rm Var}(\widehat{O_0}^{{\rm (r)eq}}_{\dyadic{A}})$, and its upper bound $\mathbb{E}((\widehat{O_0}^{{\rm (r)eq}}_\dyadic{A})^2)$ is rewritten by,
        \begin{align}
			\mathbb{E}((\widehat{O_0}^{{\rm (r)eq}}_\dyadic{A})^2)\nonumber=&\,\frac{2^{\gamma}}{\left|\mathcal{S}^{{\rm (r)eq}}_n\right|}\sum_{\ket{\phi^{{\rm (r)eq}}_\dyadic{A}}\in\mathcal{S}^{{\rm (r)eq}}_n}\bra{\phi^{{\rm (r)eq}}_\dyadic{A}}\rho\ket{\phi^{{\rm (r)eq}}_{\dyadic{A}}}\\
			&\times\bra{\phi^{{\rm (r)eq}}_\dyadic{A}}O_0\ket{\phi^{{\rm (r)eq}}_{\dyadic{A}}}\bra{\phi^{{\rm (r)eq}}_\dyadic{A}}O_0\ket{\phi^{{\rm (r)eq}}_{\dyadic{A}}}\label{eq:exp_square_sub}\\
			\nonumber=&\,2^{\gamma}\,{\rm tr}\Big((\rho\otimes O_0 \otimes O_0)\\
			&\times\momentopp^{\otimes 3}\Big)\,,\label{eq:exp_square}\\
			\gamma=&\,\begin{cases}
				3n & \text{for eq}\,,\\
				3n-2 & \text{for req}\,.
			\end{cases}\nonumber
	\end{align}
        Here, $\mathbb{E}_{\ket{\phi_{\dyadic{A}^{\rm (r)eq}}}\sim \mathcal{D}(\mathcal{S}_{n}^{\rm (r)eq})}$ means the analytic average over uniformly chosen (real) equatorial stabilizer states. We see that $\eqmomentop^{\otimes 3}$ is an important factor for the variance scaling. We clarify its structure in the following lemma, which will be proved in Sec.~\ref{proof_lem1}.

        \begin{lemma}\label{method:lem3}
			\begin{align}
				&\,\eqmomentop^{\otimes 3}\nonumber\\
				=&\,\frac{1}{8^{n}}\sum_{\left(\bm{x},\bm{y},\bm{z},\bm{w},\bm{s},\bm{t}\right)\in \mathcal{K}_1(\bfZ^n_2) }\ket{\bm{x}\bm{y}\bm{z}}\bra{\bm{w}\bm{s}\bm{t}},\\
				&\,\momentop^{\otimes 3}\nonumber\\
				=&\,\frac{1}{8^{n}}\sum_{\left(\bm{x},\bm{y},\bm{z},\bm{w},\bm{s},\bm{t}\right) \in \mathcal{K}_2(\bfZ^n_2) }\ket{\bm{x}\bm{y}\bm{z}}\bra{\bm{w}\bm{s}\bm{t}},
			\end{align}
            where
		\begin{align}
			&\mathcal{K}_1(\bfZ^n_2)\equiv\big\{\bm{v}=(\bm{x},\bm{y},\bm{z},\bm{w},\bm{s},\bm{t})|\bm{x},\bm{y},\bm{z},\bm{w},\bm{s},\bm{t}\in\bfZ^{n}_2\nonumber\\
			&\text{
            $\exists$ a partition of $\left\{1,2,\ldots,6\right\}$ into $3$ pairs}\; \text{$\{(i_1,j_1),(i_2,j_2),$} \nonumber\\&\text{$(i_3,j_3)\}$\;such that $\bm{v}_{i_k}=\bm{v}_{j_k}$ for $k=1,2,3$}\big\},\nonumber\\
			&\mathcal{K}_2(\bfZ^n_2)\equiv\big\{\bm{v}=(\bm{x},\bm{y},\bm{z},\bm{w},\bm{s},\bm{t})|\bm{v}\in\mathcal{K}_1(\bfZ^n_2)\;\text{, and if}\nonumber\\
 			&\left\{\bm{x},\bm{y},\bm{z}\right\} \;\text{and}\; \left\{\bm{w},\bm{s},\bm{t}\right\}\;\text{share a common element, say $\bm{u}$,}\nonumber\\
 			&\text{and the other 2 variables in $\left\{\bm{x},\bm{y},\bm{z}\right\}\backslash\left\{\bm{u}\right\}$ or $\left\{\bm{w},\bm{s},\bm{t}\right\}\backslash\left\{\bm{u}\right\}$}\nonumber\\ 
 			&\text{equal, say, $\bm{u}'$, then the remaining 2 variables also equal to~$\bm{u}'$.}\big\}.
		\end{align}
		\end{lemma}
        
        For example, $\left\{\left((1,1),(0,0),(0,0),(1,1),(0,0),(0,0)\right)\right.\\\left.,\left((1,0),(1,1),(0,0),(1,0),(1,1),(0,0)\right),((1,1),(0,1),\right.\\\left.(0,1),(1,1),(1,0),(1,0))\right\}\subset \mathcal{K}_1(\bfZ^2_2)$, and $\left\{((1,1),(0,0),\right.\\\left.(0,0),(1,1),(0,0),(0,0)),((1,0),(1,1),(0,0),(1,0),(1,1),\right.\\\left.(0,0))\right\}\subset\mathcal{K}_2(\bfZ^2_2)$, but $((1,1),(0,1),(0,1),(1,1),(1,0),\\(1,0))\notin\mathcal{K}_2(\bfZ^2_2)$.

        We substitute the above lemma to Eq.~\eqref{eq:exp_square} then we prove that $\mathbb{E}((\widehat{O_0}^{{\rm (r)eq}}_\dyadic{A})^2)$ is bounded by $\mathcal{O}\left(\tr{O_{0}^2}\right)$. In conclusion, the variance of our (R)ESPOVM estimator achieves similar scaling with one of the random Clifford measurements~\cite{huang2020} and a similar sampling-copy complexity. Moreover, taking average of $\mathbb{E}((\widehat{O_0}^{{\rm (r)eq}}_\dyadic{A})^2)$ over uniformly random input and output states, we can show that the average of RESPOVM is exactly half of one of ESPOVM and even half of one of randomized Clifford tomography. We encapsulate the following results as a theorem whose a detailed derivation will be shown in 
        Sec.~\ref{subsec:est-var} and Sec.~\ref{sec:proof_other_lemmas}. 
		\begin{theorem}\label{main:thm1}
			{\bf[Sampling-complexity efficiency of (R)ESPOVM shadow tomography]} Suppose that one is able to prepare multiple copies of an unknown input quantum state $\rho$, and that one is given $M\ge1 $ observables $O_1$, $O_2$, \ldots, $O_M$ (for RESPOVM shadow tomography, all $O_j$'s are real).\\
			\noindent
			(i) With (R)ESPOVM shadow-tomography scheme, we can estimate each $\tr{\rho\,O_j}$ to within an additive $\varepsilon$-error margin with a success probability $1-\delta$ if we repeat this scheme over $N=\mathcal{O}\left(\frac{\max_{1\le j\le M}\left\{\tr{O^2_{j0}}\right\}}{\varepsilon^2}\,\log\!\left(\frac{2M}{\delta}\right)\right)$ copies, where $O_{j0}$ is the traceless part, $O_j-\frac{I}{2^n}\tr{O_j}$.\\
			\noindent
			(ii) If $2^n\gg1$, the averaged sampling-copy number, to achieve the estimation within an additive $\varepsilon$-error margin and a success probability $1-\delta$, over uniformly-distributed complex pure states~$\rho$ and real pure states~$\sigma$ approaches $\frac{136\;(68\;{\rm resp.})}{\varepsilon^2}\,\log\!\left(\frac{2M}{\delta}\right)$ for (R)ESPOVM.\\
			\noindent
			(iii) If $2^n\gg 1$, the averaged sampling-copy number to achieve the estimation within an additive $\varepsilon$-error margin and a success probability $1-\delta$, averaged over uniformly-distributed complex pure states $\rho$ and $\sigma$, approaches $\frac{136}{\varepsilon^2}\,\log\!\left(\frac{2M}{\delta}\right)$ for ESPOVM.
		\end{theorem}\vspace{2ex}
		
		\begin{figure*}[t]
			\centering
			\includegraphics[width=\textwidth]{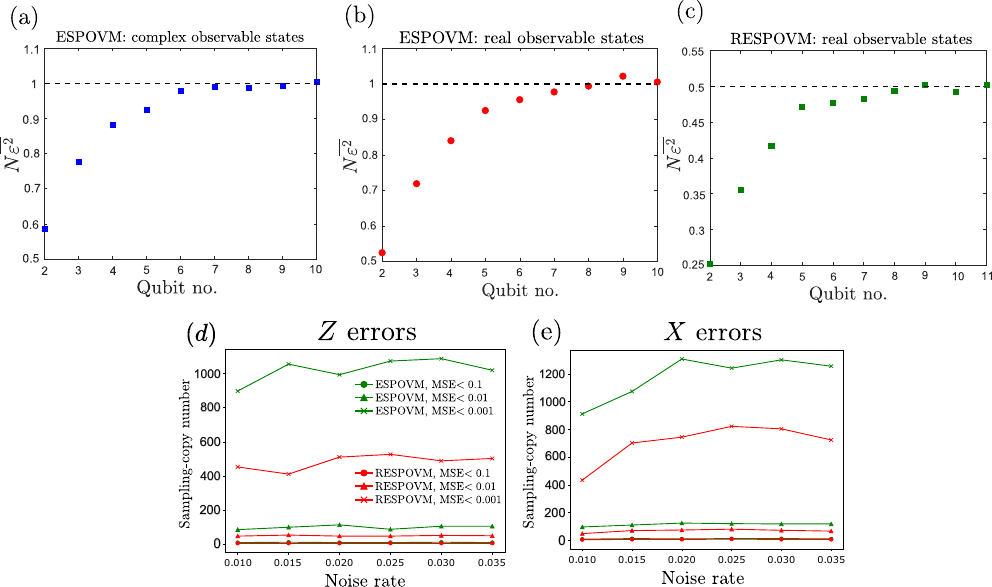}
			
			\caption{(a--c) Sampling-copy numbers $N$ multiplied with squared error $\overline{\epsilon^2}$ averaged on 200 experiments of different $n$-qubit (pure) input and target state observables. We set $N=2000$ for (a,b) and $N=1000$ for (c), also for each experiment, we put 50 copies with given $N$ to obtain the $\epsilon^2$ which is an averaged value of squared errors between the ideal expectation value of target and estimation, for one experiment. For~(a), we averaged the $\epsilon^2$ over uniform distributions of complex input and target states. For~(b) and~(c), the average is over uniform distributions of complex input and real target states. The horizontal dashed lines represent theoretically asymptotic values~\cite{supple} of $N\overline{\epsilon^2}$. (d,e)~Sampling-copy numbers $N$ required to achieve the respective mean squared error thresholds (averaged over 400 experiments) for the fidelity estimation between a $50$-qubit noisily-prepared GHZ state and pure GHZ state. We consider the i.i.d.~(d)~$Z$- and (e)~$X$-error channels acting on every qubit with an error probability $\eta_\mathrm{prep}$.}
			
			\label{fig:fig2}
		\end{figure*}
		
		Theorem~\ref{main:thm1} (i) implies that our scheme is efficient when the Frobenius norm of given observable is at most polynomial with the qubit number~$n$. Most notably, if $O$ is also a quantum state~$\sigma$, then $\tr{O_0^2}$ becomes constant, hence the required sampling-copy number~$N$ is independent of~$n$. When the target state is real, Thm~\ref{main:thm1}~(ii) and~(iii) imply that the average sampling-copy bound of RESPOVM is half of one of ESPOVM. Furthermore, we know that random target state $\ket{\psi}$ is completely anti-concentrated~\cite{dalzell2022}. That is, $\mathbb{E}_{\ket{\psi}}\sum_{\bm{x}\in \bfZ^n_2}|\braket{\bm{x}|\psi}|^4=\mathcal{O}(2^{-n})$. Markov's inequality implies that with high probability $|\braket{\bm x|\psi}|^2\le \sqrt{\frac{\poly(n)}{2^n}}$ for all $\bm{x}\in \bfZ^n_2$ and then We can neglect the estimator of the second computational basis measurement part for large $n$, saving half of the copies. Hence, we can expect RESPOVM shadow tomography to be more useful for real observable estimation.
		
		\begin{figure*}[t]
			\centering
			\includegraphics[width=2\columnwidth]{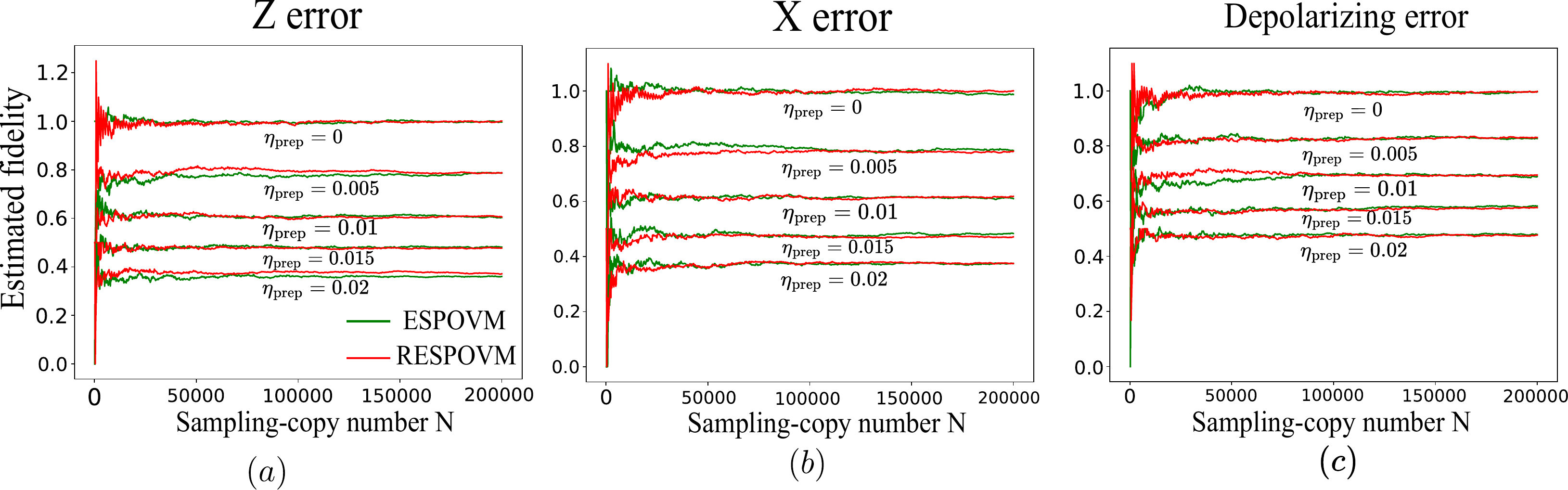}
			\caption{The estimated fidelity with the sampling-copy number~$N$ in fidelity estimation between noisily-prepared $7\times7$-qubit graph states and their noiseless counterparts $\ket{\psi}$. The (a)~$Z$, (b)~$X$ and (c)~depolarizing preparation-noise channels are i.i.d. to the individual qubits. Each estimated value of a given $N$ is obtained from the median of all $\widehat{\left<\ket{\psi}\bra{\psi}\right>}$ estimated using the (R)ESPOVM scheme over $100$ experiments ($N/100$ copies for each experiment).}
			\label{fig:fig3}
		\end{figure*}
		
		To showcase the performance of (R)ESPOVM shadow tomography, we examine scenarios in observable-expectation value estimation between arbitrary quantum states. Figure~\ref{fig:fig2}(a--c) illustrate three exemplifying cases: ESPOVM shadow tomography for uniformly-chosen random complex input and target states, ESPOVM shadow tomography for uniformly-chosen complex input and real target states, and RESPOVM shadow tomography for the same types of random states as in the second case. We observe that within $N=2000$, averaged squared error over randomly chosen states becomes inversely proportional to $N$. These results are consistent with the points of Thm.~\ref{main:thm1}. Moreover, the saturated sampling copy number to achieve $\overline{\epsilon^2}\le 0.001$ is much lower than the upper bound shown in Thm.~\ref{main:thm1}~(ii,iii). Furthermore, for real target states, the saturated value of $N$ for RESPOVM is exactly half of that for ESPOVM. Figure~\ref{fig:fig2}(d,e) demonstrates the success of (R)ESPOVM shadow tomography in estimating the fidelity between 50-qubit noisily-prepared GHZ states~\cite{briegel2001} and pure GHZ states with reasonable~$N$~values, where one also observes that RESPOVM needs significantly fewer sampling copies than ESPOVM to achieve the same accuracy. Figure~\ref{fig:fig3} brings (R)ESPOVM shadow tomography to the test through estimation of fidelity between noisy~7$\times$7 graph states~\cite{anders2006} and its pure counterparts that has CZ-gate connections to all neighboring qubits. Here, $N$ is also the number of measurement trials because every computational basis measurement outcome gives the same $\widehat{O}^{{\rm bin}}_{\bm{p'}}\equiv\bra{\bm{p'}}O\ket{\bm{p'}}=\frac{1}{2^n}$, so that such a second measurement part may be omitted.
		
		Alternatively, fidelity estimations of stabilizer states may be carried out using simpler specialized quantum circuits reported in~Ref.~\cite{flammia2011}. Nevertheless, we emphasize that \emph{all} target states achieve the same asymptotic sampling-complexity bound that is constant in~$n$ with (R)ESPOVM shadow-tomography. This means that our technique clearly applies to many other interesting states, including the Dicke states~\cite{Brtschi2019}, for which $N=\mathcal{O}(n^{2k})$ using the method in~Ref.~\cite{flammia2011}.

		\section{Resources for and noise tolerance of (R)ESPOVM}
		
		\label{eq:subsec:low-gate-count}
		
		Recalling Thm.~\ref{main:thm0}, we note that (R)ESPOVM shadow tomography requires only CZ~gates as two-qubit gate resources and hence it needs at most $\mathcal{O}(n^2)$-time. In the language of quantum circuits, Vizing's theorem~\cite{berge1991,misra1992,maslov2022} implies that an arbitrary CZ circuit requires a circuit depth of at most~$n$ when only CZ gates are employed. Whereas, the randomized Clifford tomography, given that \textsc{hf} section (see Sec.~\ref{sec:random_clifford_tomography} for its definition) can be regarded as classical post-processing after the measurement, needs at most $2n$-depth. 
		
		Furthermore, with just nearest-neighboring (NN) two-qubit gates, CZ circuits with the qubit reversal swapping needs only $2n$ depth~\cite{maslov2018,supple} of CNOT and S~gates. Since we can regard the qubit reversal as the reversal of measurement outcomes, (R)ESPOVM needs only $2n$ depth, which is lower than the worst-case circuit depth of~$7n$ for the implementation of randomized-Clifford tomography~\cite{huang2020,bravyi2021,maslov2023,supple}. In Sec.~\ref{subsec:gatecount}, we further reduced the depth of randomized Clifford tomography to $3n$. This is possible since we are using only non-adaptive Clifford circuits where all CNOT gates have a control qubit on the top. These arguments are gathered into the following statements:\\
		
		\begin{figure*}[t]
			\centering
			\includegraphics[width=1.8\columnwidth]{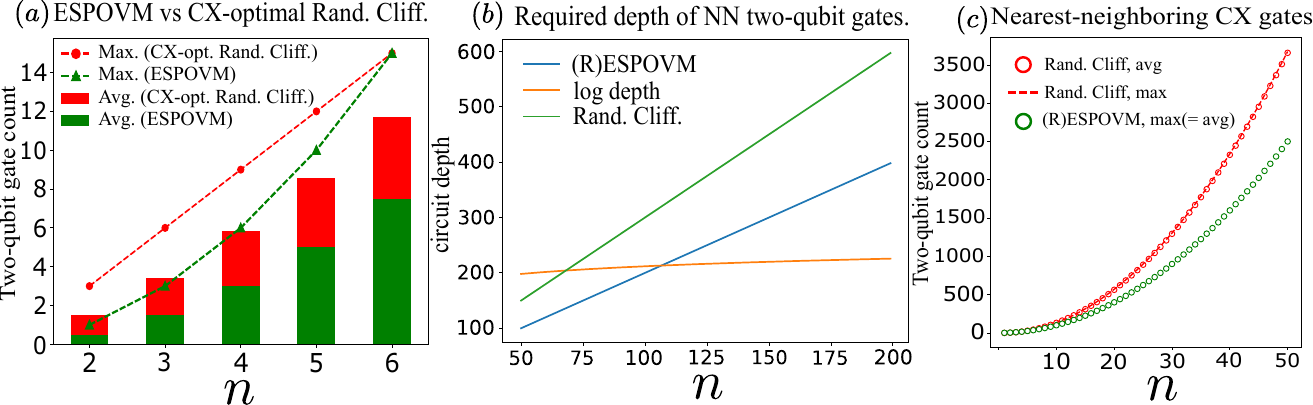}
			
			\caption{(a)~The average number of long-ranged two-qubit gates to implement (R)ESPOVM and randomized-Clifford circuits for low-qubit systems. The latter is obtained \emph{via} Tab.~6 of Ref.~\cite{bravyi2022} and is the average gate count over all equivalence classes. (b) The maximal required Clifford circuit depth for shadow tomography using neighboring gates. Using the knowledge of Ref.~\cite{maslov2023} (refer to Sec.~\ref{subsec:gatecount} in ''Methods'' for a detailed elaboration), the maximal depth of random-Clifford tomography (blue) is $3n$. The orange line indicates the depth for  \emph{approximative} $3$-design method~\cite{schuster2025,duncan2020,maslov2023} with some bias, $10\cdot2\ln\left(\frac{9n}{\rm bias}\right)$. Here, we assign the allowed bias to $0.01$. (c)~The maximum and average numbers of NN CNOT~gates~(over 4000 gate samples) to implement a $2n$-depth~(R)ESPOVM and $3n$-depth neighboring implementation of randomized-Clifford circuits~\cite{maslov2023}.}
			\label{fig:gatecounts}
		\end{figure*}
		
		\begin{theorem}\label{main:thm2}
			{\bf[Gate-count and depth efficiency of (R)ESPOVM shadow tomography]} \cite{maslov2018}~Both ESPOVM and RESPOVM shadow-tomography algorithms can be implemented by circuits of  
            depth at most $2n$ using neighboring Clifford gates, while randomized Clifford needs at most $3n$-depth neighboring Clifford gates. The implementation of linear nearest-neighboring~(LNN) architecture takes $\mathcal{O}(n^2)$-time for each sampling~copy.  
		\end{theorem}\vspace{2ex}
		
		This fact, along with Thm.~\ref{main:thm0} and~\ref{main:thm1}, argues that (R)ESPOVM algorithm is not a simple trade-off of Clifford circuit depth and sampling-copy number. Because scale-factor reduction of depth while preserving sampling copy scale gives significant improvement of gate noise threshold and error mitigation capability~\cite{tsubouchi2023}, it is also more compatible with noisy intermediate-scale quantum (NISQ) algorithms~\cite{larose2022}.
		
		\begin{figure}[t]
			\centering
			\includegraphics[width=\columnwidth]{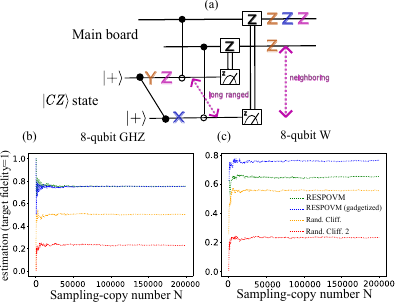}	
			\caption{(a) State-injection technique that transforms an arbitrary Pauli noise on the long-ranged gate into dephasing noise. (b,c) CZ-based RESPOVM shadow-tomography performances in the estimation of fidelity between two pure $8$-qubit GHZ (~W states resp.) states with an increasing number of input-state copies~$N$. Here, we give the depolarizing noise, with $0.05$ error rate, to each long-ranged gate (except the neighboring gates, which may also connect ($1,n$)-qubits). The RESPOVM (gadgetized) means that for long-ranged CZ gates, we gave the dephasing error with $0.05$ error rate (uniform $(I,X,Y,Z)$ operation to each qubit, with a given error rate). This setting is equivalent to the above state-injection technique with noiseless neighboring gates. The Rand. Cliff. indicates the sampling and decomposition scheme via Ref.~\cite{bravyi2021} in which noise is attributed only to the single \textsc{hf}-section. This is because classically post-processing of the measurement outcomes is sufficient to implement the other \textsc{hf}-section. The Rand. Cliff. $2$ method follows the sampling of Ref.~\cite{aaronson2004}.}
	       \label{fig:cz_gadgetization_avg_fidelity}
		\end{figure}
		
		Figure~\ref{fig:gatecounts} shows that our scheme significantly improved over conventional randomized-Clifford tomography regarding the required circuit depth and gate counts. The maximal gate count of our scheme matches with the random Clifford when $n=6$. However, such CX-optimal random Clifford implementation is inefficient for larger $n$ because a super-exponentially sized database of circuit templates is required~\cite{bravyi2022}. We also compared with a recent result of Ref.~\cite{schuster2025} using a log-depth random Clifford circuit to \emph{approximately} realize a unitary $3$ design with some bias. This is because the unitary $3$-design property is a sufficient condition for efficient shadow tomography with a similar scaling of the sampling complexity. Even though the depth-scaling of our scheme is not logarithmic, it requires a larger depth than our scheme and even the random-Clifford scheme when $n$ is of the order of tens. The reason is that Ref.~\cite{schuster2025} uses bilayers of alternating randomized-Clifford blocks, each of which has a qubit size occupying the whole circuit given that $n$ is not large. In this case, the depth of the circuit is $10$ times the size of a random Clifford block because the first layer needs the $7$ times~\cite{duncan2020,maslov2023}, and the second needs the $3$ times ($\because$ Thm.~\ref{main:thm2}) of the size of a block~\cite{maslov2023}. 
        
        The presence of multiple two-qubit gates in a circuit-based quantum-algorithm implementation, be it randomized-Clifford or (R)ESPOVM shadow tomography, enforces a restriction on the individual gate-noise threshold if one is to achieve a given total error rate of the output state the observable measurement is performed. In the previous section, we demonstrated that our (R)ESPOVM scheme requires a lower circuit depth than randomized Clifford sampling. Therefore, we expect that the required noise threshold for each two-qubit gate is affordably larger for (R)ESPOVMs. 

        Furthermore, our measurement scheme based on (R)ESPOVMs also leads to other interesting noise reduction effects. First, if the target state $\psi=\frac{1}{\sqrt{L}}\sum_{i=1}^{L}c_{i}\ket{\bm{x}_i}$, then we can find the sufficient condition $\forall i\in [L],\;\bm{x}_i\cdot \xi_{\bm{p}}=0\;({\rm mod} \;2)$ such that the noise-corrupted output $\bm{p}+\xi_{\bm{p}}$, formed by adding the bit-flip error string $\xi_{\bm{p}}$ to the ideal output $\bm{p}$ of (R)ESPOVM, still gives the right estimator. 

        To show this, suppose we have a bit-flip error string $\xi_{\bm{p}}$ added to the desired outcome $\bm{p}$. Then from the above arguments, the sampled $a_{ii}$ transforms to $a_{ii}+2(p_{i}\oplus (\xi_{\bm{p}})_i)$ where $\oplus$ is the sum of modulo~$2$. Now, suppose the target state $\ket{\psi}=\frac{1}{\sqrt{L}}\sum_{j=1}^{L}c_j\ket{\bm{x}_j}$, where each $\bm{x}_j\in \bfZ^n_2$. Then $\widehat{O}_{\dyadic{A}}\equiv 2^n\bra{\phi^{\rm eq}_\dyadic{A}}O\ket{\phi^{\rm eq}_\dyadic{A}}$ with noisy outcome $\dyadic{A}$ becomes, 
        \begin{align}
            \widehat{O}_{\dyadic{A}}=\frac{1}{L}\left|\sum_{j=1}^{L}c_j\cdot i^{\bm{x}_j\cdot \tilde{\bm{a}}+\bm{x}_j^{\top}\dyadic{A}\bm{x}_j}\right|^2,
        \end{align}
        where $\tilde{a}_i=2(p_i\oplus (\xi_{\bm{p}})_i)-2p_{i}=2(\xi_{\bm{p}})_i\;(\rm mod\;4)$. Hence, $ i^{\bm{p}_j\cdot \tilde{\bm{a}}}=(-1)^{\bm{x}_j\cdot\xi_{\bm{p}}}$. It means that if $\forall j\in [L]$, $\bm{x}_j\cdot\xi_{\bm{p}}=0$, then estimator outputs the right value. The estimator can bear more bit-flip errors as the dimension of the spanning set of $\left\{\bm{x}_j\right\}_{j\in [J]}$ gets lower.
        
        As typical examples, we suppose that the target state is the GHZ state [\,$\ket{\mathrm{GHZ}_n}\equiv\frac{1}{\sqrt{2}}(\ket{0}^{\otimes n}+\ket{1}^{\otimes n})$\,]~\cite{greenberger2007}. Then, the RESPOVM estimator reads
        \begin{align}
        	\widehat{O}=&\,\frac{1}{4}\left\{1+(-1)^{\sum_{i=1}^{n}(a_i+p_{i})+\sum_{i,j\in[n],i<j}a_{ij}}\right\}^2\nonumber\\
        	&+|\braket{{\rm GHZ}_n|\bm{p'}}|^2-\frac{1}{2},
        \end{align}
    where $\bm{p'}$ is the outcome of computational basis measurement to the second copy, and $a_{ij}$ are the elements of the sampled binary (upper triangular) matrix $\dyadic{A}$ (see Algo.~\ref{main:algo1}). In the GHZ case, even though the measurement outcome has an error, the estimator outputs the right value when the bit-flip error string $\xi_{\bm{p}}$ possesses an even Hamming-weight parity. In the Fig.~\ref{fig:cz_gadgetization_avg_fidelity} (b), we simulate the fidelity estimation of two pure GHZ states using noisy Clifford circuits. The closer the estimated fidelity is to $1$, the more tolerant the estimation is to bit-flip noise. We see that our scheme performs much better than randomized-Clifford shadow tomography.

    Second, if long-ranged CZ gates are available for our scheme, we can invoke the state-injection~\cite{catani2018} technique to reduce the Hamming weight of the bit-flip error string $\xi_{\bm{p}}$. To this end, with no loss of generality, we may suppose that the realized CZ gates are corrupted by Pauli noise, which can be obtained from twirling on an arbitrary noise channel that is not of this type~\cite{emerson2007,wallman2016}. By assuming that long-ranged gate errors are much more significant, one may set up a two-storey teleportation-like circuit as shown in Fig.~\ref{fig:cz_gadgetization_avg_fidelity}(a). It eventually has an equivalent CZ-operation on the main board. Furthermore, the Pauli noise structure acting on the long-ranged CZ gate on the ancilla floor is turned into a dephasing noise. 
    This noise transformation is carried out using neighboring CNOT gates entangling the qubits between the two floors, and post-processing Z-operations. This state-injection is a linear channel and can be applied to all CZ gates. In effect, each CZ gate undergoes only dephasing noise (as a dominant noise), and such noise does not increase its weight by commuting with other CZ gates. Therefore, the weight of bit-flip errors at the final measurement part is significantly reduced.
    
    This method will give additional noise tolerance to some symmetric states~\cite{grier2022}, given that the number of computational basis components undergoing large-weighted dephasing error is low, and hence so is the deviation of the estimation. A typical example is the  W state~[\,$\ket{\mathrm{W}_n}\equiv\frac{1}{\sqrt{n}}(\ket{0}^{\otimes (n-1)}\ket{1}+\ket{0}^{\otimes (n-2)}\ket{1,0}+\ldots+\ket{1}\ket{0}^{\otimes (n-1)})$\,]~\cite{flammia2011} whose estimator reads,

     \begin{align}
        \widehat{O}=\frac{\left\{\sum_{i=1}^n(-1)^{(a_i+p_i)}\right\}^2}{2n}+|\braket{{\rm W}_n|\bm{p'}}|^2-\frac{1}{2}.
    \end{align}

    We can see the first term gets a larger deviation as the weight of $\xi_{\bm{p}}$ gets larger. In Fig.~\ref{fig:cz_gadgetization_avg_fidelity} (c), we gave the initial depolarizing noise to the long-ranged CZ gates, where the dephasing error rate after the state-injection arrangement is equal to that of the initial noise rate on the ancillary CZ~gate. We can see that the state-injection technique can give an additional improvement in the noise threshold to the original RESPOVM circuit.

        The CZ-circuit-based architecture of our (R)ESPOVM schemes also presents an attractive feature for experimental implementation, especially with the soaring entangling-gate fidelities recently achieved in various quantum-computing platforms. While a more comprehensive survey shall be given in Sec.~\ref{subsec:survey}, we quote some recent experimental fidelities achieved in some of these platforms: 99.5\%~two-qubit gate fidelity in Rydberg-atom quantum computing~\cite{Evered:2023high-fidelity}, over 99.5\% for semiconductor-spin-based systems~\cite{Xue:2022quantum}, and up to 99.93\% with superconducting qubits~\cite{Negirneac:2021high-fidelity}.
        
        \section{Discussion}
		
		In this article, we propose resource-efficient methods to carry out quantum shadow tomography using equatorial-stabilizer measurements that are derived from smaller subsets of the full set of stabilizer states. We first prove that these equatorial-stabilizer measurements, together with the computational-basis measurements are indeed informationally complete and possess desirable sampling complexity that grows, in many cases, at most polynomially in the number of qubits~$n$ and reciprocally in the square of the target additive estimation error. In typical scenarios, such as rank-one observable estimation, the sampling complexity of equatorial-stabilizer shadow tomography is independent of~$n$.
		
		Furthermore, our shadow-tomography algorithms utilize long-ranged Clifford bases with circuit depths of~$n$. When restricted to only LNN architectures, our algorithms are considerably efficient to implement, requiring only $2n$-depth NN CNOT~gates, owing to their simpler circuit geometry. We demonstrated that our methods require shorter circuit depth than randomized Clifford tomography for both long-ranged and LNN implementations, while still preserving the classes of applicable observables for many cases in terms of the shadow-norm or estimation variance criterion. Our depth- (or gate-)efficient shadow-tomography schemes become practically important for performing realistic quantum-computation tasks as they permit larger gate-error thresholds. Moreover, this is crucial for certain auxiliary quantum tasks, such as those in quantum error mitigation~\cite{jnane2024}, for instance, where it is recently known~\cite{tsubouchi2023} that such tasks would require sampling-copy numbers that increase exponentially with the circuit depth. We also found that along with the state injection technique, the estimators for (R)ESPOVM tomography reveal a unique tolerance over the various measurement bit-flip errors. Since state injection simplifies the CZ-gate-noise structure, post-processing noise mitigation of shadow~\cite{koh2022} could be applied.
		
		Thus far, the (R)ESPOVM schemes we have found to have good shadow-tomographic properties require both computational-basis and CZ-circuit measurements that involve very different circuit depths, with the former being single-layered while the latter grows linearly with~$n$ in depth. One possible future would be to investigate if there exist other such pairs of measurement circuits of comparable depths that can reduce the maximum circuit depth employed whilst maintaining the same estimation quality. Here, MUB subsets of the CZ-circuit measurements~\cite{zhang2024,wang2024} could serve as an instructive starting point. Yet another desirable observation is the rather large error-threshold advantages of using equatorial-stabilizer shadow tomography on GHZ or W state-verification tasks in contrast to randomized-Clifford shadow tomography. This readily sets another research direction on fully identifying important classes of target quantum states that could lead to highly effective gate-error threshold improvements.Finally, motivated by a recent work~\cite{rozon2024}, we expect that the shallow depth subset of noisy (R)ESPOVM would give another sample-optimal Pauli observable estimation method. Then we could measure the noise threshold of sampling copy improvement, which was previously considered on noisy random Clifford blocks~\cite{rozon2024}, over the single-depth measurement.
		
		Clifford circuits are well-known quantum circuits~\cite{shi2022,xie2023} whose noise-tolerant properties could be realizable in the near future. Therefore, we expect our work to establish practical research directions on gate-efficient observable-expectation-value estimations for noisy intermediate-scale quantum computation.\\[2ex]

        \noindent
        \textbf{*Note:} After the completion, Ref.~\cite{schuster2025} appeared on the arXiv discussing the possibility of approximating unitary designs with logarithmic-depth circuits. We have included an updated circuit-depth comparison between the algorithm in Ref.~\cite{schuster2025} and our (R)ESPOVM~schemes. We show, in Fig.~\ref{fig:gatecounts}(b) that despite the logarithmic scaling, our scheme employs relatively shorter circuit depths for moderate qubit numbers beyond~100.
        \section*{Acknowledgements}
		This work is supported by the National Research Foundation of Korea (NRF) grants funded by the Korean government (MSIT) (Grant Nos. NRF-2023R1A2C1006115, RS-2023-00237959, RS-2024-00413957, RS-2024-00437191, RS-2024-00438415, RS-2025-02219034), the Institute of Information Communications Technology Planning Evaluation (IITP) grant funded by the Korea government (MSIT) (IITP-2025-RS-2020-II201606), the Institute of Applied Physics at Seoul National University, and the Brain Korea 21 FOUR Project grant funded by the Korean Ministry of Education.
		
		\section*{Data Availability}
		All data that support the plots within this
		paper and other findings of this study are available from the corresponding authors upon
		reasonable request.
    
		
	\appendix

		\section{Algorithm details for equatorial-stabilizer shadow tomography}\label{sec:technical}
		
		Before we present the details of (R)ESPOVM shadow tomography, we first introduce the notation $\mathbb{E}_{\sigma \sim\mathcal{D}(A)}f(\sigma)$ that refers to the expected value of the function $f(\sigma)$ of a random variable $\sigma\in A$ that follows a uniformly random distribution~$\mathcal{D}(A)$ with domain~$A$. In particular, we call $\mathbb{E}_{\sigma \sim\mathcal{D}(A)}\sigma^{\otimes t}$ as the \emph{$t$-th moment of $A$}. With this, we have the following lemma~\cite{supple}.
       
		\begin{lemma}\label{main:lem1}
			For the $n$-qubit system,
            
			\begin{align}
				&{\rm(i)}~\text{\cite{bravyi2019}}~\momentop\nonumber\\&=\eqmomentop=\frac{I}{2^n},\\
				&{\rm(ii)}~\text{\cite{bravyi2019}}~\momentop^{\otimes 2}\nonumber\\&=\frac{1}{4^n}\left(I+\sum_{\bm{x},\bm{y}\in \bfZ^n_2}\ket{\bm{x}\bm{y}}\bra{\bm{y}\bm{x}}-\sum_{\bm{x}\in\bfZ^n_2}\ket{\bm{x}\bm{x}}\bra{\bm{x}\bm{x}}\right),\\
				&{\rm(iii)}~\eqmomentop^{\otimes 2}\nonumber\\&=\frac{1}{4^n}\left(I+\sum_{\bm{x},\bm{y}\in \bfZ^n_2}\ket{\bm{x}\bm{y}}\bra{\bm{y}\bm{x}}+\sum_{\bm{x},\bm{y}\in \bfZ^n_2}\ket{\bm{x}\bm{x}}\bra{\bm{y}\bm{y}}\nonumber\right.\\&\left.-2\sum_{\bm{x}\in\bfZ^n_2}\ket{\bm{x}\bm{x}}\bra{\bm{x}\bm{x}}\right).
			\end{align}
            
		\end{lemma}
        \begin{proof}
            	The first and second moments of the equatorial-stabilizer set are shown in Ref.~\cite{bravyi2019}. So, we only consider the RESPOVM case and adopt similar strategies for ESPOVMs. All bold-faced italicized symbols refer to columns~(synonymous with ``bit strings'') and non-italicized uppercase ones that are also bold-faced refer to matrices. Note that for $\forall \bm{p}\in \mathbb{Z}^2_n$,
		\begin{align}
			&Z^{\bm{p}}\eqmomentop Z^{\bm{p}}\nonumber\\
			=&\,\eqmoment\frac{1}{2^n}\!\!\sum_{\bm{x},\bm{y}\in \mathbb{Z}^n_2} (-1)^{\bm{x}^\top\dyadic{A}\bm{x}+\bm{y}^\top\dyadic{A}\bm{y}+\bm{p}\cdot(\bm{x}+\bm{y})}\ket{\bm{x}}\bra{\bm{y}}\nonumber\\
			=&\,\eqmoment\frac{1}{2^n}\!\!\sum_{\bm{x},\bm{y}\in \mathbb{Z}^n_2} (-1)^{\bm{x}^\top\dyadic{A}'\bm{x}+\bm{y}^\top\dyadic{A}'\bm{y}}\ket{\bm{x}}\bra{\bm{y}}\nonumber\\
			=&\,\eqmoment\ket{\phi_{\dyadic{A}'}}\bra{\phi_{\dyadic{A}'}}\nonumber\\
			=&\,\eqmoment\ket{\phi_\dyadic{A}}\bra{\phi_\dyadic{A}},
		\end{align} 
		where for $\forall i,j\in[n]$, $A'_{ii}=A_{ii}+p_i({\rm mod\;2})$ otherwise $A'_{ij}=A_{ij}$. The third equality holds because the reparametrization of $\dyadic{A}$ does not affect its uniform-distribution first moment. Thus $\eqmomentop$ commutes with all Pauli $Z$-operators and is diagonal in the computational basis. Notice that all diagonal elements are $\frac{1}{2^n}$. Therefore, $\eqmoment\ket{\phi_\dyadic{A}}\bra{\phi_\dyadic{A}}=\frac{I}{2^n}$, which is Lem.~\ref{main:lem1}~(i).
		
		Next, note that 
		\begin{align}
			\label{eq:reqsecond}
			&\eqmomentop^{\otimes 2}\nonumber\\
			=&\,\eqmoment\frac{1}{4^n}\nonumber\\
			&\times\sum_{\bm{x},\bm{y},\bm{z},\bm{w}\in\bfZ^n_2}(-1)^{\bm{x}^\top\dyadic{A}\bm{x}+\bm{y}^\top\dyadic{A}\bm{y}+\bm{z}^\top\dyadic{A}\bm{z}+\bm{w}^\top\dyadic{A}\bm{w}}\ket{\bm{x}\bm{y}}\bra{\bm{z}\bm{w}},\nonumber\\
			=&\,\frac{1}{4^n}\sum_{\bm{x},\bm{y},\bm{z},\bm{w}\in\bfZ^n_2}\eqmoment\Big\{\nonumber\\
			&\times(-1)^{\bm{x}^\top\dyadic{A}\bm{x}+\bm{y}^\top\dyadic{A}\bm{y}+\bm{z}^\top\dyadic{A}\bm{z}+\bm{w}^\top\dyadic{A}\bm{w}}\Big\}\ket{\bm{x}\bm{y}}\bra{\bm{z}\bm{w}}.
		\end{align}
		Now, we focus on terms in the curly brackets. Since the elements of $\dyadic{A}$ oscillate between $0$ and $1$, the nonzero coefficients of $\ket{\bm{x}\bm{y}}\bra{\bm{z}\bm{w}}$ must satisfy the following equations for all $p,q\in[n](p<q)$, 
		\begin{align}\label{lem1:sys1}
			\begin{cases}x_p+y_p+z_p=\,w_p,\\	x_px_q+y_py_q+z_pz_q+w_pw_q=\,0.
			\end{cases}
		\end{align}
		Now, substitute the upper equation to lower one. We obtain that 
		\begin{equation}
			x_p(y_q+z_q)+y_p(z_q+x_q)+z_p(x_q+y_q)=0.
		\end{equation}
		If $\bm{x}=\bm{y}=\bm{z}$, the above equation must hold. Suppose $\bm{x},\bm{y},\bm{z}$ are not all the same. Then there exists $p'\in[n]$ such that two of $x_{p'},y_{p'},z_{p'}$ are equal but the rest are not. Also, $x_i=y_i=z_i$ for $i\le p'$. WLOG, we may assume $x_{p'}=y_{p'}\ne z_{p'}$, then for  $\forall q> p'$,
		\begin{equation}
			(x_q+y_q)(x_{p'}+z_{p'})=0 \;
			\rightarrow x_q=y_q\;\because\;x_{p'}\ne z_{p'}. 
		\end{equation}
		It follows that nonzero coefficients of $\ket{\bm{x}\bm{y}}\bra{\bm{z}\bm{w}}$ must be such that two of $\bm{x},\bm{y},\bm{z}$ are equal. 
		Based on this observation, \eqref{lem1:sys1} implies several cases: (i) $\bm{x}=\bm{y},\bm{z}=\bm{w}$, (ii) $\bm{x}=\bm{z},\bm{y}=\bm{w}$, and (iii) $\bm{x}=\bm{w},\bm{y}=\bm{z}$. In all these cases $\eqmoment\left\{(-1)^{\bm{x}^\top\dyadic{A}\bm{x}+\bm{y}^\top\dyadic{A}\bm{y}+\bm{z}^\top\dyadic{A}\bm{z}+\bm{w}^\top\dyadic{A}\bm{w}}\right\}=1$.\\
		Hence, we obtain the following results.
		\begin{align}
			&\,\eqmomentop^{\otimes 2}\nonumber\\
			=&\,\frac{1}{4^n}\Bigg(I+\sum_{\bm{x},\bm{y}\in \bfZ^N_2}\ket{\bm{x}\bm{y}}\bra{\bm{y}\bm{x}}+\sum_{\bm{x}\bm{y}\in \bfZ^N_2}\ket{\bm{x}\bm{x}}\bra{\bm{y}\bm{y}}\nonumber\\
			&-2\sum_{\bm{x}\in\bfZ^N_2}\ket{\bm{x}\bm{x}}\bra{\bm{x}\bm{x}}\Bigg),
		\end{align}
		which is the result of Lem.~\ref{main:lem1} (iii) and completes the proof.
        \end{proof}
		We note that from Lem.~\ref{main:lem1}~(i), all elements from the respective sets $\left\{\frac{2^n}{\left|\mathcal{S}^{\rm eq}_n\right|}\ket{\phi^{{\rm eq}}_\dyadic{A}}\bra{\phi^{{\rm eq}}_\dyadic{A}}\right\}$ and $\left\{\frac{2^n}{\left|\mathcal{S}^{\rm req}_n\right|}\ket{\phi^{\rm req}_\dyadic{A}}\bra{\phi^{\rm req}_\dyadic{A}}\right\}$ form a POVM. 
        
        Using Lem.~\ref{main:lem1}~(ii) and (iii), let us prove Eq.~\eqref{eq:equatorial_tomography} and Eq.~\eqref{eq:real_equatorial_tomography} which are directly related to our tomography schemes.
      
        First, we consider ESPOVMs. By~Lem.~\ref{main:lem1}, 
            
		\begin{align}
			&\mathbb{E}_{\ket{\phi^{{\rm eq}}_\dyadic{A}}\sim\mathcal{D}\left(\mathcal{S}^{{\rm eq}}_n\right)}\left(\bra{\phi^{{\rm eq}}_\dyadic{A}}\sigma\ket{\phi^{{\rm eq}}_\dyadic{A}}\bra{\phi^{{\rm eq}}_\dyadic{A}}O\ket{\phi^{{\rm eq}}_{\dyadic{A}}}\right)\nonumber\\&=\,{\rm tr}\left\{\left(\momentop^{\otimes 2}\right)(\sigma\otimes O)\right\}\nonumber\\
			&=\,\frac{1}{4^n}{\rm tr}\Bigg\{\Bigg(I+\sum_{\bm{x},\bm{y}\in \bfZ^n_2}\ket{\bm{x}\bm{y}}\bra{\bm{y}\bm{x}}-\sum_{\bm{x}\in\bfZ^n_2}\ket{\bm{x}\bm{x}}\bra{\bm{x}\bm{x}}\Bigg)\nonumber\\&(\sigma\otimes O)\Bigg\}\nonumber,\\
			&=\,\frac{1}{4^n}\Bigg\{{\rm tr}(O)+\sum_{\bm{x},\bm{y}\in\bfZ^n_2}\bra{\bm{y}}\sigma\ket{\bm{x}}\bra{\bm{x}}O\ket{\bm{y}}\nonumber\\&-\sum_{\bm{x}\in\bfZ^n_2}\bra{\bm{x}}\sigma\ket{\bm{x}}\bra{\bm{x}}O\ket{\bm{x}}\Bigg\}\nonumber,\\
			&=\,\frac{1}{4^n}\left\{{\rm tr}(O)+{\rm tr}(O\sigma)-\sum_{\bm{x}\in\bfZ^n_2}\bra{\bm{x}}\sigma\ket{\bm{x}}\bra{\bm{x}}O\ket{\bm{x}}\right\}. 
		\end{align}
		Hence, we obtain the following result, 
		\begin{align}
			{\rm tr}(O\rho)=&\,\sum_{\ket{\phi^{{\rm eq}}_\dyadic{A}}\in \mathcal{S}^{{\rm eq}}_n}\frac{2^{2n}}{\left|\mathcal{S}^{{\rm eq}}_n\right|}\bra{\phi^{{\rm eq}}_{\dyadic{A}}}\rho\ket{\phi^{{\rm eq}}_{\dyadic{A}}}\bra{\phi^{{\rm eq}}_\dyadic{A}}O\ket{\phi^{{\rm eq}}_{\dyadic{A}}}\nonumber\\&+\sum_{\bm{x}\in\bfZ^n_2}\bra{\bm{x}}\rho\ket{\bm{x}}\bra{\bm{x}}O\ket{\bm{x}}-{\rm tr}(O).\label{eq:sub_equatorial_tomography}
		\end{align}
      
		Next, we consider RESPOVMs and also assume $O$ to be a real observable, that is $\bra{\bm{a}}O\ket{\bm{b}}\in\mathbb{R}$ for $\forall \bm{a},\bm{b}\in\bfZ^n_2$. Then, 
        
		\begin{align}
			&\mathbb{E}_{\ket{\phi^{{\rm req}}_\dyadic{A}}\sim\mathcal{D}\left(\mathcal{S}^{{\rm req}}_n\right)}\left(\bra{\phi^{{\rm req}}_\dyadic{A}}\sigma\ket{\phi^{{\rm req}}_\dyadic{A}}\bra{\phi^{req}_\dyadic{A}}O\ket{\phi^{{\rm req}}_{\dyadic{A}}}\right)\nonumber\\
			&=\,{\rm tr}\left\{\left(\eqmomentop^{\otimes 2}\right)(\sigma\otimes O)\right\}\nonumber\\
			&=\,\frac{1}{4^n}{\rm tr}\Bigg\{\Bigg(I+\sum_{\bm{x},\bm{y}\in \bfZ^n_2}\ket{\bm{x}\bm{y}}\bra{\bm{y}\bm{x}}+\sum_{\bm{x},\bm{y}\in \bfZ^n_2}\ket{\bm{x}\bm{x}}\bra{\bm{y}\bm{y}}\nonumber\\
			&\qquad\qquad\qquad\qquad-2\sum_{\bm{x}\in\bfZ^n_2}\ket{\bm{x}\bm{x}}\bra{\bm{x}\bm{x}}\Bigg)(\sigma\otimes O)\Bigg\}\nonumber\\&=\,\frac{1}{4^n}\Bigg\{{\rm tr}(O)+\sum_{\bm{x},\bm{y}\in\bfZ^n_2}\bra{\bm{y}}\sigma\ket{\bm{x}}\bra{\bm{x}}O\ket{\bm{y}}\nonumber\\&+\sum_{\bm{x},\bm{y}\in\bfZ^N_2}\!\!\bra{\bm{y}}\sigma\ket{\bm{x}}\!\bra{\bm{y}}O\ket{\bm{x}}-2\!\sum_{\bm{x}\in\bfZ^n_2}\!\!\bra{\bm{x}}\sigma\ket{\bm{x}}\!\bra{\bm{x}}O\ket{\bm{x}}\Bigg\}\nonumber\\&=\,\frac{1}{4^n}\Bigg\{{\rm tr}(O)+2\sum_{\bm{x},\bm{y}\in\bfZ^n_2}\bra{\bm{y}}\sigma\ket{\bm{x}}\bra{\bm{x}}O\ket{\bm{y}}\nonumber\\&\qquad\qquad\qquad\qquad-2\sum_{\bm{x}\in\bfZ^n_2}\bra{\bm{x}}\sigma\ket{\bm{x}}\bra{\bm{x}}O\ket{\bm{x}}\Bigg\}\label{eq:RESPOVM_lemma2_sub}
            \end{align}
            \begin{align}
    			&=\,\frac{1}{4^n}\left\{{\rm tr}(O)+2{\rm tr}(O\sigma)-2\!\sum_{\bm{x}\in\bfZ^n_2}\!\!\bra{\bm{x}}\sigma\ket{\bm{x}}\!\bra{\bm{x}}O\ket{\bm{x}}\right\}.
    			\label{eq:RESPOVM_lemma2}
		\end{align}
      
		This is equivalent to, 
		\begin{align}
			{\rm tr}(O\rho)=&\,\sum_{\ket{\phi^{{\rm req}}_\dyadic{A}}\in \mathcal{S}^{{\rm req}}_n}\frac{2^{2n-1}}{\left|\mathcal{S}^{{\rm req}}_n\right|}\bra{\phi^{{\rm req}}_{\dyadic{A}}}\rho\ket{\phi^{{\rm req}}_{\dyadic{A}}}\bra{\phi^{{\rm req}}_{\dyadic{A}}}O\ket{\phi^{{\rm req}}_{\dyadic{A}}}\nonumber\\
			&+\sum_{\bm{x}\in\bfZ^n_2}\bra{\bm{x}}\rho\ket{\bm{x}}\bra{\bm{x}}O\ket{\bm{x}}-\frac{1}{2}{\rm tr}(O).\label{eq:sub_real_equatorial_tomography}
		\end{align}
		Equation~\eqref{eq:RESPOVM_lemma2_sub} is valid when $O$ is real. This proves Eq.~\eqref{eq:equatorial_tomography} and Eq.~\eqref{eq:real_equatorial_tomography}. In Eq.~\eqref{main:eq_vareq} of the main text, we ignore the ${\rm tr}(O)$ term since it is constant and, hence, does not affect the variance.

        Let us give some remarks before proceeding further. An arbitrary Hermitian input density operator~$\rho$ can be expressed as $\frac{1}{2^n}\sum_{P\in \mathcal{P}_n}\tr{\rho P}P$. That is, setting $O$ as Pauli operators, we note that $\rho$ takes the equivalent form
        \begin{align}
            \rho=&\,\sum_{\ket{\phi^{{\rm eq}}_\dyadic{A}}\in \mathcal{S}^{{\rm eq}}_n}\frac{2^{2n}}{\left|\mathcal{S}^{{\rm eq}}_n\right|}\bra{\phi^{{\rm eq}}_{\dyadic{A}}}\rho\ket{\phi^{{\rm eq}}_{\dyadic{A}}}\ket{\phi^{{\rm eq}}_\dyadic{A}}\bra{\phi^{{\rm eq}}_{\dyadic{A}}}\nonumber\\&+\sum_{\bm{x}\in\bfZ^n_2}\bra{\bm{x}}\rho\ket{\bm{x}}\ket{\bm{x}}\bra{\bm{x}}-I.
        \end{align}
        Using the POVM properties,  $\sum_{\dyadic{A}}\frac{2^n}{\left|\mathcal{S}^{\rm eq}_n\right|}\braket{\phi^{{\rm eq}}_\dyadic{A}|\rho|\phi^{{\rm eq}}_\dyadic{A}}=\sum_{\bm{x}}\braket{\bm{x}|\rho|\bm{x}}=1$, we note that $\mathcal{M}^{-1}\circ\mathcal{M}(\rho)=\rho$ following the notations in Eqs.~\eqref{eq:measurement_channel} and \eqref{eq:classical_shadow}. Therefore, we conclude that Eq.~\eqref{eq:classical_shadow} is indeed the inverse of $\mathcal{M}$. Since only the real part of $\rho$ is expanded by only real Pauli basis~operators, we can certify the RESPOVM inverse of Eq.~\eqref{eq:classical_shadow} in the same manner.

		From the Lem.~\ref{main:lem1}, Eq.~\eqref{eq:sub_equatorial_tomography}, and Eq.~\eqref{eq:sub_real_equatorial_tomography}, the following estimators (denoted with a caret) are constructed:
	
		\begin{equation}
			\widehat{\left<O\right>}=\begin{cases}
				\dfrac{2^{n}}{(N/2)}\mathlarger{\mathlarger{\sum}}_{\substack{\text{sampled $\dyadic{A}$}\\ \text{over $N/2$ copies}}}\bra{\phi^{{\rm eq}}_{\dyadic{A}}}O\ket{\phi^{{\rm eq}}_{\dyadic{A}}}\\+\dfrac{1}{(N/2)}\mathlarger{\mathlarger{\sum}}_{\substack{\text{measured $\bm{p}'$}\\ \text{over $N/2$ copies}}}\bra{\bm{p}'}O\ket{\bm{p}'}-{\rm tr}(O) &\\ (\text{ESPOVM and comp. basis})\,,\\[6ex]
				\dfrac{2^{n-1}}{(N/2)}\mathlarger{\mathlarger{\sum}}_{\substack{\text{sampled $\dyadic{A}$}\\ \text{over $N/2$ copies}}}\bra{\phi^{{\rm req}}_{\dyadic{A}}}O\ket{\phi^{{\rm req}}_{\dyadic{A}}}\\+\dfrac{1}{(N/2)}\mathlarger{\mathlarger{\sum}}_{\substack{\text{measured $\bm{p}'$}\\ \text{over $N/2$ copies}}}\bra{\bm{p}'}O\ket{\bm{p}'}-\dfrac{1}{2}{\rm tr}(O)& \\(\text{RESPOVM and comp. basis})\,,
			\end{cases}
			\label{eq:Oest}
		\end{equation}

		where $\dyadic{A}$ and $\bm{p}'$ are sampled from the distribution of $\left\{\frac{2^n}{|\mathcal{S}^{{\rm (r)eq}}_n|}\bra{\phi^{{\rm (r)eq}}_{\dyadic{A}}}\rho\ket{\phi^{{\rm (r)eq}}_{\dyadic{A}}}\right\}$ and $\left\{\braket{\bm{p}'|\rho|\bm{p}'}\right\}$ respectively. 
		Note that these estimators are unbiased, i.e. data-average $\mathbb{E}(\widehat{\left<O\right>})=\left<O\right>$.
		
		We are ready to provide more details concerning the (R)ESPOVM shadow-tomography procedure, thereby explaining the validity of Algo.~\ref{main:algo1}. We shall only discuss matters related to the ESPOVM, as the dissertation for RESPOVM follows similarly. Let us assume an input state $\rho=\sum_{\bm{z},\bm{w}\in\bfZ^n_2}\ket{\bm{z}}\bra{\bm{w}}\bra{\bm{z}}\rho\ket{\bm{w}}$. The procedure begins by a random selection of $\ket{\phi^{\rm eq}_{\dyadic{A}}}\in\mathcal{S}^{{\rm eq}}_n$, that is the matrix $\dyadic{A}$ from uniformly distributed elements in the set
      
		\begin{align}
			\left\{\dyadic{A}\bigg|\dyadic{A}=(a_{ij})_{i,j\in [n]},\;a_{ij}=a_{ji},\;a_{ii}\in\bfZ_4,\;a_{ij(i\ne j)}\in\bfZ_2\right\}\,.
			\label{eq:set_A}
		\end{align}
		Operationally, this is akin to flipping a coin $\frac{n^2-n}{2}$ times to determine all symmetric off-diagonal elements and flipping a uniform 4-sided die $n$ times to determine all diagonal elements of $\dyadic{A}$. Now, given a single-qubit~(two-qubit resp.) operation $K$, we define $K_{i}\;(K_{ij})$ as $K$ acting on the $i$-th\;($i$-th and $j$-th) qubit. After choosing a random $\dyadic{A}$, we perform $\mathrm{CZ}_{ij}$ gate operations whenever $a_{ij}=1$ and $i<j$, followed by single-qubit operations $U_i$ such that 
		\begin{align}
			U_i\equiv
			\begin{cases}
				I_i & {\rm if}\;a_{ii}=0,\\
				S^{\dagger} & {\rm if}\;a_{ii}=1,\\
				Z_i & {\rm if}\;a_{ii}=2,\\
				S_i & {\rm if}\;a_{ii}=3,
			\end{cases}	\quad\text{or simply},\;U_{i}\equiv (S^{\dagger})^{a_{ii}}. 
		\end{align}
		Next, we apply Hadamard gates $H\equiv\bigotimes_{i=1}^{n}H_i$. The subsequent evolved state $\rho'$ is written as
        
		\begin{align}
			&\rho'=\nonumber\\&\frac{1}{2^{n}}\sum_{ \bm{z},\bm{w},\bm{s},\bm{t}\in\bfZ^n_2}\I^{3\bm{z}^\top\dyadic{A}\bm{z}+\bm{w}^\top\dyadic{A}\bm{w}}(-1)^{\bm{z}\bm{\cdot} \bm{s}+\bm{w}\bm{\cdot} \bm{t} }\ket{\bm{s}}\bra{\bm{t}}\bra{\bm{z}}\rho\ket{\bm{w}}.
		\end{align}
		Here, $\bm{a}\bm{\cdot}\bm{b}$ refers to the binary inner product between $\bm{a}$ and $\bm{b}$. We shall now perform the Pauli $Z$-basis measurement independently on all the qubits. The probability $\mathrm{prob}(\bm{p}\mid A)$ to obtain a specific binary string $\bm{p}\in\bfZ^n_2$ is
        \begin{widetext}
		\begin{align}
			\mathrm{prob}(\bm{p}\mid \dyadic{A})\equiv \bra{\bm{p}}\rho'\ket{\bm{p}}
			&=\frac{1}{2^{n}}\sum_{ \bm{z},\bm{w},\bm{s},\bm{t}\in\bfZ^n_2}\I^{3\bm{z}^\top\dyadic{A}\bm{z}+\bm{w}^\top\dyadic{A}\bm{w}}(-1)^{\bm{z}\bm{\cdot} \bm{s}+\bm{w}\bm{\cdot} \bm{t} }\braket{\bm{p}\mid \bm{s}}\braket{\bm{t}\mid \bm{p}}\bra{\bm{z}}\sigma\ket{\bm{w}},\nonumber\\
			&=\frac{1}{2^{n}}\sum_{ \bm{z},\bm{w}\in\bfZ^n_2}\I^{3\bm{z}^\top\dyadic{A}\bm{z}+\bm{w}^\top\dyadic{A}\bm{w}}(-1)^{\bm{z}\bm{\cdot} \bm{p}+\bm{w}\bm{\cdot} \bm{p} }\bra{\bm{z}}\sigma\ket{\bm{w}},\nonumber\\
			&=\frac{1}{2^{n}}\sum_{ \bm{z},\bm{w}\in\bfZ^n_2}(-\I)^{\bm{z}^\top\dyadic{A}\bm{z}-2\bm{z}\bm{\cdot} \bm{p}}\,\I^{\bm{w}^\top\dyadic{A}\bm{w}-2\bm{w}\bm{\cdot} \bm{p}}\bra{\bm{z}}\sigma\ket{\bm{w}},\label{eq:prob}\\
			&=\bra{\phi^{\rm eq}_{\dyadic{A}'}}\rho\ket{\phi^{\rm eq}_{\dyadic{A}'}}, 
		\end{align}
        \end{widetext}
		where $\dyadic{A}'$ is an $n\times n$ matrix such that $a'_{ij}=a_{ij}$ for $i\ne j$, and $a'_{ii}=a_{ii}+2p_{i}\;({\rm mod}\;4)$. The last equation can be obtained by following the definition of Eq.~\eqref{main:eq1}. Now, we define $\dyadic{D}(\bm{p})$ as the $n\times n$ diagonal matrix whose $i$-th diagonal element is equal to $p_i$, and also note, from Eq.~\eqref{main:eq1}, that
		\begin{align}
			&\frac{2^n}{\left|\mathcal{S}^{\rm eq}_n\right|}\bra{\phi^{\rm eq}_\dyadic{A}}\rho\ket{\phi^{\rm eq}_\dyadic{A}}\nonumber\\&=\frac{1}{\left|\mathcal{S}^{\rm eq}_n\right|}\sum_{ \bm{z},\bm{w}\in\bfZ^n_2}(-\I)^{\bm{z}^\top\dyadic{A}\bm{z}}\,\I^{\bm{w}^\top\dyadic{A}\bm{w}}\bra{\bm{z}}\rho\ket{\bm{w}}\nonumber\\&=\sum_{\bm{p}\in\bfZ^n_2}\frac{1}{\left|\mathcal{S}^{\rm eq}_n\right|}\mathrm{prob}(\bm{p}\mid \dyadic{A}-2\dyadic{D}(\bm{p}))\,,\label{eq:prob2} 
		\end{align}
		where the final equality is a result of Eq.~\eqref{eq:prob}.

		The left-hand side of Eq.~\eqref{eq:prob2} corresponds to the probability distribution of $\dyadic{A}$ we want to sample from. Also, given a $\bm{p}\in\bfZ^n_2$, the probability to sample the matrix $\dyadic{A}-2\dyadic{D}(\bm{p})$ is constant~$\left(\frac{1}{\left|\mathcal{S}^{\rm eq}_n\right|}\right)$. Hence, we can interpret the right-most side of Eq.~\eqref{eq:prob2} as the probability to get the matrix $\dyadic{A}$ after we uniformly sampled the matrix $\dyadic{A}'$ and add $2\dyadic{D}(\bm{p})$ (see steps~\ref{algo1:start_for} through \ref{algo1:end_for} in Algo.~\ref{main:algo1}). This way, the desired sample of~$\dyadic{A}$ is obtained. 
		
		Furthermore, we note that one can propagate the measurement section to the left-over Hadamard gates and phase gates, becoming $X$ or $Y$-basis measurements. In that case, the CZ gates are the only remaining intermediate gates, and Vizing's theorem states that only at most $n$~layers of CZ~gates are needed to realize an arbitrary~$n$-qubit CZ circuit, which settles part of Thm.~\ref{main:thm0}. One must still carry out either the Pauli $X$ or $Y$ measurement depending on the phase gate acting on that qubit. However, the phase gate implementation is unnecessary and we may simply perform an $X$ or $Y$-basis measurement. The reason is that $Z$~($S$ resp.) just flips the measurement outcome of $X$~($Y$)~measurement. For example, implementing $Z$ and obtaining the $X$-measurement outcome~$0$ is equivalent to just measuring in the $X$-basis, obtaining the outcome~$1$, which gives the same diagonal value of the sampled $\dyadic{A}$. This closes the proof of Thm.~\ref{main:thm0} for the ESPOVM scheme.
		
		Harking back to Eq.~\eqref{eq:sub_equatorial_tomography}, in order to complete the shadow tomography protocol, one also acquires $\bm{p}'\in\bfZ^n_2$ following the distribution $\bra{\bm{p}'}\rho\ket{\bm{p}'}$, which can be done easily by measuring $\rho$ in the Pauli~$Z$ basis. Hence, the algorithm for an unbiased estimation of ${\rm tr}(O\rho)$ can be summarized as follows: 
		
		\begin{enumerate}
			\item Measure one copy of $\rho$ with the projector $\ket{\phi^{\rm eq}_\dyadic{A}}\bra{\phi^{\rm eq}_\dyadic{A}}$ defined by a uniformly-chosen $\dyadic{A}$ from the set $\mathcal{A}$ in \eqref{eq:set_A}. This is done by (i)~uniformly choosing an $\dyadic{A}'\in\mathcal{A}$, (ii)~obtain a $\bm{p}\in\mathbb{Z}^n_2$ from measuring this $\rho$~copy with $\ket{\phi^{\rm eq}_{\dyadic{A}'}}\bra{\phi^{\rm eq}_{\dyadic{A}'}}$, (iii) and take $\dyadic{A}=\dyadic{A}'+2\dyadic{D}(\bm{p})$. The measurement outcome $\dyadic{A}$ would then follow the distribution $\frac{2^n}{\left|\mathcal{S}^{\rm eq}_n\right|}\bra{\phi^{\rm eq}_\dyadic{A}}\rho\ket{\phi^{\rm eq}_\dyadic{A}}$.
			\item Measure one copy of $\rho$ and obtain $\bm{p}'\in\bfZ^n_2$ following the distribution $\bra{\bm{p}'}\rho\ket{\bm{p}'}$.
			\item Calculate $\widehat{O}_{\dyadic{A}.\bm{p'}}\equiv 2^n\bra{\phi^{\rm eq}_\dyadic{A}}O\ket{\phi^{\rm eq}_\dyadic{A}}+\bra{\bm{p}'}O\ket{\bm{p}'}-{\rm tr}(O)$. 
		\end{enumerate}
		Algorithm~\ref{main:algo1} is the pseudocode encapsulation of all the arguments we made thus far.
        
		\section{Estimation-variance analysis of equatorial-stabilizer shadow tomography}\label{subsec:est-var}

    We shall now sketch the proof of Thm.~\ref{main:thm1}~(i). 
    We start from Eq.~\eqref{main:eq_vareq} and note that 
		
		\begin{align}\label{method:var0}
			{\rm Var}(\widehat{O}^{{\rm eq}}_{\dyadic{A}})=&\,\mathbb{E}\left(\left\{2^n\bra{\phi^{{\rm eq}}_{\dyadic{A}}}O\ket{\phi^{{\rm eq}}_{\dyadic{A}}}-\mathbb{E}\left(2^n\bra{\phi^{{\rm eq}}_{\dyadic{A}}}O\ket{\phi^{{\rm eq}}_{\dyadic{A}}}\right)\right\}^2\right)\nonumber\\
			=&\,\mathbb{E}\Bigg(\Big\{2^n\bra{\phi^{{\rm eq}}_{\dyadic{A}}}(O_0+\frac{I}{2^n}\tr{O})\ket{\phi^{{\rm eq}}_{\dyadic{A}}}\nonumber\\
			&-\mathbb{E}\left(2^n\bra{\phi^{{\rm eq}}_{\dyadic{A}}}(O_0+\frac{I}{2^n}\tr{O})\ket{\phi^{{\rm eq}}_{\dyadic{A}}}\right)\Big\}^2\Bigg)\nonumber\\
			=&\,\mathbb{E}\left(\left\{2^n\!\bra{\phi^{{\rm eq}}_{\dyadic{A}}}O_0\ket{\phi^{{\rm eq}}_{\dyadic{A}}}-\mathbb{E}\left(2^n\!\bra{\phi^{{\rm eq}}_{\dyadic{A}}}O_0\ket{\phi^{{\rm eq}}_{\dyadic{A}}}\right)\right\}^2\right)\nonumber\\
			=&\,{\rm Var}(\widehat{O_0}^{{\rm eq}}_{\dyadic{A}})
		\end{align}
        
		where the second line is easily derived noting that $O_0$ is \emph{traceless part}, $O-\frac{I}{2^n}\tr{O}$, hence $\tr{O_0}=0$. Since we can follow the same reasoning to  ${\rm Var}(\widehat{O}^{{\rm req}}_{\dyadic{A}})$ and  ${\rm Var}(\widehat{O}^{{\rm bin}}_{\bm{p'}})$, we only need to consider the variance of the $O_0$. In other words, we obtain that
		
		\begin{align}
			{\rm Var}^{{\rm (r)eq}}(\widehat{O}^{{\rm (r)eq}}_\dyadic{A})=&\,{\rm Var}^{{\rm (r)eq}}(\widehat{O_0}^{{\rm (r)eq}}_\dyadic{A})\nonumber\\
			=&\,\mathbb{E}((\widehat{O_0}^{{\rm (r)eq}}_\dyadic{A})^2)-\mathbb{E}(\widehat{O_0}^{{\rm (r)eq}}_\dyadic{A})^2\nonumber\\
			\le&\,\mathbb{E}((\widehat{O_0}^{{\rm (r)eq}}_\dyadic{A})^2),
		\end{align}
		where $\mathbb{E}$ is the expected value of $\widehat{O_0}_\dyadic{A}\equiv 2^n\bra{\phi^{{\rm (r)eq}}_\dyadic{A}}O_0\ket{\phi^{{\rm (r)eq}}_\dyadic{A}}$. Hence, to obtain the upper bound of the estimation variance of $\widehat{O_0}^{{\rm (r)eq}}_\dyadic{A}$, all we need to do is to bound $\mathbb{E}((\widehat{O_0}^{{\rm (r)eq}}_\dyadic{A})^2)$ from above. 
		
		Eq.~\eqref{eq:exp_square} leads us to certify the third moment $\momentopp^{\otimes 3}$, and Lem.~\ref{method:lem3} is the answer. We assume the lemma holds (see Sec.~\ref{proof_lem1} for its proof). Then, we define five more finite series, $\mathcal{C}^{\rm(i)}\equiv\sum_{\bm{x}\in\bfZ^n_2}\ket{\bm{x}\bm{x}\bm{x}}\bra{\bm{xxx}}$, and,
		
		\begin{align}\label{method:common_string}
			\mathcal{C}^{\rm(ii)}\equiv&\sum_{\bm{x},\bm{y}\in\bfZ^n_2}\!\!\!\!\big(\ket{\bm{yxx}}\bra{\bm{yxx}}+\ket{\bm{xyx}}\bra{\bm{yxx}}+\ket{\bm{xxy}}\bra{\bm{yxx}}\nonumber\\
			&+\ket{\bm{yxx}}\bra{\bm{xyx}}+\ket{\bm{xyx}}\bra{\bm{xyx}}+\ket{\bm{xxy}}\bra{\bm{xyx}}\nonumber\\
			&+\ket{\bm{yxx}}\bra{\bm{xxy}}+\ket{\bm{xyx}}\bra{\bm{xxy}}+\ket{\bm{xxy}}\bra{\bm{xxy}}\big),\\
			\mathcal{C}^{\rm(iii)}\equiv&\sum_{\bm{x},\bm{y}\in\bfZ^n_2}\!\!\!\!\big(\ket{\bm{y}\bm{x}\bm{x}}\bra{\bm{yyy}}+\ket{\bm{xyx}}\bra{\bm{yyy}}+\ket{\bm{xxy}}\bra{\bm{yyy}}\nonumber\\
			&+\ket{\bm{yyy}}\bra{\bm{y}\bm{x}\bm{x}}+\ket{\bm{yyy}}\bra{\bm{xyx}}+\ket{\bm{yyy}}\bra{\bm{xxy}}\big),\\
            \mathcal{C}^{\rm(iv)}\equiv&\sum_{\bm{x},\bm{y},\bm{z}\in\bfZ^n_2}\!\!\!\!\!\!\big(\ket{\bm{xyz}}\bra{\bm{xyz}}+\ket{\bm{xyz}}\bra{\bm{zxy}}+\ket{\bm{xyz}}\bra{\bm{yzx}}\nonumber\\
			&+\ket{\bm{xyz}}\bra{\bm{yxz}}+\ket{\bm{xyz}}\bra{\bm{xzy}}+\ket{\bm{xyz}}\bra{\bm{zyx}}\big),
			\end{align}
            \begin{align}
            \mathcal{C}^{\rm(v)}\equiv&\sum_{\bm{x},\bm{y},\bm{z}\in\bfZ^n_2}\!\!\!\!\!\!\big(\ket{\bm{xxz}}\bra{\bm{yyz}}+\ket{\bm{xxz}}\bra{\bm{yzy}}+\ket{\bm{xxz}}\bra{\bm{zyy}}\nonumber\\
			&+\ket{\bm{xzx}}\bra{\bm{yyz}}+\ket{\bm{xzx}}\bra{\bm{yzy}}+\ket{\bm{xzx}}\bra{\bm{zyy}}\nonumber\\
			&+\ket{\bm{zxx}}\bra{\bm{yyz}}+\ket{\bm{zxx}}\bra{\bm{yzy}}+\ket{\bm{zxx}}\bra{\bm{zyy}}\big).
		\end{align}
		Following Lem.~\ref{method:lem3}, the result for ESPOVMs can be rewritten as, 
		\begin{align}\label{method:esc}
			&\,\momentop^{\otimes 3}\nonumber\\
			=&\,\frac{1}{8^n}\sum_{\left(\bm{x},\bm{y},\bm{z},\bm{w},\bm{s},\bm{t}\right)\in \mathcal{K}_2(\bfZ^n_2) }\ket{\bm{x}\bm{y}\bm{z}}\bra{\bm{w}\bm{s}\bm{t}}\nonumber\\
			=&\,\frac{1}{8^n}\Big\{\mathcal{C}^{\rm(i)}+(\mathcal{C}^{\rm(ii)}-9\mathcal{C}^{\rm(i)})\nonumber\\
			&\qquad+[\mathcal{C}^{\rm(iv)}-6\mathcal{C}^{\rm(i)}-2(\mathcal{C}^{\rm(ii)}-9\mathcal{C}^{\rm(i)})]\Big\}\nonumber\\
			=&\,\frac{1}{8^n}\left(4\mathcal{C}^{\rm(i)}-\mathcal{C}^{\rm(ii)}+\mathcal{C}^{\rm(iv)}\right).
		\end{align}
		Also, for RESPOVMs, 
		\begin{align}\label{method:resc}
			&\,\eqmomentop^{\otimes 3}\nonumber\\
			=&\,\frac{1}{8^n}\sum_{\left(\bm{x},\bm{y},\bm{z},\bm{w},\bm{s},\bm{t}\right)\in \mathcal{K}_1(\bfZ^n_2) }\ket{\bm{x}\bm{y}\bm{z}}\bra{\bm{w}\bm{s}\bm{t}}\nonumber\\
			=&\,\frac{1}{8^n}\Big\{\mathcal{C}^{\rm(i)}+(\mathcal{C}^{\rm(ii)}-9\mathcal{C}^{\rm(i)})+(\mathcal{C}^{\rm(iii)}-6\mathcal{C}^{\rm(i)})\nonumber\\
			&\qquad+[\mathcal{C}^{\rm(iv)}-6\mathcal{C}^{\rm(i)}-2(\mathcal{C}^{\rm(ii)}-9\mathcal{C}^{\rm(i)})]\nonumber\\
			&\qquad+[\mathcal{C}^{\rm(v)}-9\mathcal{C}^{\rm(i)}-3(\mathcal{C}^{\rm(iii)}-6\mathcal{C}^{\rm(i)})-(\mathcal{C}^{\rm(ii)}-9\mathcal{C}^{\rm(i)})]\Big\}\nonumber\\
			=&\,\frac{1}{8^n}\left(16\mathcal{C}^{\rm(i)}-2\mathcal{C}^{\rm(ii)}-2\mathcal{C}^{\rm(iii)}+\mathcal{C}^{\rm(iv)}+\mathcal{C}^{\rm(v)}\right).
		\end{align}
		Upon recalling Eq.~\eqref{eq:exp_square}, 
		\begin{align}
			&\,\mathbb{E}((\widehat{O_0}^{{\rm eq}}_\dyadic{A})^2)\nonumber\\
			=&\,2^{3n}{\rm tr}\left(\momentop^{\otimes 3}\,(\rho\otimes O_0 \otimes O_0)\right).
		\end{align}
		For ESPOVMs, by Eq.~\eqref{method:esc}, we can rewrite this as 
		\begin{align}\label{method:esc2}
			\mathbb{E}((\widehat{O_0}^{{\rm eq}}_\dyadic{A})^2)
			&={\rm tr}\left((4\mathcal{C}^{\rm(i)}-\mathcal{C}^{\rm(ii)}+\mathcal{C}^{\rm(iv)})(\rho\otimes O_0 \otimes O_0)\right).
		\end{align}
		Similarly, we can show that from Eqs.~\eqref{eq:exp_square} and~\eqref{method:resc}, 
		\begin{align}\label{method:resc2}
			&\,\mathbb{E}((\widehat{O_0}^{{\rm req}}_\dyadic{A})^2)\nonumber\\
			=&\,\frac{1}{4}{\rm tr}\left((16\mathcal{C}^{\rm(i)}-2\mathcal{C}^{\rm(ii)}-2\mathcal{C}^{\rm(iii)}+\mathcal{C}^{\rm(iv)}+\mathcal{C}^{\rm(v)})(\rho\otimes O_0^{\otimes2})\right).
		\end{align}
        
		We are now ready to provide an upper bound for $\mathbb{E}((\widehat{O_0}^{{\rm req}}_\dyadic{A})^2)$ by utilizing
		\begin{lemma}\label{method:lem4}
			Given the above definitions, for arbitrary state $\rho$ and hermitian operator $O_0$,
			\begin{align}
				\left|{\rm tr}\left(\mathcal{C}^{\rm(i)}(\rho\otimes O_0 \otimes O_0)\right)\right|\le&\, \tr{O_0^2}\,,\label{method:lem4_1}\\
				\left|{\rm tr}\left(\mathcal{C}^{\rm(ii)}(\rho\otimes O_0 \otimes O_0)\right)\right|\le&\, 7\tr{O_0^2}\,,\label{method:lem4_2}\\
				\left|{\rm tr}\left(\mathcal{C}^{\rm(iii)}(\rho\otimes O_0 \otimes O_0)\right)\right|\le&\, 6\tr{O_0^2}\,,\label{method:lem4_3}\\
				\left|{\rm tr}\left(\mathcal{C}^{\rm(iv)}(\rho\otimes O_0 \otimes O_0)\right)\right|\le&\, 3\tr{O_0^2}\,,\label{method:lem4_4}\\
				\left|{\rm tr}\left(\mathcal{C}^{\rm(v)}(\rho\otimes O_0 \otimes O_0)\right)\right|\le&\, 7\tr{O_0^2}\,,\label{method:lem4_5}
			\end{align}
		\end{lemma}
		\noindent
		See Sec.~\ref{proof_lem4} for its proof. This result implies that
		\begin{align}
			\mathbb{E}((\widehat{O_0}^{{\rm eq}}_\dyadic{A})^2)
			\le&\,4\left|{\rm tr}\left(\mathcal{C}^{\rm(i)}(\rho\otimes O_0 \otimes O_0)\right)\right|\nonumber\\
			&+\left|{\rm tr}\left(\mathcal{C}^{\rm(ii)}(\rho\otimes O_0 \otimes O_0)\right)\right|\nonumber\\
			&+\left|{\rm tr}\left(\mathcal{C}^{\rm(iv)}(\rho\otimes O_0 \otimes O_0)\right)\right|\le 14\tr{O_0^2}.
		\end{align}
		Similarly, we can show that $E_{\rho}((\widehat{O_0}^{{\rm req}}_\dyadic{A})^2)\le 13\tr{O_0^2}$.
		
		Furthermore, for $\mathbb{E}((\widehat{O_0}^{{\rm bin}}_\dyadic{A})^2)$,
		\begin{align}
			\mathbb{E}((\widehat{O_0}^{{\rm bin}}_{\bm{p'}})^2)=&\sum_{\bm{x}\in\bfZ^n_2}\bra{\bm{x}}\rho\ket{\bm{x}}\bra{\bm{x}}O_0\ket{\bm{x}}\bra{\bm{x}}O_0\ket{\bm{x}}\nonumber\\
			\le&\sum_{\bm{x},\bm{y}\in\bfZ^n_2}\bra{\bm{x}}\rho\ket{\bm{x}}\bra{\bm{x}}O_0\ket{\bm{y}}\bra{\bm{y}}O_0\ket{\bm{x}}\nonumber\\
			\le&{\rm tr}(O_0^2).
		\end{align}        
		With this, and the arguments leading to  Eq.~\eqref{method:var0} and Eq.~\eqref{main:eq_vareq}, Thm.~\ref{main:thm1}~(i) is proven. The upper bound of the sampling-copy number, that is $\frac{(68\times 2){\rm Var}^{{\rm (r)eq}}_\rho(\widehat{O_0}_{\dyadic{A},j})}{\varepsilon^2}\,\log\!\left(\frac{2M}{\delta}\right)$ for $M$ observables, follows from the statistical benefits of using the median-of-means estimation method~\cite{huang2020,lerasle2019lecture,Blair:1985problem,Jerrum:1986random}.
		
		Suppose that the single observable $O$ of interest is a quantum state~$\sigma$. Then  $\tr{O_0^2}=\tr{\left(\sigma-\frac{I}{2^n}\right)^2}=\tr{\sigma^2}-\frac{\tr{\sigma}}{2^{n-1}}+\frac{1}{2^n}\le \tr{\sigma^2}\le 1$.  Hence, \emph{any quantum-state observable $\sigma$ possesses a constant estimation variance and demonstrates shadow-tomographic advantage with (R)ESPOVM, provided} that the trace inner-product between $\sigma$ and equatorial stabilizer state is computable efficiently. 
		
		The proofs of Thm.~\ref{main:thm1}~(ii) and (iii) only need an average of estimation variance over uniformly chosen input and target state. The result is shown below.
		\begin{lemma}\label{method:lem5}
			\noindent
			(i) If $2^n\gg1$, the averaged value $\overline{{\rm Var}^{{\rm (r)eq}}(\hat{\sigma})}^{\rho,\sigma}$ over uniformly-distributed complex pure states~$\rho$ and real pure states~$\sigma$ approaches $1$ for~${\rm eq}$ and $\frac{1}{2}$ for~${\rm req}$. In randomized-Clifford tomography, this average variance approaches~$1$.\\
			\noindent
			(ii) If $2^n\gg 1$, the averaged value $\overline{{\rm Var}^{{\rm eq}}(\hat{\sigma})}^{\rho,\sigma}$ over uniformly-distributed complex pure states $\rho$ and $\sigma$ approaches~$1$. In randomized-Clifford tomography, this average variance approaches~$1$.
		\end{lemma}
		
		Assisted by previous arguments, the proof of the above lemma naturally leads to Thm.~\ref{main:thm1} (ii) and (iii). However, it contains more technical elements that are deferred to Sec.~\ref{sec:thm_average}.
		
		Furthermore, we can see that the average estimation variance of RESPOVM is exactly half of one of ESPOVM and even one of randomized Clifford tomography. Even though RESPOVM requires two copies of input states for a single sampling trial, the total number of gates used with the same number of sampling copies is, on average, half of what is required in randomized-Clifford shadow tomography.
		
		\begin{figure*}[t]
			\centering
			\includegraphics[width=1.8\columnwidth]{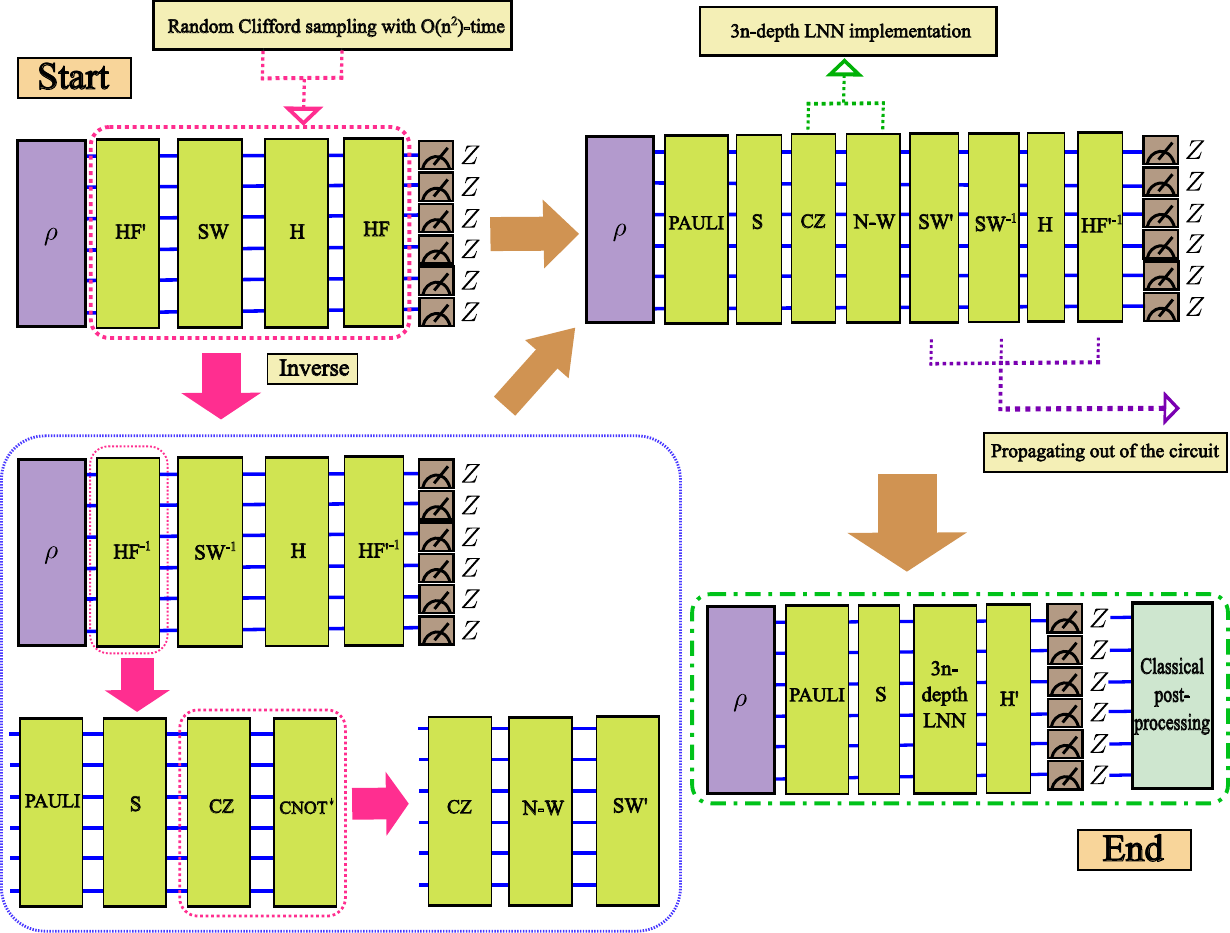}	
			\caption{The schematic diagram of uniform sampling of a $3n$-depth Clifford-circuit measurement in the LNN-architecture. Here, the \textsc{n}-\textsc{w} section is a CNOT circuit expressing the $n\times n$ north-western matrix~\cite{maslov2023,kutin2007} for the binary strings. The final classical post-processing step multiplies an $n\times n$ binary matrix to the measured outcome bit~string, which takes $\mathcal{O}(n^2)$-time.}
			\label{fig:fig6}
		\end{figure*}
		
		\section{Gate, depth, and time complexity for circuit implementation}
		\label{subsec:gatecount}
		
		We now discuss the architectural requirements of (R)ESPOVM shadow tomography, which cover the main ideas of the proof for Thm.~\ref{main:thm2}. In particular, this discussion focuses on the implementation of an~(LNN) architecture for (R)ESPOVM tomography and compares its circuit-depth and time-complexity resources with those of randomized-Clifford shadow tomography.

		We recall~\cite{bravyi2021} that we can efficiently and uniformly sample a Clifford unitary $U$ from $\mathrm{Cl}_n$ in $\mathcal{O}(n^2)$-time. In addition, the sampled form contains the \textsc{hf}$'$\textsc{---sw---h---hf} sections. Each small upper case letters are explained in the Sec.~\ref{sec:random_clifford_tomography}. Here, we can further reduce the depth of circuit lower than previously known results~\cite{maslov2022,maslov2023}. We note that the \textsc{hf}$'$---\textsc{sw} sections are not needed for randomized-Clifford shadow tomography. The reason is that the measurement outcome of a uniform-Clifford measurement can be obtained from the following procedure: One first uniformly samples a Clifford unitary operator in the above form and takes its inverse, which possesses the form \textsc{hf}$^{-1}$---\textsc{h}---\textsc{sw}$^{-1}$---\textsc{hf}$'^{-1}$. Then one just measures in the $Z$~basis after the \textsc{hf}$^{-1}$---\textsc{h} sections and linearly transform this intermediate bit-string outcome with the subsequent CNOT and Pauli transformations in \textsc{sw}$^{-1}$---\textsc{hf}$'$$^{-1}$ to arrive at the final bitstring outcome. In $\mathcal{O}(n^3)$-time, one can optimize the maximal depth of \textsc{hf}$^{-1}$---\textsc{h} sections with long-range gates to $n+\mathcal{O}(\log(n))$ if $n\ge 70$~\cite{maslov2022}. Even though the corresponding asymptotic gate-count upper bound is $\mathcal{O}(\frac{n}{\log(n)})$~\cite{brugiere2021,jiang2020}, it has a large constant factor~\cite{maslov2022} and hence will not improve the former limit for circuits holding hundreds of qubits. If only NN gates are available, we can express the \textsc{cnot} gate layers inside the \textsc{hf}-section in terms of a northwestern matrix~\cite{kutin2007} and SWAP~gates. Here, we do not need to implement the swapping operations, since one can simply propagate all swapping operations out of the circuit. The CNOT circuit of a northwestern-matrix form after the CZ~circuits can be further decomposed into $3n$-depth NN CNOT circuits (sorting network) and phase-gate insertion layers~\cite{maslov2023,kutin2007} in $\mathcal{O}(n^3)$-time. A schematic illustration is presented in Fig.~\ref{fig:fig6}.
		
		As we saw in Thm.~\ref{main:thm0}, there is an alternative depth reduction technique for CZ circuits. (R)ESPOVM shadow tomography can be implemented with at most $n$-depth long-ranged CZ~gates~(or with CNOT and H gates). This maximal depth can be reduced to $\frac{n}{2}+\mathcal{O}(\log(n))$ given that $n\ge 39$ according to Ref.~\cite{maslov2022}, which is half the gate complexity in contrast to randomized-Clifford implementation. Furthermore, with only neighboring CNOT~gates and $S$~gates, we can implement (R)ESPOVM with $2n$-depth LNN architecture in $\mathcal{O}(n^2)$-time. To show this time-complexity bound, we may use the strategy employed in Ref.~\cite{maslov2018}. We first construct a $2n$-depth circuit consisting of NN CNOT~gates in advance following the designated rules dictated by Thm.~6 in~\cite{maslov2018}. With this, the main task for the time-complexity-bound proof is as follows: Given a CZ~circuit that performs the map $\ket{\bm{x}}\mapsto (-1)^{Q(\bm{x})}\ket{\bm{x}}$, where $\bm{x}\in\bfZ^n_2$ and $Q$ is a quadratic binary-valued polynomial, find the values $u_{j},u_{j,k}$ in $\bfZ_4$ $(j\in[n],k\in\left\{j+1,j+2,\ldots,n\right\})$ such that 
		\begin{align}	(-1)^{Q(\bm{x})}=\I^{\sum_{j=1}^{n}u_jy_j+\sum_{j=1}^{n}\sum_{k=j+1}^{n}u_{j,k}(y_j\oplus y_k)}.
		\end{align}
		Here, $\oplus$ and $+$ respectively refer to modulo-$2$ and modulo-$4$ additions. The $y_i$'s~$(i\in[n])$ are defined by,
		\begin{align}
			y_i=x_1\oplus x_2\oplus x_3\oplus \cdots\oplus x_i.
			\label{eq:syseqn}
		\end{align}
		After inserting the phase gates $S^{u_{j,k}}$ and $S^{u_l}$ following the manual in Ref.~\cite{maslov2018}, we finally realize the LNN architecture of CZ~circuit. To analyze the time-complexity, we rewrite Eq.~\eqref{eq:syseqn} as
		\begin{equation}\label{method:bintrans}
			\begin{cases}
				x_1=y_1\\
				x_i=y_i\oplus y_{i-1}\;(i\in\left\{2,3,\ldots,n\right\}).
			\end{cases}
		\end{equation}
		Now, consider one off-diagonal component of $(-1)^{Q(\bm{x})}$. We have simple identity $(-1)^{x_{\mu}x_{\nu}}=\I^{3x_\mu+3x_\nu+(x_\mu\oplus x_\nu)}$. Then by Eq.~\eqref{method:bintrans},
        \begin{align}
        	(-1)^{x_\mu x_\nu}=\I^{3(y_{\mu} \oplus y_{\mu-1})+3(y_{\nu} \oplus y_{\nu-1})+(y_{\mu} \oplus y_{\mu-1} \oplus y_{\nu} \oplus y_{\nu-1})}\,,
        \end{align}
    	with $y_0\equiv0$. Using the equality~\cite{maslov2018} 
        \begin{align}
            \I^{a\oplus b\oplus c}=\I^{3a+3b+3c+3(a\oplus b)+3(a\oplus c)+3(b\oplus c)}\;(a,b,c\in\bfZ_2),
        \end{align}
         and setting $a=y_\mu,b=y_{\mu-1},c=y_{\nu}\oplus y_{\nu-1}$, we can express $\I^{y_{\mu} \oplus y_{\mu-1} \oplus y_{\nu} \oplus y_{\nu-1}}$ as a product of $\I^{y_\mu},\I^{y_{\mu-1}},\I^{y_\nu\oplus y_{\nu-1}}$ and $\I^{y_\mu \oplus y_{\mu-1}},\I^{y_{\mu} \oplus y_{\nu} \oplus y_{\nu-1}},\I^{y_{\mu-1}\oplus y_{\nu} \oplus y_{\nu-1}}$. Again, in a similar manner, we can decompose $\I^{y_{\mu} \oplus y_{\nu} \oplus y_{\nu-1}}$ and $\I^{y_{\mu-1}\oplus y_{\nu} \oplus y_{\nu-1}}$ into a product of powers of $\I$, each with an exponent not exceeding 3 variables. In conclusion, there exist \emph{constant-sized} sets $A,B\subset [n]$ such that
		\begin{align}
			(-1)^{x_\mu x_\nu}=\I^{\sum_{a\in A} u'_{a}y_{a}+\sum_{a,b\in B} u'_{a,b}(y_{a} \oplus y_{b})},
		\end{align}
		where $u'_{\mu},u'_{\mu,\nu}\in\bfZ_4$. Therefore, upon noting that a CZ~circuit has at most $\frac{n^2-n}{2}$ CZ~gates, we have the following scheme for implementing \emph{any} (long-ranged) CZ~circuit in an LNN architecture:
		
		\begin{enumerate}
			\item Prepare a $2n$-depth NN CNOT circuit based on the prescriptions in~Ref.~\cite{maslov2018}.
			
			\item For each unitary $\mathrm{CZ}_{j,k}$ acting on the $j$-th and $k$-th qubits, find the sets $A_{j,k}$ and $B_{j,k}$ such that  $(-1)^{x_a x_b}=\I^{\sum_{a\in A_{j,k}} u'_{a}y_{a}+\sum_{a,b\in B_{j,k}} u'_{a,b}(y_{a} \oplus y_{b})}$. Record all $u'_a$ and $u'_{a,b}$ corresponding to $A_{j,k}$ and $B_{j,k}$.
			
			\item Insert the phase gates~$S^{u'_a} (a\in A_{j,k})$ and $S^{u'_{a,b}} (a,b\in B_{j,k})$ according to the insertion rules in Ref.~\cite{maslov2018}.
			
			\item Repeat steps~2 and~3 for all $j,k\in[n]$. If there is already some phase gate acting on the location, we multiply the existing one and the gate we put this time. 
			
		\end{enumerate}
		The Reader may consult Ref.~\cite{supple} for the precise NN CNOT-circuit prescription and phase-gate insertion rule. For each $j,k\in[n]$, we note that steps~2 to~4 each takes constant time to execute. Since there is a total of $\mathcal{O}(n^2)$ steps to be repeated, the total time-complexity is $\mathcal{O}(n^2)$.
		
		\section{Survey of recent CZ-gate implementation schemes}
		\label{subsec:survey}
		
		We shall list all recent fidelity values sorted according to some of these platforms. In photonic quantum computing, a fidelity of over 81\% has been reported on silicon-nitride photonic chips~\cite{Lee:2022controlled-not}, and over 99\% when Zeeman~coupling with quantum~dots and nitrogen-vacancy centers is exploited in simulations~\cite{Russo:2018photonic}. Spin-based quantum computation allows for more than 99.5\% fidelity~\cite{Xue:2022quantum,Rimbach-Russ:2023simple}. On Rydberg-interacting platforms, Ref.~\cite{Graham:219rydberg}, \cite{Levine:2019parallel} and later \cite{Evered:2023high-fidelity} respectively reported gate~fidelities of 89\%, 97\% and 99.5\% using atomic-qubit-array entanglement. Using adiabatic pulses, simulations revealed fidelity values greater than 99.9\%~\cite{Saffman:2020symmetric}. On superconducting platforms, adiabaticity-shortcut techniques allows for over 96\% fidelity~\cite{Wang:2018ExperimentalRO}, 97.6\% with fluxonium qubits~\cite{Simakov:2023coupler} and 99.1\% through interfering superconducting qubits with weak anharmonicity~\cite{Rol:2019fast}. Nonadiabatic CZ gates were also realized with 99.54\% fidelity~\cite{Li:2019realisation}. Furthermore, gmon qubits employed by Google AI Quantum demonstrated gate fidelities as high as 99.81\%~\cite{Foxen:2020demonstrating}. The more recent fidelity of 99.93\% is possible with intermediate leakage control~\cite{Negirneac:2021high-fidelity}.

        \section{Proofs of other Lemmas}\label{sec:proof_other_lemmas}


		\subsection{Proof of Lemma~\ref{method:lem3}}\label{proof_lem1}
		
		We start with $\mathcal{S}^{{\rm req}}_n$. Similar to Eq.~\eqref{eq:reqsecond}, 
		\begin{align}
			&\eqmomentop^{\otimes 3}\nonumber\\
			=&\,\eqmoment\frac{1}{8^n}\sum_{\bm{x},\bm{y},\bm{z},\bm{w},\bm{s},\bm{t}\in\bfZ^n_2}\ket{\bm{x}\bm{y}\bm{z}}\bra{\bm{w}\bm{s}\bm{t}}\nonumber\\
			&\times(-1)^{\bm{x}^\top\dyadic{A}\bm{x}+\bm{y}^\top\dyadic{A}\bm{y}+\bm{z}^\top\dyadic{A}\bm{z}+\bm{w}^\top\dyadic{A}\bm{w}+\bm{s}^\top\dyadic{A}\bm{s}+\bm{t}^\top\dyadic{A}\bm{t}}\nonumber\\
			=&\,\frac{1}{8^n}\sum_{\bm{x},\bm{y},\bm{z},\bm{w},\bm{s},\bm{t}\in\bfZ^n_2}\ket{\bm{x}\bm{y}\bm{z}}\bra{\bm{w}\bm{s}\bm{t}}\eqmoment\bigg\{\nonumber\\
			&\times(-1)^{\bm{x}^\top\dyadic{A}\bm{x}+\bm{y}^\top\dyadic{A}\bm{y}+\bm{z}^\top\dyadic{A}\bm{z}+\bm{w}^\top\dyadic{A}\bm{w}+\bm{s}^\top\dyadic{A}\bm{s}+\bm{t}^\top\dyadic{A}\bm{t}}\bigg\}.
		\end{align}
		By the same token, for nonzero coefficients in the summand, the following equations must hold for all $p,q\in[n]$:
		\begin{align}\label{eq:thirdsys}
			\begin{cases}
				x_p+y_p+z_p+w_p+s_p=t_p\;({\rm mod\;2}),\\		x_px_q+y_py_q+z_pz_q+w_pw_q+s_ps_q+t_pt_q=\,0\;\,\,({\rm mod\;2}).
			\end{cases}
		\end{align}
		Substituting the upper equation into the lower one, we obtain,
		\begin{align}\label{eq:thirdreq}
			0=&\,x_p(y_q+z_q+w_q+s_q)+y_p(z_q+x_q+w_q+s_q)\nonumber\\
			&+z_p(x_q+y_q+w_q+s_q)+w_p(x_q+y_q+z_q+s_q)\nonumber\\
			&+s_p(x_q+y_q+z_q+w_q).
		\end{align}
		If $\bm{x}=\bm{y}=\bm{z}=\bm{w}=\bm{s}$, the above equation clearly holds. Otherwise, we have two cases:\\
		(i) There exists $p'\in[N]$ such that two of $x_{p'},y_{p'},z_{p'},w_{p'},s_{p'}$ are $0~(1 \;{\rm resp.})$ and the rest are $1~(0)$, and $x_i=y_i=z_i=w_i=s_i$ for $i<p'$.\\
		(ii) There exists $p'\in[N]$ such that four of $x_{p'},y_{p'},z_{p'},w_{p'},s_{p'}$ are $0~(1\;{\rm resp.})$ and the other is $1~(0)$, and $x_i=y_i=z_i=w_i=s_i$ for $i<p'$.\\
		First, suppose case~(i). WLOG, when $x_{p'}=y_{p'}\;(*)$, then $z_{p'}=w_{p'}=s_{p'}=x_{p'}+1$, which together with Eq.~\eqref{eq:thirdreq}, implies that
		\begin{align}
			x_{p'}y_q+x_qy_{p'}+(x_{p'}+1)(x_q+y_q)=x_q+y_q=0\;({\rm mod\;2}) 
		\end{align}
		as~$x_{p'}=y_{p'}$. So, we obtain $\bm{x}=\bm{y}$. Substituting this back to Eq.~\eqref{eq:thirdsys},
		\begin{equation}
			\begin{cases}\label{eq:thirdsys2}
				z_p+w_p+s_p=t_p\;({\rm mod\;2}),\\
				z_pz_q+w_pw_q+s_ps_q+t_pt_q=0\;({\rm mod\;2}).
			\end{cases}
		\end{equation}
		This is exactly what we have obtained during the calculation of the second moment. Hence we conclude the same statement, two of $\bm{z},\bm{w},\bm{s}$ are equal. By the first equation of Eq.~\eqref{eq:thirdsys2}, the remaining one equals $\bm{t}$. Therefore, we conclude that each two of $\bm{x},\bm{y},\bm{z},\bm{w},\bm{s},\bm{t}$ must be equal. Note that even if we chose two other variables, not $\bm{x}$ and $\bm{y}$ in assumption$(*)$, we arrive at the same conclusion. 
		
		Next, let us suppose case~(ii). By the first equation of Eq.~\eqref{eq:thirdsys}, $x_i=y_i=z_i=w_i=s_i=t_i$ for $i<p'$. WLOG, assume that $x_{p'}=y_{p'}=z_{p'}=w_{p'}\ne s_{p'}$. Then, we obtain $s_{p'}=t_{p'}$ from Eq.~\eqref{eq:thirdsys}. Next, we rewrite the problem as 
		\begin{align}
			\begin{cases}
				y_p+z_p+w_p+s_p+t_p=x_p\;({\rm mod\;2}),\\
				x_px_q+y_py_q+z_pz_q+w_pw_q+s_ps_q+t_pt_q=0\;({\rm mod\;2}).
			\end{cases}
		\end{align}
		Similarly, substituting the first equation into the second,
		
		\begin{align}
			0=&\,t_p(y_q+z_q+w_q+s_q)+y_p(z_q+x_q+w_q+s_q)\nonumber\\
			&+z_p(t_q+y_q+w_q+s_q)+w_p(t_q+y_q+z_q+s_q)\nonumber\\
			&+s_p(t_q+y_q+z_q+w_q).
		\end{align}
		Remember we assumed that $t_{p'}=s_{p'}$, $t_{p'}\ne w_{p'}$, $y_{p'}=z_{p'}=w_{p'}$, and $y_i=z_i=w_i=s_i=t_i$ for $i<p'$. This is exactly the case of solving under the assumption (i). So we conclude the same statement that $\bm{t}=\bm{s}$, two of $\bm{y},\bm{z},\bm{w}$ are equal and hence remaining one equals $\bm{x}$.
		
		Throughout cases (i) and (ii), we finally conclude that the coefficient of $\ket{\bm{x}\bm{y}\bm{z}}\bra{\bm{w}\bm{s}\bm{t}}$ is nonzero if and only if each triplet in the variable set $\left\{\bm{x},\bm{y},\bm{z},\bm{w},\bm{s},\bm{t}\right\}$ must be equal and in all such cases, the coefficients $\eqmoment\big\{(-1)^{\bm{x}^\top\dyadic{A}\bm{x}+\bm{y}^\top\dyadic{A}\bm{y}+\bm{z}^\top\dyadic{A}\bm{z}+\bm{w}^\top\dyadic{A}\bm{w}+\bm{s}^\top\dyadic{A}\bm{s}+\bm{t}^\top\dyadic{A}\bm{t}}\big\}$ all equal~1. The proof for RESPOVM is now completed. 
		
		Now, we move on to $\mathcal{S}^{{\rm eq}}_n$, whence
		\begin{align}
			&\,\momentop^{\otimes 3}\nonumber\\
			=&\,\moment\frac{1}{8^n}\sum_{\bm{x},\bm{y},\bm{z},\bm{w},\bm{s},\bm{t}\in\bfZ^n_2}\ket{\bm{x}\bm{y}\bm{z}}\bra{\bm{w}\bm{s}\bm{t}}\nonumber\\
			&\times\I^{\bm{x}^\top\dyadic{A}\bm{x}+\bm{y}^\top\dyadic{A}\bm{y}+\bm{z}^\top\dyadic{A}\bm{z}-\bm{w}^\top\dyadic{A}\bm{w}-\bm{s}^\top\dyadic{A}\bm{s}-\bm{t}^\top\dyadic{A}\bm{t}}\nonumber\\
			=&\,\frac{1}{8^n}\sum_{\bm{x},\bm{y},\bm{z},\bm{w},\bm{s},\bm{t}\in\bfZ^n_2}\ket{\bm{x}\bm{y}\bm{z}}\bra{\bm{w}\bm{s}\bm{t}}\moment\bigg\{\nonumber\\
			&\times\I^{\bm{x}^\top\dyadic{A}\bm{x}+\bm{y}^\top\dyadic{A}\bm{y}+\bm{z}^\top\dyadic{A}\bm{z}-\bm{w}^\top\dyadic{A}\bm{w}-\bm{s}^\top\dyadic{A}\bm{s}-\bm{t}^\top\dyadic{A}\bm{t}}\bigg\}.
		\end{align}
		By the same argument, for the coefficient not to be 0, the following equation must hold. For all $p,q\in\bfZ^n_2\;(p<q)$, 
		\begin{align}\label{eq:thirdsys3}
			\begin{cases}
				x_p+y_p+z_p=w_p+s_p+t_p\;({\rm mod\;4}),\\
				x_px_q+y_py_q+z_pz_q+w_pw_q+s_ps_q+t_pt_q=0\;({\rm mod\;2}).
			\end{cases}
		\end{align}
		We add $w_p+s_p$ for both sides of the first equation of~\eqref{eq:thirdsys3}, we obtain that 
		\begin{align}
			\begin{cases}
				x_p+y_p+z_p+w_p+s_p=t_p\;({\rm mod\;2}),\\
				x_px_q+y_py_q+z_pz_q+w_pw_q+s_ps_q+t_pt_q=0\;({\rm mod\;2}).
			\end{cases}
		\end{align}
		This is equation for $\mathcal{S}^{{\rm req}}_n$-case. Hence we obtain the same statement, each $3$-couples in variable set $\left\{\bm{x},\bm{y},\bm{z},\bm{w},\bm{s},\bm{t}\right\}$ must be equal. However, in this case, we have a stricter condition, the first equation of Eq.~\eqref{eq:thirdsys3}. WLOG, we only consider two cases.\\
		(i) $\bm{x}=\bm{y}:2x_p+z_p=w_p+s_p+t_p\;({\rm mod\;4})\rightarrow z_p=w_p+s_p+t_p\;({\rm mod\;2})$. Hence two variables in $\left\{\bm{z},\bm{w},\bm{s},\bm{t}\right\}$ are equal each other.\\
		(ii) $\bm{x}=\bm{w}:y_p+z_p=s_p+t_p\;({\rm mod\;4})$. If $\bm{y}=\bm{z},2y_p=s_p+t_p\rightarrow \bm{y}=\bm{z}=\bm{s}=\bm{t}$. If $\bm{y}=\bm{s},\bm{z}=\bm{t}$. If $\bm{y}=\bm{w}$, all cases are equivalent to previous cases. Therefore the statement that coefficients are non-zero is equivalent to that if $\ket{\bm{x}\bm{y}\bm{z}}\bra{\bm{w}\bm{s}\bm{t}}$ has 2 same variables where each do not locate in the same position(bra and ket) and the other 2 variables in bra (or ket) are equal, remaining variables in ket (or bra) must be equal to previous 2 variables. Finally, in those cases, the coefficient cannot be other values but $1$.  
		
		Other cases can be managed similarly and are known to have similar effects of (i) or (ii). We thus proved Lem.~\ref{method:lem3}.
		
		\subsection{Proof of Lemma~\ref{method:lem4}}\label{proof_lem4}
		First, we note that 
		\begin{align}
			&\sum_{\bm{x}\in \bfZ^N_2}\tr{\ket{\bm{x}\bm{x}\bm{x}}\bra{\bm{x}\bm{x}\bm{x}}(\sigma\otimes O_0 \otimes O_0)}\nonumber\\
			=&\sum_{\bm{x}\in \bfZ^N_2}\bra{\bm{x}}\sigma\ket{\bm{x}}\bra{\bm{x}}O_0\ket{\bm{x}}\bra{\bm{x}}O_0\ket{\bm{x}}\nonumber\\
			\le&\sum_{\bm{x},\bm{y}\in\bfZ^N_2}\bra{\bm{x}}\sigma\ket{\bm{x}}\bra{\bm{x}}O_0\ket{\bm{y}}\bra{\bm{y}}O_0\ket{\bm{x}}\le {\rm tr}(O_0^2).
		\end{align}
		Here, we used the fact that $\bra{\bm{x}}O_0\ket{\bm{y}}\bra{\bm{y}}O_0\ket{\bm{x}}\ge 0$ for $\forall \bm{x},\bm{y}\in\bfZ^n_2$. This proves the Eq.~\eqref{method:lem4_1}.
		
		Also, we easily note that 
		\begin{align}
			&\sum_{\bm{x},\bm{y}\bm{z}\in\bfZ^{N}_2}{\rm tr}\Big((\ket{\bm{x}\bm{y}\bm{z}}\bra{\bm{x}\bm{y}\bm{z}}+\ket{\bm{x}\bm{y}\bm{z}}\bra{\bm{x}\bm{z}\bm{y}}\nonumber\\
			&\qquad\qquad+\ket{\bm{x}\bm{y}\bm{z}}\bra{\bm{y}\bm{x}\bm{z}}+\ket{\bm{x}\bm{y}\bm{z}}\bra{\bm{z}\bm{x}\bm{y}}\nonumber\\
			&\qquad\qquad+\ket{\bm{x}\bm{y}\bm{z}}\bra{\bm{z}\bm{y}\bm{x}}+\ket{\bm{x}\bm{y}\bm{z}}\bra{\bm{y}\bm{z}\bm{x}})(\sigma\otimes O_0 \otimes O_0)\Big)\nonumber\\
			=&\left({\rm tr}(O_0)\right)^2+{\rm tr}(O_0^2)+2{\rm tr}(O_0){\rm tr}(\sigma O_0)+2{\rm tr}(\sigma 
			O_0^2)\nonumber\\
			\le&\,({\rm tr}(O_0))^2+{\rm tr}(O_0^2)+2{\rm tr}(O_0)\|O_0\|_{\infty}+2\|O_0^2\|_{\infty}\nonumber\\
			\le&\,({\rm tr}(O_0))^2+3{\rm tr}(O_0^2)+2{\rm tr}(O_0)\|O_0\|_{\infty}\nonumber\\
			\le&\,3\tr{O_0^2},
		\end{align}
		where $\|O\|_{\infty}\equiv \max_{\ket{\psi}}\left\{\left|\bra{\psi}O\ket{\psi}\right|\right\}$. 
		Furthermore, we can consider another case when two variables are the same in ket. Let $\sigma=\sum_{\lambda}\lambda\ket{\lambda}\bra{\lambda}$ be spectral decomposition of $\sigma$, hence $\forall \lambda\ge0$ and $\sum_{\lambda}\lambda=1$. Then we obtain the following list of results:
		\begin{widetext}
        
		\begin{align}
			\left|\sum_{\bm{x},\bm{y},\bm{z}\in \bfZ^N_2}\tr{\ket{\bm{x}\bm{x}\bm{y}}\bra{\bm{z}\bm{z}\bm{y}}(\sigma\otimes O_0 \otimes O_0)}\right|=\left|\sum_{\bm{x},\bm{y},\bm{z}\in \bfZ^N_2}\bra{\bm{z}}\sigma\ket{\bm{x}}\bra{\bm{z}}O_0\ket{\bm{x}}\bra{\bm{y}}O_0\ket{\bm{y}}\right|
			&=\left|{\rm tr}(O_0){\rm tr}(\sigma O_0^\top)\right|=0;\\
	       \left|\sum_{\bm{x},\bm{y},\bm{z}\in \bfZ^N_2}\tr{\ket{\bm{x}\bm{x}\bm{y}}\bra{\bm{z}\bm{y}\bm{z}}(\sigma\otimes O_0 \otimes O_0)}\right|=\left|\sum_{\bm{x},\bm{y},\bm{z}\in \bfZ^N_2}\bra{\bm{z}}\sigma\ket{\bm{x}}\bra{\bm{y}}O_0\ket{\bm{x}}\bra{\bm{z}}O_0\ket{\bm{y}}\right|&=\left|{\rm tr}(\sigma (O_0^\top)^2)\right|\le\|(O_0^\top)^2\|_{\infty}\nonumber\\&\le {\rm tr}((O_0^\top)^2)={\rm tr}(O_0^2);\\					
			\left|\sum_{\bm{x},\bm{y},\bm{z}\in \bfZ^N_2}\tr{\ket{\bm{x}\bm{x}\bm{y}}\bra{\bm{y}\bm{z}\bm{z}}(\sigma\otimes O_0 \otimes O_0)}\right|=\left|\sum_{\bm{x},\bm{y},\bm{z}\in \bfZ^N_2}\bra{\bm{y}}\sigma\ket{\bm{x}}\bra{\bm{z}}O_0\ket{\bm{x}}\bra{\bm{z}}O_0\ket{\bm{y}}\right|
			&=\left|{\rm tr}(\sigma O_0^\top O_0)\right|\le \|O_0^\top O_0\|_{\infty};\label{proof_lem4_res1}\\
			\left|\sum_{\bm{x},\bm{y},\bm{z}\in \bfZ^N_2}\tr{\ket{\bm{y}\bm{x}\bm{x}}\bra{\bm{z}\bm{z}\bm{y}}(\sigma\otimes O_0 \otimes O_0)}\right|=\left|\sum_{\bm{x},\bm{y},\bm{z}\in \bfZ^N_2}\bra{\bm{z}}\sigma\ket{\bm{y}}\bra{\bm{z}}O_0\ket{\bm{x}}\bra{\bm{y}}O_0\ket{\bm{x}}\right|&=\left|{\rm tr}(\sigma O_0O_0^\top)\right|\le \|O_0O_0^\top\|_{\infty}; \label{proof_lem4_res2}\\
            \left|\sum_{\bm{x},\bm{y},\bm{z}\in \bfZ^N_2}\tr{\ket{\bm{y}\bm{x}\bm{x}}\bra{\bm{z}\bm{y}\bm{z}}(\sigma\otimes O_0 \otimes O_0)}\right|=\left|\sum_{\bm{x},\bm{y},\bm{z}\in \bfZ^N_2}\bra{\bm{z}}\sigma\ket{\bm{x}}\bra{\bm{y}}O_0\ket{\bm{x}}\bra{\bm{z}}O_0\ket{\bm{y}}\right|&={\rm tr}(\sigma (O_0^\top)^2)\le\|(O_0^\top)^2\|_{\infty}\nonumber\\&\le {\rm tr}((O_0^\top)^2)={\rm tr}(O_0^2);
        \end{align}
        \begin{align}
			\left|\sum_{\bm{x},\bm{y},\bm{z}\in \bfZ^N_2}\tr{\ket{\bm{y}\bm{x}\bm{x}}\bra{\bm{y}\bm{z}\bm{z}}(\sigma\otimes O_0 \otimes O_0)}\right|=\left|\sum_{\bm{x},\bm{y},\bm{z}\in \bfZ^N_2}\bra{\bm{y}}\sigma\ket{\bm{y}}\bra{\bm{z}}O_0\ket{\bm{x}}\bra{\bm{z}}O_0\ket{\bm{x}}\right|&=\left|{\rm tr}(O_0O_0^\top)\right|;\label{proof_lem4_res3}\\
			\left|\sum_{\bm{x},\bm{y},\bm{z}\in \bfZ^N_2}\tr{\ket{\bm{x}\bm{y}\bm{x}}\bra{\bm{z}\bm{z}\bm{y}}(\sigma\otimes O_0 \otimes O_0)}\right|=\left|\sum_{\bm{x},\bm{y},\bm{z}\in \bfZ^N_2}\bra{\bm{z}}\sigma\ket{\bm{x}}\bra{\bm{z}}O_0\ket{\bm{y}}\bra{\bm{y}}O_0\ket{\bm{x}}\right|&\le{\rm tr}(O_0^2);\\
			\left|\sum_{\bm{x},\bm{y},\bm{z}\in \bfZ^N_2}\tr{\ket{\bm{x}\bm{y}\bm{x}}\bra{\bm{z}\bm{y}\bm{z}}(\sigma\otimes O_0 \otimes O_0)}\right|=\left|\sum_{\bm{x},\bm{y},\bm{z}\in \bfZ^N_2}\bra{\bm{z}}\sigma\ket{\bm{x}}\bra{\bm{y}}O_0\ket{\bm{y}}\bra{\bm{z}}O_0\ket{\bm{x}}\right|&\le({\rm tr}(O_0))^2=0;\\
			\left|\sum_{\bm{x},\bm{y},\bm{z}\in \bfZ^N_2}\tr{\ket{\bm{x}\bm{y}\bm{x}}\bra{\bm{y}\bm{z}\bm{z}}(\sigma\otimes O_0 \otimes O_0)}\right|=\left|\sum_{\bm{x},\bm{y},\bm{z}\in \bfZ^N_2}\bra{\bm{y}}\sigma\ket{\bm{x}}\bra{\bm{z}}O_0\ket{\bm{y}}\bra{\bm{z}}O_0\ket{\bm{x}}\right|&\le\|O_0O_0^\top\|_{\infty};\label{proof_lem4_res4}\\
            \left|\sum_{\bm{x},\bm{y}\in \bfZ^N_2}\tr{\ket{\bm{x}\bm{x}\bm{y}}\bra{\bm{x}\bm{y}\bm{x}}(\sigma\otimes O_0 \otimes O_0)}\right|=\left|\sum_{\bm{x},\bm{y}\in \bfZ^N_2}\bra{\bm{x}}\sigma\ket{\bm{x}}\bra{\bm{y}}O_0\ket{\bm{x}}\bra{\bm{x}}O_0\ket{\bm{y}}\right|&\le\sum_{\bm{x},\bm{y}\in \bfZ^N_2}\bra{\bm{y}}O_0\ket{\bm{x}}\bra{\bm{x}}O_0\ket{\bm{y}}\nonumber\\&={\rm tr}(O_0^2)\\
              \left|\sum_{\bm{x},\bm{y}\in \bfZ^N_2}\tr{\ket{\bm{x}\bm{x}\bm{y}}\bra{\bm{x}\bm{y}\bm{x}}(\sigma\otimes O_0 \otimes O_0)}\right|=\left|\sum_{\bm{x},\bm{y}\in \bfZ^N_2}\bra{\bm{x}}\sigma\ket{\bm{x}}\bra{\bm{y}}O_0\ket{\bm{x}}\bra{\bm{x}}O_0\ket{\bm{y}}\right|&\le\sum_{\bm{x},\bm{y}\in \bfZ^N_2}\bra{\bm{y}}O_0\ket{\bm{x}}\bra{\bm{x}}O_0\ket{\bm{y}}\nonumber\\&={\rm tr}(O_0^2);\\
			\left|\sum_{\bm{x},\bm{y}\in \bfZ^N_2}\tr{\ket{\bm{x}\bm{x}\bm{y}}\bra{\bm{x}\bm{x}\bm{y}}(\sigma\otimes O_0 \otimes O_0)}\right|=\left|\sum_{\bm{x},\bm{y}\in \bfZ^N_2}\bra{\bm{x}}\sigma\ket{\bm{x}}\bra{\bm{x}}O_0\ket{\bm{x}}\bra{\bm{y}}O_0\ket{\bm{y}}\right|&\le\sum_{\bm{x}\in \bfZ^N_2}\left|\bra{\bm{x}}\sigma\ket{\bm{x}}\bra{\bm{x}}O_0\ket{\bm{x}}{\rm tr}(O_0)\right|\nonumber\\&=\left|{\rm tr}(O_0)\right|\|O_0\|_{\infty}=0;\\
            \left|\sum_{\bm{x},\bm{y}\in \bfZ^N_2}\tr{\ket{\bm{x}\bm{x}\bm{y}}\bra{\bm{y}\bm{x}\bm{x}}(\sigma\otimes O_0 \otimes O_0)}\right|=\left|\sum_{\bm{x},\bm{y}\in \bfZ^N_2}\bra{\bm{y}}\sigma^\top\ket{\bm{x}}\bra{\bm{y}}O_0\ket{\bm{y}}\bra{\bm{y}}O_0\ket{\bm{x}}\right|&\le {\rm tr}(O_0^2)\nonumber\\&(\because\;{\rm see}\;{\rm Eq.}~\eqref{proof_lem4_long} \;{\rm below});\\
            \left|\sum_{\bm{x},\bm{y}\in \bfZ^N_2}\tr{\ket{\bm{y}\bm{x}\bm{x}}\bra{\bm{y}\bm{y}\bm{y}}(\sigma\otimes O_0 \otimes O_0)}\right|=\left|\sum_{\bm{x},\bm{y}\in \bfZ^N_2}\bra{\bm{y}}\sigma\ket{\bm{y}}\bra{\bm{y}}O_0\ket{\bm{x}}\bra{\bm{x}}O_0^\top\ket{\bm{y}}\right|&\le \sum_{\bm{y}\in \bfZ^N_2}\bra{\bm{y}}\sigma\ket{\bm{y}}\left|\bra{\bm{y}}O_0O_0^\top\ket{\bm{y}}\right|\nonumber\\&\le\|O_0O_0^\top\|_{\infty}\label{proof_lem4_res5};\\
			\left|\sum_{\bm{x},\bm{y}\in \bfZ^N_2}\tr{\ket{\bm{x}\bm{y}\bm{x}}\bra{\bm{y}\bm{y}\bm{y}}(\sigma\otimes O_0 \otimes O_0)}\right|=\left|\sum_{\bm{x},\bm{y}\in \bfZ^N_2}\bra{\bm{y}}\sigma\ket{\bm{x}}\bra{\bm{y}}O_0\ket{\bm{y}}\bra{\bm{y}}O_0\ket{\bm{x}}\right|&\le {\rm tr}(O_0^2) \;(\because {\rm Eq.}~\eqref{proof_lem4_long});\\
			\left|\sum_{\bm{x},\bm{y}\in \bfZ^N_2}\tr{\ket{\bm{y}\bm{x}\bm{x}}\bra{\bm{x}\bm{x}\bm{y}}(\sigma\otimes O_0 \otimes O_0)}\right|=\left|\sum_{\bm{x},\bm{y}\in \bfZ^N_2}\bra{\bm{x}}\sigma\ket{\bm{y}}\bra{\bm{x}}O_0\ket{\bm{x}}\bra{\bm{y}}O_0\ket{\bm{x}}\right|&\le {\rm tr}(O_0^2)\;(\because {\rm Eq.}~\eqref{proof_lem4_long});\\
			\left|\sum_{\bm{x},\bm{y}\in \bfZ^N_2}\tr{\ket{\bm{y}\bm{x}\bm{x}}\bra{\bm{y}\bm{x}\bm{x}}(\sigma\otimes O_0 \otimes O_0)}\right|=\left|\sum_{\bm{x},\bm{y}\in \bfZ^N_2}\bra{\bm{y}}\sigma\ket{\bm{y}}\bra{\bm{x}}O_0\ket{\bm{x}}\bra{\bm{x}}O_0\ket{\bm{x}}\right|&\le\sum_{\bm{x},\bm{z}\in \bfZ^N_2}\bra{\bm{x}}O_0\ket{\bm{z}}\bra{\bm{z}}O_0\ket{\bm{x}}\nonumber\\&\le {\rm tr}(O_0^2);\\
			\left|\sum_{\bm{x},\bm{y}\in \bfZ^N_2}\tr{\ket{\bm{y}\bm{x}\bm{x}}\bra{\bm{x}\bm{y}\bm{x}}(\sigma\otimes O_0 \otimes O_0)}\right|=\left|\sum_{\bm{x},\bm{y}\in \bfZ^N_2}\bra{\bm{x}}\sigma\ket{\bm{y}}\bra{\bm{y}}O_0\ket{\bm{x}}\bra{\bm{x}}O_0\ket{\bm{x}}\right|&\le {\rm tr}(O_0^2)\;(\because {\rm Eq.}~\eqref{proof_lem4_long});\\
			\left|\sum_{\bm{x},\bm{y}\in \bfZ^N_2}\tr{\ket{\bm{x}\bm{y}\bm{x}}\bra{\bm{x}\bm{x}\bm{y}}(\sigma\otimes O_0 \otimes O_0)}\right|=\left|\sum_{\bm{x},\bm{y}\in \bfZ^N_2}\bra{\bm{x}}\sigma\ket{\bm{x}}\bra{\bm{x}}O_0\ket{\bm{y}}\bra{\bm{y}}O_0\ket{\bm{x}}\right|&\le {\rm tr}(O_0^2);\\ 
			\left|\sum_{\bm{x},\bm{y}\in \bfZ^N_2}\tr{\ket{\bm{x}\bm{y}\bm{x}}\bra{\bm{x}\bm{y}\bm{x}}(\sigma\otimes O_0 \otimes O_0)}\right|=\left|\sum_{\bm{x},\bm{y}\in \bfZ^N_2}\bra{\bm{x}}\sigma\ket{\bm{x}}\bra{\bm{y}}O_0\ket{\bm{y}}\bra{\bm{x}}O_0\ket{\bm{x}}\right|&\le ({\rm tr}(O_0))^2=0;\\ 
			\left|\sum_{\bm{x},\bm{y}\in \bfZ^N_2}\tr{\ket{\bm{x}\bm{y}\bm{x}}\bra{\bm{y}\bm{x}\bm{x}}(\sigma\otimes O_0 \otimes O_0)}\right|=\left|\sum_{\bm{x},\bm{y}\in \bfZ^N_2}\bra{\bm{y}}\sigma\ket{\bm{x}}\bra{\bm{x}}O_0\ket{\bm{y}}\bra{\bm{x}}O_0\ket{\bm{x}}\right|&\le {\rm tr}(O_0^2)\;(\because {\rm Eq.}~\eqref{proof_lem4_long}).
		\end{align}
        \begin{align}
			\left|\sum_{\bm{x},\bm{y}\in \bfZ^N_2}\tr{\ket{\bm{x}\bm{x}\bm{y}}\bra{\bm{y}\bm{y}\bm{y}}(\sigma\otimes O_0 \otimes O_0)}\right|&=\left|\sum_{\bm{x},\bm{y}\in \bfZ^N_2}\bra{\bm{y}}\sigma\ket{\bm{x}}\bra{\bm{y}}O_0\ket{\bm{x}}\bra{\bm{y}}O_0\ket{\bm{y}}\right|\le\sum_{\bm{y}\in\bfZ^N_2}\left|\bra{\bm{y}}\sigma O_0^\top\ket{\bm{y}}\bra{\bm{y}}O_0\ket{\bm{y}}\right|\nonumber\\&\le\sqrt{\sum_{\bm{y}\in\bfZ^N_2}\left|\bra{\bm{y}}\sigma O_0^\top\ket{\bm{y}}\right|^2\sum_{\bm{z}\in\bfZ^N_2}\bra{\bm{z}}O_0\ket{\bm{z}}\bra{\bm{z}}O_0\ket{\bm{z}}}\nonumber\\&\le\sqrt{\sum_{\bm{y}\in\bfZ^N_2}\left|\sum_{\lambda}\lambda\braket{\bm{y}|\lambda}\bra{\lambda} O_0^\top\ket{\bm{y}}\right|^2 \sum_{\bm{z},k\in\bfZ^N_2}\bra{\bm{z}}O_0\ket{\bm{k}}\bra{\bm{k}}O_0\ket{\bm{z}}}\nonumber\\
			&\le \sqrt{\sum_{\bm{y}\in\bfZ^N_2}\left|\bra{\lambda_{\bm{y}}} O_0^\top\ket{\bm{y}}\right|^2\left|\braket{\bm{y}|\lambda_{\bm{y}}}\right|^2{\rm tr}(O_0^2)}\nonumber\\&\left(\lambda_{\bm{y}}\equiv{\rm argmax}_{\lambda}\left\{\left|\bra{\lambda} O_0^\top\ket{\bm{y}}\braket{\bm{y}|\lambda}\right|\right\}\right)\nonumber
			\\&\le\sqrt{\sum_{\bm{y}\in\bfZ^N_2}\sum_{\lambda}\left|\braket{\bm{y}|\lambda}\right|^2\bra{\lambda} O_0^\top\ket{\bm{y}}\bra{\bm{y}} O_0^\top\ket{\lambda}{\rm tr}(O_0^2)}\nonumber\\
			&\le\sqrt{\sum_{\lambda}\bra{\lambda} O_0^\top\ket{\bm{y}_{\lambda}}\bra{\bm{y}_{\lambda}} O_0^\top\ket{\lambda}{\rm tr}(O_0^2)}\left(\bm{y}_{\lambda}\equiv {\rm argmax}_{\bm{y}\in\bfZ^n_2}\left\{\left|\bra{\lambda}O_0^\top\ket{\bm{y}}\right|^2\right\}\right)\nonumber\\
			&\le\sqrt{\sum_{\lambda}\sum_{\bm{z}\in\bfZ^N_2}\bra{\lambda} O_0^\top\ket{\bm{z}}\bra{\bm{z}} O_0^\top\ket{\lambda}{\rm tr}(O_0^2)}=\sqrt{\sum_{\lambda}\bra{\lambda} (O_0^\top)^2\ket{\lambda}{\rm tr}(O_0^2)}\nonumber\\
			&=\sqrt{{\rm tr}((O_0^\top)^2)}\sqrt{{\rm tr}(O_0^2)}={\rm tr}(O_0^2);\label{proof_lem4_long}
		\end{align}
		
		\end{widetext}
		
		Throughout the above equations, we used the fact that $\tr{O_0}=0$. Also, Eqs.~\eqref{proof_lem4_res1}, \eqref{proof_lem4_res2}, \eqref{proof_lem4_res3}, \eqref{proof_lem4_res4}, and \eqref{proof_lem4_res5} do not contribute to $\mathcal{C}^{\rm(i,ii,iv)}$, and are thus of no consequence to the ESPOVM analysis, for which we can assume $O_0$ is real and $O_0^{\top}=O_0$. So we can easily note that these can be also upper bounded by $\tr{O_0^2}$. Note that values of other cases can be calculated from the above-known cases, by re-parametrizing or transpose-conjugating both sides for example. After realizing that other cases are $0$ or upper bounded by $\tr{O_0^2}$, based on simple trigonometric inequalities, the remaining equations of Lem.~\ref{method:lem4} other than Eq.~\eqref{method:lem4_1} may also be proven. 
		
		\subsection{Proof of Lemma~\ref{method:lem5}}\label{sec:thm_average}
		
		If the averaged variance $\overline{{\rm Var}^{{\rm (r)eq}}(\widehat{\sigma}_{\dyadic{A},\bm{p'}})}^{\rho,\sigma}$ over uniformly chosen input and target state is proved to be $1$ ($1/2$ resp.) for (R)ESPOVM when $2^n\gg1$, by the result of Ref.~\cite{lerasle2019lecture,Blair:1985problem,Jerrum:1986random,huang2020}, Thm.~\ref{main:thm1} (ii),(iii) will be proved. 
		
		Moreover, we recall that the mean squared error (MSE) is defined by the sample mean value of the squares of the additive error between estimation and target value. If an estimation variance is ${\rm Var}(\hat{\sigma})$, then for a large sampling-copy number $N$, the averaged MSE $\overline{\epsilon^2}$ over many experiments of different input and target state becomes $\frac{2\overline{{\rm Var}^{{\rm (r)eq}}(\widehat{\sigma}_{\dyadic{A},\bm{p'}})}^{\rho,\sigma}}{N}$ (Note that sampling trial-number is $N/2$ that is why we have the factor 2). Therefore, the asymptotic lines of Fig.~\ref{fig:fig2} (a--c) are explained.
		
		Now, let us assume that $2^n\gg1$. We take the RESPOVM first. Recall the Eq.~\eqref{eq:sub_real_equatorial_tomography}.
		\begin{align}
			&\sum_{\ket{\phi^{{\rm req}}_\dyadic{A}}\in \mathcal{S}^{{\rm req}}_n}\frac{2^{2n-1}}{\left|\mathcal{S}^{{\rm req}}_n\right|}\bra{\phi^{{\rm req}}_{\dyadic{A}}}\rho\ket{\phi^{{\rm req}}_{\dyadic{A}}}\bra{\phi^{{\rm req}}_{\dyadic{A}}}O\ket{\phi^{{\rm req}}_{\dyadic{A}}}\nonumber\\
			=&\sum_{\bm{x}\in\bfZ^n_2}\bra{\bm{x}}\rho\ket{\bm{x}}\bra{\bm{x}}O\ket{\bm{x}}+\frac{1}{2}{\rm tr}(O)+{\rm tr}(O\rho).
		\end{align}
		Then we obtain that 
		\begin{align}
			\mathbb{E}(\hat{\sigma}^{{\rm req}}_{\dyadic{A}})^2=&\,\left({\rm tr}(\rho \sigma)+\frac{1}{2}{\rm tr}(\sigma)-\sum_{\bm{x}\in \bfZ^N_2}\braket{\bm{x}|\rho|\bm{x}}\braket{\bm{x}|\sigma|\bm{x}}\right)^2
			\nonumber\\
			=&\,\frac{1}{4}+{\rm tr}((\sigma\otimes\sigma)(\rho\otimes\rho))+{\rm tr}(\sigma \rho)\nonumber\\
			&+\sum_{\bm{x},\bm{y}\in\bfZ^n_2}{\rm tr}\left((\rho^{\otimes 2}\sigma^{\otimes 2})(\ket{\bm{x}\bm{y}}\bra{\bm{y}\bm{x}}^{\otimes 2})\right)\nonumber\\
			&\,-2{\rm tr}(\sigma\rho)\sum_{\bm{x}\in\bfZ^n_2}{\rm tr}\left((\rho\otimes \sigma)(\ket{\bm{x}\bm{x}}\bra{\bm{x}\bm{x}})\right).
		\end{align}
		
		Now, we use the following famous results: 
		\begin{lemma}[\cite{Pucha2017}]
			\label{supp:lem_design}
			For a general N-qubit pure state $\rho=\ket{\phi}\bra{\phi}$, the following identity holds:
			\begin{align}
				\int(d\phi)\ket{\phi}\bra{\phi}=&\,\dfrac{I}{2^n}\,,\label{eq:avg_cmplx1}\\
				\int(d\phi)\ket{\phi}\bra{\phi}^{\otimes 2}=&\,\dfrac{1}{2^n(2^n+1)}\left(I+\tau\right)\,,\label{eq:avg_cmplx2}
			\end{align}
			where $I$ is identity operator and $\tau$ is the swap operator between two tensor product parties. The volume measure $(d\phi)$ refers to the uniform distribution (Haar-measure) of all single-qubit pure states.
		\end{lemma}
		From this, when we averaged over all pure input states $\rho$, all terms except $\frac{1}{4}$ are at most the order of $\frac{1}{2^n}$ for all~$\sigma$. Hence, when $2^n\gg1$, 
		\begin{align}\label{sup:thm2_req4}
			\overline{\mathbb{E}(\hat{\sigma}^{\rm req}_\dyadic{A})^2}^{\rho,\sigma}\simeq \frac{1}{4}.
		\end{align}
		
		Next, 
	
		\begin{align}
			&\mathbb{E}((\hat{\sigma}^{\rm req}_{\dyadic{A}})^2)\nonumber\\&=2^{3n-2}\,\mathrm{tr}\left(\eqmomentop^{\otimes 3}\rho\otimes\sigma^{\otimes 2}\right)\\
			&\rightarrow \overline{\mathbb{E}((\hat{\sigma}^{\rm req}_{\dyadic{A}})^2)}^{\rho}\nonumber\\&=2^{2n-2}\,\mathrm{tr}\left(\eqmomentop^{\otimes 2}\sigma^{\otimes 2}\right)\nonumber\\&(\because {\rm Lem.~}\ref{supp:lem_design}).
		\end{align}
		
		Now, since $\sigma$ is pure and real, let us rewrite $\sigma=\ket{\phi}\bra{\phi}$ where $\ket{\phi}=\sum_{\bm{i}\in\bfZ^n_2}c_{\bm{i}}\ket{\bm{i}}$ and~$c_{\bm{i}}\in\mathbb{R}$.
		Then $\ket{\phi}\bra{\phi}^{\otimes 2}=\sum_{\bm{i},\bm{j},\bm{k},\bm{l}\in\bfZ^n_2}c_{\bm{i}}c_{\bm{j}}c_{\bm{k}}c_{\bm{l}}\ket{\bm{i}\bm{j}}\bra{\bm{k}\bm{l}}$. Since $\bm{c}$ lies on hypersphere $S^{2^n-1}$, we obtain that from Lem.~1,
		\begin{align}
			&\overline{\mathbb{E}((\hat{\sigma}^{\rm req}_{\dyadic{A}})^2)}^{\rho,\sigma}\nonumber\\
			=&\,\frac{1}{4A}\int_{S^{2^n-1}}d^{2^n-1}\Omega\,\sum_{\bm{i},\bm{j},\bm{k},\bm{l}}c_{\bm{i}}c_{\bm{j}}c_{\bm{k}}c_{\bm{l}}\ket{\bm{i}\bm{j}}\bra{\bm{k}\bm{l}}\nonumber\\
			&\times\!\!\sum_{\bm{x},\bm{y}\in\bfZ^n_2}\!\!\!\!\left(\ket{\bm{x}\bm{y}}\!\bra{\bm{x}\bm{y}}+\ket{\bm{x}\bm{y}}\!\bra{\bm{y}\bm{x}}+\ket{\bm{x}\bm{x}}\!\bra{\bm{y}\bm{y}}-2\ket{\bm{x}\bm{x}}\!\bra{\bm{x}\bm{x}}\right)\nonumber\\=&\,\frac{1}{4A}\sum_{\bm{x},\bm{y}}\int_{S^{2^n-1}}d^{2^n-1}\Omega\left(3c_{\bm{x}}^2c_{\bm{y}}^2-2c_{\bm{x}}^4\right),
		\end{align}
		where $A$ is the surface area of~$S^{2^n-1}$. Note that the sphere is invariant under permutation hence, 
		\begin{align}
			\label{sup:thm2_req_main}
			\overline{\mathbb{E}((\hat{\rho}^{\rm req}_{\dyadic{A}})^2)}^{\rho,\sigma}=&\,\frac{1}{4A}\Bigg\{2^n\int_{S^{2^n-1}}d^{2^n-1}\Omega\,c_1^4\nonumber\\
			&+3\cdot2^n(2^n-1)\int_{S^{2^n-1}}d^{2^n-1}\Omega\,c_1^2c_2^2\Bigg\}.
		\end{align}
		We will separately calculate the first and second terms of the right-hand side. By noting that $c_1=\cos(\phi_1)$, 
		\begin{align}
			&\int_{S^{2^n-1}}d^{2^n-1}\Omega\,c_1^4\nonumber\\
			=&\oint d\phi_1\,d\phi_2\cdots d\phi_{2^n-1}\cos^4(\phi_1)\nonumber\\
			&\qquad\times\sin^{2^n-2}(\phi_1)\sin^{2^n-3}(\phi_2)\cdots \sin(\phi_{2^n-2})\nonumber\\
			&{\rm (where\;0\le\phi_1,\phi_2\ldots,\phi_{2^n-2}\le\pi,\;0\le\phi_{2^n-1}\le2\pi)}\nonumber\\
			=&\int_{S^{2^n-1}} d^{2^n-1}\Omega\, \frac{\int_0^\pi d\phi_1 \cos^4(\phi_1)\sin^{2^n-2}(\phi_1)}{\int^\pi_0 d\phi_1 \sin^{2^n-2}(\phi_1)}\nonumber\\
			=&\,\frac{A\int_0^{\frac{\pi}{2}}d\phi_1\cos^4(\phi_1)\sin^{2^n-2}(\phi_1)}{\int^{\frac{\pi}{2}}_0 d\phi_1 \sin^{2^n-2}(\phi_1)}.
		\end{align}
		We now use the formula, 
		\begin{align}
			\int^{\frac{\pi}{2}}_0 \sin^m(\phi)d\phi=\begin{cases}
				\frac{\pi}{2}\left(\frac{m!}{2^m\left(\left(\frac{m}{2}!\right)\right)^2}\right) & \mbox{if }m\mbox{ is even,} \\
				\frac{2^{m-1}\left(\left(\frac{m-1}{2}\right)!\right)^2}{m!} & \mbox{if }m\mbox{ is odd,}
			\end{cases}
		\end{align}
		and the beta function $\frac{1}{2}B(p,q)=\frac{\Gamma(p)\Gamma(q)}{2\Gamma(p+q)}=\int^{\frac{\pi}{2}}_0 \sin^{2p-1}(\phi)\cos^{2q-1}(\phi)\;d\phi$, where $\Gamma(p)$ is gamma function. Then we obtain that with Stirling's formula,

		\begin{align}\label{sup:thm2_req1}
			&\,\int_{S^{2^n-1}}d^{2^n-1}\Omega\,c_1^4
			=\,\frac{A}{\pi}\frac{B(\frac{5}{2},2^{n-1}-\frac{1}{2})}{\left(\frac{(2^n-2)!}{2^{(2^n-2)\left((2^{n-1}-1)!\right)^2}}\right)}\nonumber\\
			\simeq&\,\frac{A\sqrt{2\pi e}}{\pi}\Gamma\!\left(\frac{5}{2}\right)\frac{2^{2^n-2}}{(2^{n-1}+1)2^{n-1}}\sqrt{\frac{(2^{n-1}-\frac{3}{2})(2^{n-1}-1)}{2^n-2}}\nonumber\\ &\cdot\frac{(2^{n-1}-\frac{3}{2})^{2^{n-1}-\frac{3}{2}}(2^{n-1}-1)^{2^{n-1}-1}}{(2^n-2)^{2^n-2}}\nonumber\\
			\simeq&\,\frac{A\sqrt{2\pi e}}{\pi}\Gamma\!\left(\frac{5}{2}\right)\nonumber\\ &\cdot\underbrace{2^{2^n-\frac{5n}{2}-1}\,\frac{(2^{n-1}-\frac{3}{2})^{2^{n-1}-\frac{3}{2}}(2^{n-1}-1)^{2^{n-1}-1}}{(2^n-2)^{2^n-2}}}_{\equiv C}.
		\end{align}

		Note that 
		\begin{align}\label{sup:thm2_req2}
			\log_2(C)\simeq&\,\left(2^n-\frac{5n}{2}-1\right)+\left(2^{n-1}-\frac{3}{2}\right)(n-1)\nonumber\\
			&+(2^{n-1}-1)(n-1)-(2^n-2)n\nonumber\\
			=&-3n+\frac{3}{2}\,,
		\end{align}
		giving us
		\begin{equation}
			\frac{1}{A}\int_{S^{2^n-1}}c_1^4\;d^{2^n-1}\Omega\simeq\frac{\sqrt{2\pi e}}{\pi}2^{-3n+\frac{3}{2}}\,\Gamma\!\left(\frac{5}{2}\right).
		\end{equation}
		Let us go to the second term: 
		\begin{align}\label{sup:thm2_req3}
			&\int_{S^{2^n-1}}d^{2^n-1}\Omega\,\bm{c}_1^2\bm{c}_2^2\nonumber\\
			=&\oint d\phi_1\,d\phi_2\cdots d\phi_{2^n-1} \cos^2(\phi_1)\cos^2(\phi_2)\nonumber\\
			&\,\qquad\times\sin^2(\phi_2)\sin^{2^n-2}(\phi_1)\sin^{2^n-3}(\phi_2)\cdots \sin(\phi_{2^n-2})\;\nonumber\\
			&{\rm (where\;0\le\phi_1,\phi_2\ldots,\phi_{2^n-2}\le\pi,\;0\le\phi_{2^n-1}\le2\pi)}\,.
		\end{align}
		This leads to
		\begin{align}
			&\int d^{2^n-1}\Omega\,\frac{\int_0^\pi \cos^2(\phi_1)\sin^{2^n-2}(\phi_1)d\phi_1}{\int^\pi_0 \sin^{2^n-2}(\phi_1)d\phi_1}\nonumber\\
			&\qquad\qquad\times\frac{\int_0^\pi \cos^2(\phi_2)\sin^{2^n-1}(\phi_2)d\phi_2}{\int^\pi_0 \sin^{2^n-3}(\phi_2)d\phi_2}\nonumber\\
			=&\,\frac{A}{2\pi}\frac{B(\frac{3}{2},2^{n-1}-\frac{1}{2})B(\frac{3}{2},2^{n-1})}{\left(\frac{(2^n-2)!}{2^{2^n-2}((2^{n-1}-1)!)^2}\right)\left(\frac{2^{2^n-4}((2^{n-1}-2)!)^2}{(2^n-3)!}\right)}\nonumber\\
			=&\,\frac{2A}{\pi}\frac{\Gamma\!\left(\frac{3}{2}\right)^2(2^{n-1}-\frac{3}{2})!(2^{n-1}-1)!(2^{n-1}-1)^2}{(2^n-2)(2^{n-1})!(2^{n-1}+\frac{1}{2})!}\nonumber\\
			\simeq&\,\frac{2A}{\pi}\Gamma\!\left(\frac{3}{2}\right)^2\frac{(2^{n-1})^2}{(2^n-2)2^{n-1}(2^{n-1}+\frac{1}{2})(2^{n-1}-\frac{1}{2})}\nonumber\\
			\simeq&\, 2^{-2n}\,A\qquad\left[\because \Gamma\!\left(\frac{3}{2}\right)=\frac{\sqrt{\pi}}{2}\right].
		\end{align}
		Therefore, substituting the Eq.~\eqref{sup:thm2_req1}, \eqref{sup:thm2_req2}, \eqref{sup:thm2_req4}, and \eqref{sup:thm2_req3} into Eq.~\eqref{sup:thm2_req_main} leads to 
		\begin{align}
			\overline{{\rm Var}(\hat{\sigma}^{{\rm req}}_{\dyadic{A}})}^{\rho,\sigma}=\overline{\mathbb{E}((\hat{\sigma}^{\rm req}_{\dyadic{A}})^2)}^{\rho,\sigma}-\overline{\mathbb{E}(\hat{\sigma}^{\rm req}_{\dyadic{A}})^2}^{\rho,\sigma}\simeq\frac{1}{2}.
		\end{align}
		Next, let us direct our attention to ESPOVM. As we proceed in a similar manner, we may skip most of the steps but note that 
		\begin{align}\label{sup:thm2_sameway}
			&\sum_{\ket{\phi^{{\rm req}}_{\dyadic{A}}}\in \mathcal{S}^{{\rm eq}}_N}\frac{2^{2n}}{\left|\mathcal{S}^{{\rm req}}_n\right|}\bra{\phi^{{\rm req}}_{\dyadic{A}}}\rho\ket{\phi^{{\rm req}}_{\dyadic{A}}}\bra{\phi^{{\rm req}}_{\dyadic{A}}}O\ket{\phi^{{\rm req}}_{\dyadic{A}}}\nonumber\\
			=&\,\sum_{\bm{x}\in\bfZ^N_2}\bra{\bm{x}}\rho\ket{\bm{x}}\bra{\bm{x}}O\ket{\bm{x}}+{\rm tr}(O)+{\rm tr}(O\rho)
		\end{align}
		and 
		\begin{align}
			\overline{\mathbb{E}((\hat{\sigma}^{\rm eq}_{\dyadic{A}})^2)}^{\rho,\sigma}=&\,\frac{1}{A}\int_{\bm{c}\in S^{2^n-1}}d^{2^n-1}\Omega\sum_{i,j,k,l}c_{\bm{i}}c_{\bm{j}}c_{\bm{k}}c_{\bm{l}}\ket{\bm{i}\bm{j}}\bra{\bm{k}\bm{l}}\nonumber\\
			&\sum_{\bm{\bm{x}},\bm{y}\in\bfZ^n_2}\left(\ket{\bm{x}\bm{y}}\bra{\bm{x}\bm{y}}+\ket{\bm{x}\bm{y}}\bra{\bm{y}\bm{x}}-\ket{\bm{x}\bm{x}}\bra{\bm{x}\bm{x}}\right)\nonumber\\
			=&\,\frac{1}{A}\sum_{\bm{x},\bm{y}}\int_{\bm{c}\in S^{2^n-1}}d^{2^n-1}\Omega\left(2c_{\bm{x}}^2c_{\bm{y}}^2-c_{\bm{x}}^4\right).
		\end{align}
		By the same method with RESPOVM, we can find that $\overline{\mathbb{E}((\hat{\sigma}^{\rm eq}_{\dyadic{A}})^2)}^{\rho,\sigma}\simeq2$ and $\overline{\mathbb{E}(\hat{\sigma}^{\rm eq}_{\dyadic{A}})^2}^{\rho,\sigma}\simeq1$. These result in $\overline{{\rm Var}(\hat{O}^{{\rm eq}}_{\dyadic{A}})}^{\rho,\sigma}=\overline{\mathbb{E}((\hat{\sigma}^{\rm eq}_{\dyadic{A}})^2)}^{\rho,\sigma}-\overline{\mathbb{E}(\hat{\sigma}^{\rm eq}_{\dyadic{A}})^2}^{\rho,\sigma}\simeq 1$.
		Lastly, for both ESPOVM and RESPOVM, 
		\begin{align}
			\mathbb{E}((\hat{\sigma}^{{\rm bin}})^2)=&\,\sum_{\bm{x}}\bra{\bm{x}}\rho\ket{\bm{x}}\left|\bra{\bm{x}}\sigma\ket{\bm{x}}\right|^2\nonumber\\
			=&\,\sum_{\bm{x}}\tr{(\rho\otimes\sigma\otimes\sigma)\ket{\bm{x}\bm{x}\bm{x}}\bra{\bm{x}\bm{x}\bm{x}}}\\
			\rightarrow \overline{\mathbb{E}((\hat{\sigma}^{{\rm bin}})^2)}^{\rho}=&\,\frac{1}{2^n}\sum_{\bm{x}}(\bra{\bm{x}}\sigma\ket{\bm{x}})^2\simeq 0.
		\end{align}
		Hence, even if we average over real pure states $\sigma$, the result will be nearly $0$. Lastly, let us consider randomized-Clifford tomography. Then, we can refer to Eq.~(S43) of Ref.~\cite{huang2020} stating that the estimation variance is equal to the \emph{shadow norm} without maximizing over $\rho$ and it is,
		\begin{align}
			{\rm Var}(\hat{\sigma}^{{\rm s}})=&\,\frac{2^n+1}{2^n+2}\Bigg[\tr{\rho}\tr{\left(\sigma-\frac{I}{2^n}\right)^2}\nonumber\\
			&+2\tr{\rho\left(\sigma-\frac{I}{2^n}\right)^2}\Bigg]\!-\!\left[\tr{\rho\left(\sigma-\frac{I}{2^n}\right) }\right]^2.
		\end{align}
		After we average over pure $\rho$, by Lem.~\ref{supp:lem_design}, we the second and last terms are ignored and 
		\begin{align}\label{sup:thm2_sameway2}
			\overline{{\rm Var}(\hat{\sigma}^{{\rm s}})}^{\rho,\sigma}\simeq \overline{\tr{\sigma^2-\frac{2\sigma}{2^n}+\frac{I}{4^n}}}^{\sigma}\simeq 1\,,
		\end{align}
		which settles the proof for Lem.~\ref{method:lem5}~(i). Note that RESPOVM gives half the average variance than even randomized-Clifford. Since we need two input states for a single sampling trial, the number of states is not improved but we can reduce the total gate count for the whole estimation process. This is because in RESPOVM, half of the input-state copies are used for direct $Z$-basis measurement.
		
		The proof of Lem.~\ref{method:lem5}~(ii) is much simpler. We start with 
		\begin{align}
			\overline{\mathbb{E}((\hat{\sigma}^{\rm eq}_{\dyadic{A}})^2)}^{\rho}&=2^{2n}\mathrm{tr}\left(\momentop^{\otimes 2}\sigma^{\otimes 2}\right).
		\end{align}
		By averaging over Haar-random pure state $\sigma$, by Lem.~\ref{supp:lem_design} and using the fact that $\tr{(\rho\otimes \sigma)\tau}=\tr{\rho\sigma}$, we conclude that  $\overline{\mathbb{E}((\hat{\sigma}^{\rm eq}_{\dyadic{A}})^2)}^{\rho,\sigma}=2$. Also, $\overline{\mathbb{E}(\hat{\sigma}^{\rm eq}_{\dyadic{A}})^2}^{\rho,\sigma}\simeq1$, in exactly the same way as with Eq.~\eqref{sup:thm2_sameway}. Hence $\overline{{\rm Var}(\hat{\sigma}^{{\rm eq}}_{\dyadic{A}})}^{\rho,\sigma}=\overline{\mathbb{E}((\hat{\sigma}^{\rm eq}_{\dyadic{A}})^2)}^{\rho,\sigma}-\overline{\mathbb{E}(\hat{\sigma}^{\rm eq}_{\dyadic{A}})^2}^{\rho,\sigma}\simeq 1$. Furthermore, similar to Eq.~\eqref{sup:thm2_sameway2}, the average variance of random Clifford tomography is also $1$. The proof is thus completed.

		\section{Informational completeness of ESPOVM+computational basis}
		
		In this section, we prove that the ESPOVM+computational basis is informationally complete~(IC) so that this POVM can fully extract all density-matrix information of any given input state. Explicitly, the condition for informational completeness of the POVM $\left\{\Pi_j\right\}_{j\in J}$, where $J$ is an index set and $\Pi_j>0$ for all $j\in J$, necessitates that any quantum state be expressible as a linear combination of the elements $\Pi_j$. 
		
		We now prove the following lemmas:
		
		\begin{lemma}\label{lem:lem_IC_V}
			For a $d$-dimensional Hilbert space, a POVM, comprising a set of $M\geq d^2$ positive-operator elements ${\Pi_j}$ that sum to the identity, is IC if the matrix
			\begin{equation}
				\dyadic{V}=\left(\sket{\Pi_1}\,\,\sket{\Pi_2}\,\,\cdots\,\,\sket{\Pi_M}\right)
			\end{equation}
			has rank~$d^2$, and only then. Here $\sket{\Pi_j}$ is a basis-dependent superket~(vectorization) of $\Pi_j$.
		\end{lemma}
		
		\begin{proof}
			Since $\dyadic{V}$ has rank $d^2=\mathrm{dim}(\mathbb{C}^{d}\otimes \mathbb{C}^{d})$ if and only if it spans a set of vectorized quantum states. By reversing the vectorization (which is an injective map), we note that all quantum state can be expressed as a linear combination of $\left\{\Pi_j\right\}$.
		\end{proof}
		\begin{lemma}
			\label{lem:lem_IC_F}
			A POVM $\{\Pi_j\}$ is IC if the frame operator \cite{scott2006}	
			\begin{equation}
				\mathcal{F}\equiv\sum^M_{j=1}\dfrac{\sket{\Pi_j}\sbra{\Pi_j}}{\tr{\Pi_j}}
				\label{eq:frame_op}
			\end{equation}
			is invertible, and only then. 
		\end{lemma}
		\begin{proof}
			Suppose we have the IC POVM. Clearly from Lem.~\ref{lem:lem_IC_V}, $\rank\{\dyadic{V}\}=d^2$. Note  that $\mathcal{F}\,\widehat{=}\,\dyadic{V}\bm{\Pi}^{-1}\dyadic{V}^\dag$, where
			\begin{equation}
				\bm{\Pi}=\begin{pmatrix}
					\tr{\Pi_1} & 0 & \cdots & 0\\
					0 & \tr{\Pi_2} & \cdots & 0\\
					\vdots & \vdots &\ddots & \vdots\\
					0 & 0 & \cdots & \tr{\Pi_M}
				\end{pmatrix}\,.
			\end{equation}
			Also, if $\bm{x}\in {\rm Ker}(\dyadic{V})$, it is easy to see that $\bm{x}\in {\rm Ker}(\mathcal{F})$, or $\bm{x}^{\top}\mathcal{F}\bm{x}=\bm{x}^{\top}\dyadic{V}^{\dag}\bm{\Pi}^{-1}\dyadic{V}\bm{x}=0\rightarrow \|\dyadic{V}\bm{x}\|^2=0$ (the inner product is defined by the Euclidean product with the weights $1/\tr{\Pi_j}>0\;(j\in[M])$). Hence $\dyadic{V}\bm{x}=0\rightarrow \bm{x}\in{\rm Ker}(\dyadic{V})$. So, $\rank(\mathcal{F})=M-{\rm dim}({\rm Ker}(\mathcal{F}))=M-{\rm dim}({\rm Ker}(\dyadic{V}))=\rank(\dyadic{V})$. 
			Since $\rank(\dyadic{V})=d^2$,  $\rank(\mathcal{F})=d^2$. As $\mathcal{F}$ is also a $d^2\times d^2$ matrix, $\mathcal{F}^{-1}$ exists. 
			
			Conversely, suppose that $\mathcal{F}^{-1}$ exists. Then $\rank\{\mathcal{F}\}=d^2$. Since $\rank\{\mathcal{F}\}\le \rank\{\dyadic{V}\}\\\le d^2$, we conclude that $\rank\{\dyadic{V}\}=d^2$. By Lem.~\ref{lem:lem_IC_V}, $\{\Pi_j\}_{j\in J}$ is IC. 
		\end{proof}
		Setting $d=2^n$ therefore yields the desirable result.
		\begin{proposition}\label{thm:thm_IC_ES_comp}
			The collection of ESPOVM and computational basis forms an IC measurement. However, ESPOVM itself is not an IC~POVM. 
		\end{proposition}
		
		\begin{proof}
			First, we calculate the frame operator of the ESPOVM. 
			\begin{align}
				\mathcal{F}_{\rm es}=&\,\sum_{\ket{\phi^{{\rm eq}}_\dyadic{A}}\in\mathcal{S}^{{\rm eq}}_n}\frac{2^n}{\left|\mathcal{S}^{{\rm eq}}_n\right|}\ket{\phi^{{\rm eq}}_\dyadic{A}}(\bra{\phi^{{\rm eq}}_\dyadic{A}})^{\top}\bra{\phi^{{\rm eq}}_\dyadic{A}}(\bra{\phi^{{\rm eq}}_\dyadic{A}})^{*}\nonumber\\
				=&\,2^n\left(\momentop^{\otimes 2}\right)^{\top_2},
			\end{align}
			where $\top_2$ means partial transpose for the second part. From Lem.~1, we note that 
			\begin{align}
				\mathcal{F}_\mathrm{es}=&\,\frac{1}{2^n}\left(1+\sum_{\bm{x},\bm{y}\in\mathbb{Z}_2^n}\ket{\bm{x}\bm{x}}\bra{\bm{y}\bm{y}}-\sum_{\bm{x}\in\mathbb{Z}_2^n}\ket{\bm{x}\bm{x}}\bra{\bm{x}\bm{x}}\right)\nonumber\\
				=&\,\frac{1}{2^n}\left(1+\sum_{\bm{x}\in\mathbb{Z}_2^n}\sum_{\bm{y}\neq \bm{x}\in\mathbb{Z}_2^n}\ket{\bm{x}\bm{x}}\bra{\bm{y}\bm{y}}\right)\,.
			\end{align}
			One can show that the operator $W=\sum_{\bm{x}\in\mathbb{Z}_2^n}\sum_{\bm{y}\neq \bm{x}\in\mathbb{Z}_2^n}\ket{\bm{x}\bm{x}}\bra{\bm{y}\bm{y}}$ has at least one eigenvalue equal to $-1$ by inspecting its action on $\ket{\bm{z}\bm{z}}$ for any $\bm{z}\in\mathbb{Z}_2^n$,
			\begin{equation}
				W\ket{\bm{z}\bm{z}}=\sum_{\bm{x}\neq \bm{z}\in\mathbb{Z}_2^n}\ket{\bm{x}\bm{x}}=\sum_{\bm{x}\in\mathbb{Z}_2^n}\ket{\bm{x}\bm{x}}-\ket{\bm{z}\bm{z}}\,.
			\end{equation}
			It follows that any ket $(\ket{z_1,z_1}-\ket{z_2,z_2})/\sqrt{2}$ with $z_1,z_2\in\mathbb{Z}_2^n$ is an eigenket in the ``$-1$''-eigenvalue subspace.
			As a side note, it is easy to check that the superposition $\ket{\phi}\equiv\sum_{\bm{z}\in\mathbb{Z}_2^n}\ket{z,z}2^{-n/2}$ is an eigenket that gives a positive eigenvalue of $2^n-1$,
			\begin{align}
				W\sum_{\bm{z}\in\mathbb{Z}_2^n}\ket{\bm{z}\bm{z}}2^{-n/2}=&\sum_{\bm{x}\in\mathbb{Z}_2^n}\ket{\bm{x}\bm{x}}\sum_{\bm{z}\neq \bm{x}\in\mathbb{Z}_2^n}2^{-n/2}\nonumber\\
				=&\sum_{\bm{x}\in\mathbb{Z}_2^n}\ket{\bm{x}\bm{x}}2^{-n/2}(2^n-1)\,.
			\end{align}
			We may then conclude that $\mathcal{F}_\mathrm{es}$ has at least one null eigenvalue. Based on Lem.~\ref{lem:lem_IC_F}, it follows that the ESPOVM is non-IC.
			
			Upon recognizing that  $\sum_{\bm{x}\in\mathbb{Z}_2^n}\ket{\bm{x}\bm{x}}\bra{\bm{x}\bm{x}}$ is the frame operator $\mathcal{F}_{{\rm comp}}$ for the computational basis, the total frame operator may be written as 
			\begin{align}
				\mathcal{F}=&\,\,\mathcal{F}_\mathrm{es}+\mathcal{F}_\mathrm{comp}\nonumber\\
				=&\,\frac{1}{2^n}+\ket{\phi}\bra{\phi}+\left(1-\frac{1}{2^n}\right)\sum_{\bm{x}\in\bfZ^n_2}\ket{\bm{x}\bm{x}}\bra{\bm{x}\bm{x}}> 0.
			\end{align} 
			Therefore, $\mathcal{F}$ has full rank (with all eigenvalues positive) and is invertible. By Lem.~\ref{lem:lem_IC_F}, \mbox{ESPOVM+computational} basis is IC~POVM.
		\end{proof}
		
		\section{Linear nearest-neighboring (LNN) architecture of CZ~circuits~\cite{maslov2018}}
		
		In this section, we explain how to transform the arbitrary CZ~circuit to the $2n+2$-depth LNN form, which has only neighboring CNOT~gates. We encourage readers who are interested in Clifford circuit synthesis to read Ref.~\cite{maslov2018}, from which some techniques were borrowed. This section briefly discusses some relevant background content in~Ref.~\cite{maslov2018}.
		
		The actual circuit we transform is the so-called $\widehat{\mathrm{CZ}}$~circuit. This is defined by CZ~circuit followed by \emph{qubit-reversal} layers, $\textsc{rev}:\, \ket{x_1,x_2,\ldots,x_n}\mapsto\ket{x_n,x_{n-1},\ldots,x_1}$. However, the CZ~section is equal to \textsc{cz}-\textsc{rev}-\textsc{rev} and the last $\textsc{rev}$ can be ignored by reversing the measurement outcome if we only do the non-adaptive measurement. Hence, if $\widehat{\mathrm{CZ}}$ can be transformed to a ($2n+2$)-depth LNN architecture, the proof is~complete. 
		
		Now, consider some arbitrary $\widehat{\mathrm{CZ}}$~circuit. We will prove this only when $n$ is odd, since that for even $n$ follows similarly. Suppose we take $\ket{x_1,x_2,\ldots,x_n}$ as the input and define the two unitaries 
		\begin{widetext}
		\begin{align}
			&C_1\equiv \mathrm{CNOT}_{1\rightarrow 2}\mathrm{CNOT}_{3\rightarrow 4}\ldots \mathrm{CNOT}_{{n-2}\rightarrow {n-1}}\mathrm{CNOT}_{3\rightarrow 2}\mathrm{CNOT}_{5\rightarrow 4}\ldots \mathrm{CNOT}_{{n}\rightarrow {n-1}}\,,\\
			&C_2\equiv \mathrm{CNOT}_{2\rightarrow 1}\mathrm{CNOT}_{4\rightarrow 3}\ldots \mathrm{CNOT}_{{n-1}\rightarrow {n-2}}\mathrm{CNOT}_{2\rightarrow 3}\mathrm{CNOT}_{4\rightarrow 5}\ldots \mathrm{CNOT}_{{n-1}\rightarrow {n}}
		\end{align}
		\end{widetext}
		and $C\equiv C_1C_2$. Now, let $n=2m+1$. If we operate $C$ on the input $t\le m$~times, the output state becomes
		\begin{align}
			&\,C^{t}\ket{x_1,x_2,\ldots,x_n}\nonumber\\
			=&\,\bigotimes_{i=1}^{n}\ket{[Pj(n-3-2(t-1)+i),Pk(2(t-1)+i)]}\,,
		\end{align}
		where $[j,k](j\ne k)\equiv x_j\oplus x_{j+1}\oplus\cdots\oplus x_{k}$ ($\oplus$ is the ${\rm mod} \;2$ summation), $[j,j]\equiv x_j$ and the \emph{patterns},

		\begin{align}
			&Pj\equiv(n-1,n-3,n-3,\ldots,4,4,2,2,1,1,3,3,\ldots\nonumber\\&n-2,n-2)\,,\\
			&Pk\equiv(3,3,5,5,\ldots,n,n,n-1,n-1,n-3,n-3,\ldots,\nonumber\\&6,6,4,4,2).
		\end{align}
		
		Note that both are of length $2n-3$. Also, $P(j \;{\rm or} \;k)(a)$ where $a\in [2n-3]$ means the $a$-th element of such patterns. We note that given $t\le m$ and $i\in [n]$, there is no case when $[Pj(n-3-2(t-1)+i),Pk(2(t-1)+i)]$ is equal with $[Pj(n-3-2(t'-1)+i'),Pk(2(t'-1)+i')]$ for the other $t'\le m $ and $i'\in[n]$. This is schematically seen in Ref.~\cite{maslov2018}. Also, the number of tuples obtained during $t=0,1,\ldots,m$ is $(m+1)n=\frac{n(n+1)}{2}=\binom{n}{2}+n$. These mean that after we recorded all elements of output ket at each step we operate $S$, for $m$~times, we can see all $[j,k]\;(j,k\in[n])$. Note that $[j,j]$ can be seen at the starting point ($t=0$).
		
		Now, we show that $C^{m+1}$ is a qubit-reversal operation. After the $m$-th operation, we operate $C$ once more and calculate the $p$-th qubit output of the output state for an even $p$. By the  neighboring property, we just need to sum the $\left\{p-1,p,p+1\right\}$-th inputs, resulting in $[n-(2\lfloor\frac{p-1}{2}\rfloor+1),n-(2\lfloor\frac{p-2}{2}\rfloor-1)]\oplus[n-(2\lfloor\frac{p}{2}\rfloor+1),n-(2\lfloor\frac{p-1}{2}\rfloor+1)]\oplus[n-(2\lfloor\frac{p+1}{2}\rfloor+1),n-(2\lfloor\frac{p}{2}\rfloor+1)]=n-p+1$, which is a reversal of~$p$. We can also note that the result for an odd $p$ is also a reversal. Similarly, we conclude that $S^{m+1}$ is a qubit~reversal. 
		
		Next, we recall Eq.~\eqref{eq:syseqn}, that is $y_i=x_1\oplus x_2\oplus\ldots\oplus x_i$, in the main text. In this context, $y_i=[1,i]$ and $y_j\oplus y_k=[j+1,k]$. Hence, one can achieve $i^{u_jy_j}$ and $i^{u_{j,k}(y_j\oplus y_k)}$ ($u_j,u_{j,k}\in\bfZ_4$) by inserting phase gates in each section of the $C$-operated output, or at the starting point. Also, the depth of the LNN architecture is $4(m+1)=2n+2$. The 1-to-1 mapping $(t,i)\leftrightarrow [Pj(n-3-2(t-1)+i),Pk(2(t-1)+i)]$ must be recorded in $\mathcal{O}(n^2)$-memory before the circuit sampling so that we do not need to be bothered with the time to find the right location for phase-gate insertion. 
		A similar consideration applies when $n$ is even but with slightly different $C_1, C_2$, and patterns~\cite{maslov2018}. In conclusion, we can realize the $\widehat{\mathrm{CZ}}$~operation with only nearest-neighboring CNOT and phase gates of depth $2n+2$. Lastly, we note that $Y$-basis measurements can be carried with $X$-basis measurements after phase gates. If these phase operations are performed before the ${\rm CZ}$ action to transform the ${\rm CZ}$ circuit into an LNN architecture, the double-layered rightmost CNOT~gates can simply be regarded as classical post-processing of the $X$-basis outcome of the equatorial stabilizer measurement. Therefore, only $2n$-depth circuits are needed.

\bibliographystyle{apsrev4-2-titles}
\bibliography{rev-bibliography}

\end{document}